\documentclass[11pt,a4paper]{article}

\usepackage{jheppub}

\usepackage{bm}
\usepackage{amsthm}
\usepackage{mathtools}

\usepackage{xcolor}
\usepackage{booktabs}
\usepackage{enumitem}

\usepackage{orcidlink}
\usetikzlibrary{decorations.markings}

\usepackage{csquotes}

\apptocmd{\thebibliography}{%
  \setlength{\itemsep}{4pt plus 0.2ex minus 0.05ex}%
}{}{}

\theoremstyle{definition}
\newtheorem{theorem}{Theorem}[section]
\newtheorem{proposition}[theorem]{Proposition}

\newcommand{\diag}{\operatorname{diag}}
\newcommand{\order}[1]{\mathcal{O}(#1)}

\newcommand{\TT}{\mathrm{TT}}

\newcommand{\iu}{\mathrm{i}\mkern1mu}
\newcommand{\dd}{\mathrm{d}}
\newcommand{\cB}{\mathcal{B}}
\newcommand{\cH}{\mathcal{H}}
\newcommand{\cO}{\mathcal{O}}

\newcommand{\cX}{\mathcal{X}}

\newcommand{\Ker}{\operatorname{Ker}}
\newcommand{\spec}{\operatorname{spec}}
\newcommand{\STr}{\operatorname{STr}}

\title{A Spectral Criterion for Emergent Gravity}

\author[a]{Jinku Guo\,\orcidlink{0009-0000-6600-6171}}
\affiliation[a]{Northwestern Polytechnical University, Xi'an 710072, China}
\emailAdd{guojk@nwpu.edu.cn}

\abstract{%
We give a source-resolved spectral criterion for long-range tensor
exchange between conserved matter sources.
Exact channel elimination preserves the propagation kernel together
with the source and readout maps and the direct response.
Common-space estimates control the residue at the displaced critical
root of a finite approximation.
A positive K\"all\'en--Lehmann representation identifies this residue
with a zero-invariant-mass atom even at a continuum endpoint.
For a complete positive Lorentz-covariant conserved stress tensor
we classify the null-shell measure and prove that its TT atomic weight
vanishes.
A reduced compact spectral realisation instead carries internal
tensor vertices into an external two-helicity response.
Conservation reduces the leading vertex of localised stationary matter
to its total energy.
The mixed residue then factorises into the two energies and a common
propagation coefficient.
Spatial resolvent estimates determine the potential and force with
finite-size corrections.
Normalisation through the matter vertex leaves one dimensional
calibration for the overall Planck scale.
The spectral residue fixes the dimensionless exchange strength.
The resulting response admits an Einstein--matter effective
description under the stated positivity and helicity assumptions
together with universal soft coupling and local nonlinear consistency.
}


\hypersetup{
  pdftitle={A Spectral Criterion for Emergent Gravity},
  pdfauthor={Jinku Guo},
  pdfsubject={Spectral residues, matter coupling, and long-range tensor response},
  pdfkeywords={Kallen-Lehmann representation; Stress-tensor correlators; Analytic Fredholm theory; Long-range tensor exchange; Emergent gravity}
}

\makeatletter
\gdef\@fpheader{}
\makeatother

\let\paperTableOfContents\tableofcontents
\renewcommand\tableofcontents{\begingroup\small\paperTableOfContents\endgroup}
\renewcommand\afterTocSpace{\vskip2pt}
\renewcommand\afterTocRuleSpace{\clearpage}

\begin{document}

\let\acknowledgmentsOrig\acknowledgments
\renewenvironment{acknowledgments}{\acknowledgmentsOrig}{}

\maketitle\flushbottom

\section{Introduction}
\label{sec:introduction}

A spectral account of gravity begins with the response exchanged between
matter configurations. Internal modes and their products determine which
disturbances a source excites. The low-mass spectrum determines how those
disturbances propagate over large distances. In the generated compact
construction the gauge and matter channels arise from a common source
algebra whose transport extends to a Lorentz event carrier
\cite{Guo2026CompactSpectrumI}. We study how these ingredients enter a
collective tensor response and how matter measures its residue.
The tensor is emergent in the familiar sense of a collective mode such
as a phonon. Its gravitational interpretation follows from its helicity
content and universal energy coupling together with the consistency of
its nonlinear interactions. An effective metric expresses this collective
response at the macroscopic level.

Familiar spectral phenomena illustrate the physical question. A delta
contribution to a photon correlator identifies a stable propagating
channel. A multiparticle threshold instead gives continuous spectral
weight with a different spatial tail. The Casimir effect likewise shows
how a change in the spectrum of quantum fluctuations alters a macroscopic
force~\cite{Bordag2001Casimir}. For tensor exchange the complete
energy--momentum vertex specifies how matter excites the response.
A positive spectral measure records the states visible to that vertex.
The causal correlator describes their propagation.

Internal eigenvalues and scale equations determine the modal content.
A propagation pole concerns the response at external momentum.
Their connection requires one sourced correlator. We use compact spectral
data to determine the internal channel and its matter vertices.
The external response has its own propagating state space.
The complete covariant stress tensor gives a restrictive realisation
whose null-shell sector can be classified directly. In each realisation
the source changes the operator together with its state weights and
observable. Differentiating this common rule determines the propagation
kernel and the maps through which the source enters and the response is
observed. It also fixes the direct response at fixed state.
Exact block elimination preserves their combination.

The mathematical tools are classical. The K\"all\'en--Lehmann
representation identifies invariant-mass atoms. Analytic Fredholm theory
determines the singular part of an operator inverse. Schur--Feshbach
reduction removes complementary channels
\cite{Kallen1952,Lehmann1954,Kato1995,Feshbach1958}.
Here these operations retain the physical source and observable.
A critical direction contributes through the source that excites it and
the readout that detects it. A finite approximation has its own critical
root at which its residue must be evaluated.
This comparison preserves the observed response through elimination and
makes the internal calculation comparable with a positive external
Wightman measure.

The connection to matter uses the complete energy--momentum vertex.
Laue's integrated stress identity~\cite{Giulini2018Laue} reduces the
zero-momentum vertex of a localised stationary configuration to its total
energy. Spatial stresses enter the finite-size remainder.
The same reduced source map sends the leading source of each
configuration to a common critical direction. The mixed residue then
factorises into the two total energies and one propagation coefficient.
The spatial inverse determines the leading exchange kernel.
Separate estimates for that kernel and its gradient control the
potential and force. Constraint contributions and direct terms are
retained before the static limit.

Compact spectral calculations supply endpoint derivatives and a
mass-shift critical point together with explicit tensor vertices on
scalar and one-form modes. A simple internal zero supplies a local
coordinate on the susceptible channel. Matching the zero and first
derivative of its inverse response to the external two-helicity inverse
fixes the leading pole and its normalisation. Higher analytic terms
give regular or derivative corrections. The compact TT channel thus
determines the internal source data. The external helicity space
determines the macroscopic propagation law. Their source and readout
maps connect the two descriptions.

Induced gravity obtains local curvature terms by integrating out matter
on a background geometry~\cite{Sakharov1967,Visser2002}.
The present criterion concerns the nonlocal response selected by the
sourced spectrum. Its residue determines the strength entering matter
exchange. Its tensor structure and interactions determine whether it
admits an Einstein-type local completion. Local counterterms are
separated before this spectral weight is extracted.

Composite tensor dynamics has been studied in matter-only constructions
\cite{Amati1981,Amati1982}. Maitiniyazi and Yamada derived induced kinetic
terms for auxiliary tensor channels in fermionic and scalar composite
models~\cite{Maitiniyazi2026emergence}. Fehre et al.\ computed a positive
Lorentzian graviton spectral function with a massless one-particle peak
and a continuum~\cite{Fehre2023prl}. The composite calculations establish
tensor propagation in specified models. The spectral-function
calculation starts from an identified graviton field. Our question is
how a propagating object selected from a sourced response retains its
physical identification under channel elimination and finite
approximation. The source and readout maps provide this identification.
Positive Wightman matching gives its spectral interpretation.

The null-shell classification is a structural result for the complete
stress tensor. Positivity and exact Lorentz covariance together with
conservation force its null-orbit measure to be longitudinal.
The proof uses the tensor measure directly and therefore applies
without an asymptotic one-particle decomposition.
Its transverse-traceless quotient has zero atomic weight.
The constructive response analysis treats a reduced two-helicity
realisation with its own positive covariance and inherited source maps.
It transports the complete response through channel elimination and
controls its residue at the shifted root of a finite approximation.
The compact realisation supplies explicit matter vertices and an
internal simple zero for this construction.
Massless representation theory gives the related restriction on states
created by a covariant tensor~\cite{Weinberg2020Massless}.
The Weinberg--Witten theorem tests the particle stress-tensor matrix
elements and translation charge of an infrared particle interpretation
\cite{WeinbergWitten1980}.

The gravitational completion uses established consistency results.
Weinberg's soft theorem fixes the universal coupling of a helicity-two
pole under its scattering assumptions~\cite{Weinberg1965}.
Deser's self-coupling construction gives an Einstein completion
\cite{Deser1970}. Consistent local deformations with at most two
derivatives are classified under the hypotheses of
ref.~\cite{Boulanger2000}. The effective action adopts this completion
with its normalisation fixed by the mixed source residue.

The general response theory and compact spectral constructions are
developed in
refs.~\cite{Guo2026GeneralChannels,Guo2026CompactSpectrumI,
Guo2026CompactSpectrumII,Guo2026Arithmetic}.
The finite evaluations use the Coarse-Graining Genesis Framework
\cite{Guo2026CoarseGrainingFramework}.
The technical appendices give the tensor-measure classification and
operator estimates together with the source identities used in the
spectral and exchange calculations.

\section{The spectral system and its source variations}
\label{sec:spectral-datum}

The compact modal datum specifies the spectral projections and modal products together with their weights and source derivatives. Once these variations are transported to the
response space a closed reduced form supplies a positive resolvent. The
covariant QFT realisation then fixes the Wightman measure and its physical
TT quotient and Euclidean continuation produces the scalar form factor
whose endpoint weight is tested. This sequence separates an internal
channel calculation from the external propagation question while keeping
both in one source-resolved construction.

The compact spectral system of refs.~\cite{Guo2026CompactSpectrumI,Guo2026CompactSpectrumII} consists of an internal space with modal products and spectral operators whose heat semigroups define averaging.
It serves as a finite-volume carrier for computing source derivatives.
After the modal data have been transported to a common Hilbert space
external propagation is formulated on the flat event carrier in the
corresponding linear-response realisation. A gap gives a Yukawa tail. Spectral weight reaching the origin gives a continuum tail or a massless atomic contribution according to its measure. Source and readout maps
decide which of these channels is visible to matter.

\subsection{Compact modal data and their sourced variation}
\label{sec:native-spectral-system}

Write $M_L=\mathbb{RP}^3$ with round radius $L$ and volume $V_3(L)=\pi^2L^3$. Its scalar eigenvalues are $\lambda_l=l(l+2)/L^2$ with multiplicity $(l+1)^2$ for even $l$.
The round carrier combines a discrete Laplace spectrum with a closed
modal product algebra so source variations can be computed together with
the modes and weights that carry them.
The vector and spinor carriers are kept separately. A modal label $j$ records the carrier and eigenmode together with the matter multiplicity. The datum includes its spectral projection and multiplication coefficients as well as its eigenvalue. For orthonormal scalar modes for example
\begin{equation}
 \mathfrak m_{ij}^{\,k}
 =\int_{M_L}\overline{u_k}\,u_i u_j\,\dd v_L .
 \label{eq:native-modal-product}
\end{equation}
The products and their spectral projections determine the corresponding
modal couplings.

The positive spectral operator $A_L$ of a specified carrier defines $T_\upsilon=e^{-\upsilon A_L}$ for $\upsilon\geq0$. Its spectral resolution gives
\begin{equation}
 T_{\upsilon_1}T_{\upsilon_2}=T_{\upsilon_1+\upsilon_2},\qquad
 T_\upsilon=\sum_j e^{-\upsilon\lambda_j}P_j .
 \label{eq:native-internal-gaussian}
\end{equation}
These identities hold strongly including on any zero eigenspace.
They give the internal Gaussian composition used in spectral observables.
The source-space Gaussian is characterised independently by the convolution law together with orthogonal invariance and square-root scaling. Multiplication of an already
formed external response by a Gaussian is a different operation.

To differentiate a geometric source first transport the modal spaces
to one Hilbert space. For a homothety this transport includes the
$L^{3/2}$ volume factor in the pullback of scalar wavefunctions.
It gives $A_L=L^{-2}A_1$ on the transported domain.
For a finite smoothly transported spectral subspace $\mathcal I$ denote
its compressed operator by $A_{\mathcal I}(h)$ and its observable weight by
$W_{\mathcal I}(h)$.
A basic spectral functional is
\begin{equation}
 \mathcal Q_{\mathcal I}(h)=\frac1{V(h)}
       \operatorname{Tr}\!\left[W_{\mathcal I}(h)f(A_{\mathcal I}(h))\right].
 \label{eq:native-spectral-readout}
\end{equation}
The weight includes the channel coefficients and Gaussian
factor. The source rule includes variation of the weight and volume together with transport of the retained subspace. One may instead put
the Gaussian in $f$ provided its source dependence is then
differentiated there.

For later source variations suppose $A_{\mathcal I},W_{\mathcal I},V$ are twice continuously differentiable with $V>0$ and suppose $f$ is holomorphic on a neighbourhood of the
relevant spectra. Set $f=f(A_{\mathcal I})$ and $v_a=\partial_a\log V$ with
\[
 F_a=\operatorname{Tr}[W_a f+W Df(A_{\mathcal I})[A_a]].
\]
Twice differentiating
$V\mathcal Q_{\mathcal I}=\operatorname{Tr}[W_{\mathcal I}f(A_{\mathcal I})]$
gives
\begin{align}
 \mathcal Q_a&=V^{-1}F_a-v_a\mathcal Q, \notag\\
 \mathcal Q_{ab}
 &=V^{-1}F_{ab}-v_a\mathcal Q_b-v_b\mathcal Q_a
                       -V^{-1}V_{ab}\mathcal Q,                 \label{eq:native-full-spectral-derivative}\\
 F_{ab}&=\operatorname{Tr}\big[
 W_{ab}f+W_aDf[A_b]+W_bDf[A_a]\notag\\
 &\hspace{29mm}+W D^2f[A_a,A_b]+W Df[A_{ab}]\big].\notag
\end{align}
where the derivatives of $A_{\mathcal I}$ include those of its comparison maps. For
noncommuting variations one uses the standard Dunford formula
\[
 Df(A)[H]=\frac1{2\pi i}\oint_\Gamma
 f(\zeta)(\zeta-A)^{-1}H(\zeta-A)^{-1}\,\dd\zeta .
\]
and the second derivative inserts the two variations in both orders. This
formulation is independent of a choice of eigenvectors at a degenerate
eigenvalue. The contour-calculus justification is standard.

For an infinite modal sum these identities pass to the limit if the
summands and their first two source derivatives have locally uniform
summable majorants. A Gaussian times a polynomial mode bound supplies
such a majorant on compact positive-scale intervals. A moving hard
window is treated piecewise between eigenvalue crossings with boundary
contributions at a change of the retained subspace.

\paragraph{The compact TT spectral branch.}
\label{sec:native-tt-spectral-branch}

The compact closure uses the following finite spectral sum. Let $c>0$ be
the dimensionless endpoint coordinate and let $\mu_*>0$ be the fixed mass
unit used to restore dimensions in this trace. Its comparison radius is
$L_{\rm tr}=c/\mu_*$ and
$k_{\rm win}^2=\pi\mu_*^2/c^2$. The carrier above is evaluated at
$L=L_{\rm tr}$ in the following sum.
At fixed content and spin-structure branch
\begin{equation}
 \lambda_j=\frac{\mu_*^2\widehat\lambda_j}{c^2},\qquad
 m_j^2=\frac{\mu_*^2\widehat m_j^2(\vartheta)}{c^2},\qquad
 K_{\rm TT}(x,b)=\frac{x^2}{(x+b)^2}.
 \label{eq:native-tt-spectral-input}
\end{equation}
Here $\vartheta$ is the torsion parameter of the compact spectral model.
It is distinct from the internal heat parameter $\upsilon$ and the
external averaging parameter $\tau$ used below.
The TT channel weights are $w_j=1$ for real scalar
components and $w_j=4$ for gauge and Weyl components with their
modal and matter multiplicities in $d_j$. Ghost and tensor-carrier
entries have zero weight in this particular matter sum.
The dimensionless shifts on these retained carriers are
\begin{equation}
 \widehat m_{\rm sc}^2=\frac34,\qquad
 \widehat m_{\rm g}^2=\frac{C_{2,\rm rep}}2+\frac{\vartheta^2}6,\qquad
 \widehat m_{\rm f}^2=\frac{3\vartheta^2}8.
 \label{eq:native-tt-mass-shifts}
\end{equation}
Here $C_{2,\rm rep}$ is the quadratic Casimir of the retained gauge
representation. The symbol $C_2(Q^2)$ denotes the TT form factor.
The spinor curvature is already contained in $\widehat\lambda_j$.
The example uses the charged species trace with odd left-handed and even
right-handed spatial branches. The one-generation charge moments give
$\vartheta=1/50$ in its source normalisation and the common weak and colour
screening conditions select three identical copies within the replication
family of ref.~\cite{Guo2026CompactSpectrumII}. These choices specify
the content of this trace.
This matter sum defines the internal spectral closure. The dimensionless
value $\kappa_{\rm TT}=4/27$ fixes its endpoint and internal readout. The matched
sourced response separately determines the external pole and its zero-mass
spectral weight.

For the retained index set $\mathcal J_{n_{\rm sp}}$ and frequency bound $R$ define
\begin{align}
 B_j&=\sqrt{(R\pi-\widehat\lambda_j)_+},\qquad
 X_j(u)=u^2+\widehat\lambda_j,\notag\\
 A_{n_{\rm sp},R}(\vartheta)
 &=\frac1{32\pi^3}\sum_{j\in\mathcal J_{n_{\rm sp}}}d_jw_j
   \int_0^{B_j}
   \frac{X_j(u)^2 e^{-X_j(u)/\pi}}
        {[X_j(u)+\widehat m_j^2(\vartheta)]^2}\,\dd u.       \label{eq:native-tt-coefficient}
\end{align}
The source reads the spectral domain
$\omega^2+\lambda_j<Rk_{\rm win}^2$ with projector
$P_R=\mathbf1_{[0,Rk_{\rm win}^2)}(\omega^2+A_L)$ on the joint frequency
and modal space. Thus $R$ specifies its accepted
bandwidth inside the Gaussian envelope. The integer response-level family
of ref.~\cite{Guo2026CompactSpectrumII} selects $R=5$ on the three-copy
branch by the shell-support and three-position conditions. We enumerate
its modes with $n_{\rm sp}=60$. The eight nonzero rows are given in
Appendix~\ref{sec:app-compact-spectral-branch}. Since
$k_{\rm win}L_{\rm tr}=\sqrt\pi$ variations of $c$ and $\vartheta$ on
this branch leave the retained set fixed. Quadrature approximates the
integrals on this fixed domain. The Gaussian tail bound in
Eq.~\eqref{eq:app-compact-gaussian-tail} instead compares this response
with the one accepting the full Gaussian spectrum.

At the emergence endpoint $c$ equals the dimensionless coordinate
$\varkappa$ of ref.~\cite{Guo2026CompactSpectrumII}. That calculation
uses the trace arguments $(\ell_{\rm end},\widehat\Lambda_{\rm tr})
=(c,\sqrt\pi/c)$. Its physical Kaluza--Klein readout is
\[
 L_{\rm phys}=\frac{c}{M_G},\qquad
 c_{\rm gen}=c^2,\qquad
 \Lambda_{\rm gen}=\frac{c_{\rm gen}}{L_{\rm phys}}=cM_G.
\]
Here $c_{\rm gen}$ selects the retained family positions and $M_G$ is
fixed by the dimensional calibration. The reference unit $\mu_*$ restores
the trace spectrum. $L_{\rm tr}$ and $L_{\rm phys}$ coincide when this
unit is chosen to be $M_G$. This correspondence distinguishes the trace bandwidth from the family-position cutoff and physical radius.

On this fixed compact branch the complete spectral sum obeys the exact
radius and torsion identities
\begin{equation}
  \mathcal T_{n_{\rm sp},R}(c,\vartheta)
  :=\frac{V_3\Pi_{\rm TT}^{n_{\rm sp},R}}{32\pi^2\mu_*}
  =\frac{A_{n_{\rm sp},R}(\vartheta)}{c},\qquad
  c\partial_c\mathcal T_{n_{\rm sp},R}
 =-\mathcal T_{n_{\rm sp},R}.       \label{eq:native-tt-homogeneity}
\end{equation}
The coefficient is positive if a positive-weight mode has a nonempty
frequency interval. Where denominators are strictly positive
\begin{equation}
 A'_{n_{\rm sp},R}(\vartheta)=
 -\frac1{16\pi^3}\sum_jd_jw_j(\widehat m_j^2)'(\vartheta)
 \int_0^{B_j}\frac{X_j(u)^2e^{-X_j(u)/\pi}}
 {[X_j(u)+\widehat m_j^2(\vartheta)]^3}\,\dd u.             \label{eq:native-tt-torsion-derivative}
\end{equation}
Thus $A'_{n_{\rm sp},R}<0$ for $\vartheta>0$ whenever a gauge or fermion
contribution is retained with positive integration length. For the first identity substitute $u=c\omega/\mu_*$. The frequency measure and $V_3$ leave the factor $\mu_*/c$ and the normalisation in
Eq.~\eqref{eq:native-tt-homogeneity} leaves $c^{-1}$. The kernel is
homogeneous of degree zero. Differentiating the
complete expression includes the volume and frequency-boundary terms together with the Gaussian variation. At fixed $c$ the torsion variation acts only on the displayed mass
shifts. This gives Eq.~\eqref{eq:native-tt-torsion-derivative} and its strict
sign.

The radius and torsion derivatives describe two specified variations of
the spectral sum. To resolve a local tensor source we must also vary
the modal operators and products using
Eq.~\eqref{eq:native-full-spectral-derivative} for the mixed response.

\paragraph{A local tensor source on the scalar carrier.}
\label{sec:native-local-metric-source}

The scalar source follows by variation of its spectral operator.
The metric in this construction serves as a comparison source on the
compact carrier and defines the operator variation. External propagation
is the response on the flat physical background.
Let $g_\epsilon=g+\epsilon h$ be a smooth metric
deformation on the same compact manifold and let
$t=\tfrac12\operatorname{tr}_g h$. Transport
$L^2(M,\dd v_{g_\epsilon})$ to $L^2(M,\dd v_g)$ by multiplication
with $(\dd v_{g_\epsilon}/\dd v_g)^{1/2}$. Denote the transported
positive scalar Laplacian by $\widetilde A_\epsilon$.

\begin{proposition}[Geometric modal source]
\label{prop:native-geometric-source}
For scalar eigenfunctions $u_i,u_j$ of $A=-\Delta_g$ the first
source vertex is
\begin{align}
 H_{ij}[h]&:=\langle u_i,\widetilde A'_0u_j\rangle\notag\\
 &=\int_M(tg^{mn}-h^{mn})\nabla_m\overline u_i\nabla_nu_j\,\dd v_g
 -\frac{\lambda_i+\lambda_j}{2}
       \int_Mt\overline u_i u_j\,\dd v_g .
 \label{eq:native-geometric-vertex}
\end{align}
For trace-free $h$ this reduces to
$H_{ij}[h]=-\int h^{mn}\nabla_m\overline u_i\nabla_nu_j\,\dd v_g$.
For a constant compression $h=-2\eta g$ it gives
$H_{ij}=2\eta\lambda_i\delta_{ij}$.
On a finite complete-shell subspace the heat operator has derivative
\begin{equation}
 \bigl[D e^{-\upsilon A}[H]\bigr]_{ij}
 =-H_{ij}\int_0^\upsilon
       e^{-(\upsilon-r)\lambda_i-r\lambda_j}\,\dd r .
 \label{eq:native-geometric-heat-response}
\end{equation}
Thus its off-diagonal matter responses and diagonal spectral changes
come from the same geometric source.
\end{proposition}

\begin{proof}
Differentiate the Dirichlet form
$\int g_\epsilon^{mn}\nabla_m\overline u\nabla_nv\,\dd v_{g_\epsilon}$.
The inverse metric contributes $-h^{mn}$ and the measure contributes
$tg^{mn}$. The inverse comparison map has derivative $-t/2$ on
each argument. Integration against the eigenfunctions gives the last
term of Eq.~\eqref{eq:native-geometric-vertex}. Smoothness and
compactness justify these form derivatives on the common $H^1$ domain.
Differentiating the finite matrix exponential gives the integral in
Eq.~\eqref{eq:native-geometric-heat-response} including equal
eigenvalues. Compression yields the stated diagonal result by
orthonormality and the eigenfunction equation.
\end{proof}

This is the first local source realisation of the compact spectral data.
The tensor variation changes both eigenvalue weights and off-diagonal modal
mixing so the response retains channel conversion as well as endpoint
shifts.

This calculation also distinguishes matter-resolved response from a
homogeneous trace. On the round space the diagonal sum
$\sum_{i\in l}\nabla_m\overline u_i\nabla_nu_i$ over a complete scalar
shell is proportional to $g_{mn}$. The shell projector is isometry invariant and the isotropy group acts irreducibly on each tangent space. Consequently every complete-shell trace with a scalar shell
weight has zero first variation under trace-free $h$ while individual
vertices in Eq.~\eqref{eq:native-geometric-vertex} need not vanish.
Indeed for a nonconstant real mode $u$ the smooth trace-free
deformation $h=\dd u\otimes\dd u-|\nabla u|^2g/3$ gives
$H_{uu}[h]=-(2/3)\int_M|\nabla u|^4\dd v_g<0$.
Such individual vertices can survive projection onto a transverse
response space. The second variation and the matter-resolved mixed
response therefore carry information absent from the homogeneous
endpoint derivative.

For a simple positive scalar mass
$M_i=(\lambda_i+b)^{1/2}$ its geometric derivative is
$(H_{ii}[h]+b'[h])/(2M_i)$. Within a degenerate shell the corresponding
first-order splitting is obtained by diagonalising the basis-independent compression $P_l(\widetilde A'_0+b')P_l$. A curvature or mass term must be varied according to its
own spectral definition. These vertices live on the internal compact
geometry. The tensor and constraint equations construct their propagating
mixed response while the internal modal label and external momentum
retain their respective roles.

\subsection{Positive form-resolvent responses}
\label{sec:native-form-response}

The free-carrier propagation rule of
ref.~\cite{Guo2026CompactSpectrumI} is expressed through the internal spectrum. For internal
projectors $P_l$ and its scalar mass tower $M_l^2=\lambda_l+m_\delta^2$ with $m_\delta>0$ it is
\begin{equation}
 \widehat G_{\rm full}(Q)=\sum_l\frac{P_l}{Q^2+M_l^2}
       =(Q^2+\Delta_M+m_\delta^2)^{-1}.
 \label{eq:native-original-spectral-propagation}
\end{equation}
Here $Q$ is the external Euclidean momentum and $l$ labels an internal mode. The operator $\Delta_M=-\Delta_g$ is the positive scalar Laplacian.
For internal source profiles $g_a\in L^2(M)$ its resolved
matrix is $C_{ab}(Q)=\sum_l\langle g_a,P_lg_b\rangle/(Q^2+M_l^2)$.
Orthogonality and Cauchy--Schwarz bound the sum of the absolute
numerators by $\|g_a\|\|g_b\|$ proving convergence for $Q^2\geq0$.
At $Q=0$ the internal trace of a resolvent power converges for
powers greater than $3/2$. The same spectral construction applies to
a collective tensor once its sourced response operator is specified.

Suppose elimination of the nonpropagating variables of the same source
system gives the exact response
\begin{equation}
 C(z)=C_{\rm dir}(z)+B_{\rm src}^*(H+zK_{\rm prop})^{-1}B_{\rm src}.
 \label{eq:native-affine-response}
\end{equation}
The equality concerns the full response including its source maps.
The standard representation theorem for closed forms gives its positive
spectral realisation~\cite[Ch.~VI]{Kato1995}. Let $\mathfrak h$ be the dense closed nonnegative form associated with $H$. Assume $0<m_KI\leq K_{\rm prop}\leq M_KI$ and take the source map
$B_{\rm src}:\mathcal U\to\mathcal X$ to be bounded. Define
\begin{equation}
 \mathfrak a[u,v]=\mathfrak h[K_{\rm prop}^{-1/2}u,
                                      K_{\rm prop}^{-1/2}v],\quad
 D(\mathfrak a)=K_{\rm prop}^{1/2}D(\mathfrak h),\quad
 V=K_{\rm prop}^{-1/2}B_{\rm src}.
 \label{eq:native-conjugated-form}
\end{equation}
The form is closed and dense and its associated operator $A$ is
nonnegative self-adjoint. If $E_A$ is its spectral resolution then for
$z\notin(-\infty,0]$
\begin{equation}
 C(z)-C_{\rm dir}(z)
 =\int_{[0,\infty)}\frac{\dd\rho_{\rm nat}(\mu)}{\mu+z},\qquad
 \dd\rho_{\rm nat}=V^*\dd E_A V,\qquad
 Z_{\rm nat}=V^*P_{\ker A}V.
 \label{eq:native-positive-resolvent-measure}
\end{equation}
The measure is positive operator-valued and finite. Its total mass is
$V^*V$. Moreover $x(C(x)-C_{\rm dir}(x))\to Z_{\rm nat}$ strongly as
$x\downarrow0$.
Indeed the bounded invertible change of variables makes the two form norms
equivalent while the form equation gives
\[
 (H+zK_{\rm prop})^{-1}
 =K_{\rm prop}^{-1/2}(A+z)^{-1}K_{\rm prop}^{-1/2}.
\]
The spectral theorem proves Eq.~\eqref{eq:native-positive-resolvent-measure}. Spectral dominated convergence applied to $x/(\mu+x)$ gives the strong
zero-mode limit.

The closed positive form associates an internal response rule with a positive spectral measure. Its zero-mode projection is tested through the source embedding.

The variable $\mu$ is the spectral parameter of the reduced form operator.
Spacetime reconstruction supplies the invariant-mass interpretation used
in the physical matching.
The global affine form in Eq.~\eqref{eq:native-affine-response} supplies
the positive measure above. For an analytic direct term or a local
expansion $H+zK_{\rm prop}+O(z^2)$ the Laurent analysis of the reciprocal
operator supplies the corresponding local singularity classification.

A concrete gapped response is supplied by the generated Yang--Mills
construction of ref.~\cite{Guo2026CompactSpectrumI} Section~8 on its
condensed branch $L>L_c$. At the stationary mean intensity $q_L$ the
normalised radial response and the
coexact connection response have mass operator
\[
 \mathcal M_{\rm YM}^2
 =m_{\rm YM,rad}^2\oplus(\Delta_1^\perp+\delta_{q_L}I),\qquad
 \inf\operatorname{Spec}\mathcal M_{\rm YM}^2
 =\min\{m_{\rm YM,rad}^2,4/L^2+\delta_{q_L}\}>0.
\]
Here $\Delta_1^\perp$ is the coexact one-form Laplacian and
$\delta_{q_L}>0$ is the restoring shift at that stationary source state.
For the curvature intensity $\mathcal Q=\|F_B\|^2$ of the generated
connection sample $B$ let $\sigma_L^2=\operatorname{Var}_{q_L}\|B\|^2$.
Its source-mean derivative has radial coefficient
$b_{\mathcal Q}=\sigma_L\partial_q\mathbb E_q\mathcal Q|_{q=q_L}>0$.
In its quadratic radial
realisation Eq.~\eqref{eq:native-affine-response} has $H=m_{\rm YM,rad}^2$ and $K_{\rm prop}=1$ with $B_{\rm src}=b_{\mathcal Q}$. Hence
\[
 C_{\rm rad}(Q^2)=\frac{b_{\mathcal Q}^{\,2}}
                         {Q^2+m_{\rm YM,rad}^2}.
\]
The positive atom lies at a nonzero mass with weight fixed by the source
overlap. It gives a finite susceptibility at $Q=0$ and a Yukawa exchange
in the spatial realisation. The original fixed-source curvature
polynomial is represented in the same response completion with its
Gaussian products and higher spectral components intact. The radial
formula measures the change of its mean intensity. In the condensed
quadratic field realisation of that reference Appendix~C the Wick
intensity itself also contains a
linear radial term followed by a two-particle continuum. These examples
show how the chosen source and observation determine which part of a
generated positive spectrum is measured.

\subsection{The covariant stress-tensor realisation and its TT quotient}
\label{sec:covariant-stress-realisation}

The compact construction and the covariant QFT construction have
complementary roles. The former supplies explicit internal spectral data
and computable source derivatives. The latter fixes the positive measure and Lorentz orbit structure that give the stress response its physical interpretation.
Their connection is formulated through the sourced response map and
tested by covariant spectral matching.

Matching identifies the reduced operator and its source and readout maps with the TT quadratic form of the continued stress-tensor
correlator. In the affine representation the form and source embedding together with the direct term come from the same sourced equations. Under this
identification the reduced-form spectral parameter continues to the
external correlator variable.

The input in this realisation is a local unitary QFT in flat
spacetime with a renormalised symmetric conserved stress tensor. Gauge
theories are taken on their positive physical BRST cohomology with the
relevant anomalies cancelled. We assume Wightman temperedness together with the spectral condition and usual Euclidean continuation. The source quotient removes representative ambiguities. Lorentz-orbit decomposition identifies the measure. Continuation defines its Euclidean coefficient. We work in four dimensions with
metric signature $(+---)$ and units $\hbar=1$ with the velocity unit fixed. We use
a Poincar\'e-invariant vacuum and require the complete symmetric tensor to
transform exactly covariantly on the positive physical Hilbert space.
These assumptions apply to the full tensor measure. The null-shell restriction and timelike spin decomposition use their respective orbit stabilisers.

\label{sec:centered-stress}

Let $T_{\mu\nu}$ be a renormalised symmetric stress tensor satisfying
\begin{equation}
 \partial^\mu T_{\mu\nu}=0
 \label{eq:stress-conservation}
\end{equation}
as an operator identity modulo the usual contact terms. In a gauge theory
$T_{\mu\nu}$ denotes a ghost-number-zero BRST cohomology class. A convenient
representative is obtained by coupling the microscopic theory to a
nondynamical metric and differentiating the renormalised action
\begin{equation}
 T_{\mu\nu}(x)
 =\left.\frac{2}{\sqrt{-g}}
   \frac{\delta S_{\rm ren}[g]}{\delta g^{\mu\nu}(x)}
  \right|_{g=\eta}.
 \label{eq:stress-background-definition}
\end{equation}
The metric in Eq.~\eqref{eq:stress-background-definition} is a
nondynamical source used to define the insertion.

The connected spectral problem is expressed in terms of
\begin{equation}
 \Theta_{\mu\nu}(x)
 :=T_{\mu\nu}(x)
  -\langle\Omega|T_{\mu\nu}(x)|\Omega\rangle I.
 \label{eq:centered-stress-tensor}
\end{equation}
Centring removes the constant vacuum expectation value from connected
correlators. The same constant drops out of the commutator. The
cosmological term in a gravitational effective action is a separate
parameter.

The Wightman and retarded functions play different roles
\begin{align}
 W_{\mu\nu,\rho\sigma}(x)
 &=\langle\Omega|\Theta_{\mu\nu}(x)
                 \Theta_{\rho\sigma}(0)|\Omega\rangle,
 \label{eq:stress-wightman}\\
 G^R_{\mu\nu,\rho\sigma}(x)
 &=\iu\,\theta(x^0)
   \langle\Omega|[\Theta_{\mu\nu}(x),
                    \Theta_{\rho\sigma}(0)]|\Omega\rangle.
 \label{eq:stress-retarded}
\end{align}
The first supplies a positive spectral measure. The second supplies causal
linear response. The Fourier transform of the commutator is the difference of the positive- and negative-energy boundary values.

\paragraph{The physical TT source quotient.}
\label{sec:tt-source-quotient}

For nonzero Euclidean momentum define
\begin{align}
 \theta_{\mu\nu}(Q)
 &=\delta_{\mu\nu}-\frac{Q_\mu Q_\nu}{Q^2},
 \label{eq:euclidean-transverse-projector}\\
 \Pi^{(2)}_{\mu\nu,\rho\sigma}(Q)
 &=\frac12\bigl(\theta_{\mu\rho}\theta_{\nu\sigma}
                 +\theta_{\mu\sigma}\theta_{\nu\rho}\bigr)
   -\frac13\theta_{\mu\nu}\theta_{\rho\sigma}.
 \label{eq:euclidean-tt-projector}
\end{align}
It has trace
$N_2=\Pi^{(2)}_{\mu\nu,\mu\nu}=5$.
The scalar coefficient is formed at $Q^2>0$ before taking a radial or distributional limit. Its rank-five coefficient is the massive
timelike spin-two coefficient. The null-shell quotient and source matching select the massless helicities $+2$ and $-2$.

An equivalent definition avoids the singular projector. Let
$\mathcal J_{\TT}$ be the space of smooth compactly supported symmetric
test sources satisfying
\begin{equation}
 \partial_\mu J^{\mu\nu}=0,
 \qquad \eta_{\mu\nu}J^{\mu\nu}=0.
 \label{eq:tt-test-source}
\end{equation}
The smeared insertion
$\Theta(J)=\int\dd^4x\,J^{\mu\nu}\Theta_{\mu\nu}$ depends only on the
physical TT class of the source. This formulation remains meaningful on the
massless shell and is used to define positivity.

Stress tensors are not unique as local representatives. Consider the
standard improvement and gauge-theory redundancies
\begin{equation}
 T'_{\mu\nu}=T_{\mu\nu}
 +(\partial_\mu\partial_\nu-\eta_{\mu\nu}\Box)X
 +s_{\mathrm B}Y_{\mu\nu}
 +\mathcal E^R_{F,\mu\nu},
 \label{eq:stress-representative-change}
\end{equation}
where $s_{\mathrm B}$ is the BRST differential and $\mathcal E^R_F$ is
the renormalised quantum EOM insertion generated by a local field
redefinition. It includes the Jacobian and operator-mixing subtractions
needed for the Schwinger--Dyson identity rather than an arbitrary bare
product $F\,\delta S/\delta\phi$. Local counterterms quadratic in the
background source have contact second derivatives.

\paragraph{Representative independence of the nonlocal TT datum.}
\label{prop:tt-representative-independence}
At separated points and between physical states the nonlocal quadratic
form of $\Theta$ on $\mathcal J_{\TT}$ is unchanged by
Eq.~\eqref{eq:stress-representative-change}. A finite local background-source
counterterm changes only the subtraction polynomial. Consequently the existence and rank of a zero-mass TT spectral atom together with its sign are independent of
these choices. Indeed integration by parts and
Eq.~\eqref{eq:tt-test-source} give
\begin{equation}
 \int\dd^4x\,J^{\mu\nu}
 (\partial_\mu\partial_\nu-\eta_{\mu\nu}\Box)X=0.
 \label{eq:improvement-tt-zero}
\end{equation}
BRST-exact insertions vanish between physical cohomology classes. The
renormalised EOM insertion acts by the local field variation on the other
insertions so its time-ordered products reduce to contact terms by the
Schwinger--Dyson identity. Finally functional derivatives of
a local source counterterm have support at coincident points and are
polynomial in momentum. None of these operations changes the nonlocal
Wightman measure on physical TT test sources. An invertible finite
 renormalisation of a physical operator basis acts by congruence on that measure and preserves its positivity and rank together with the presence of an atom. This
 statement concerns the explicit class
Eq.~\eqref{eq:stress-representative-change}. A more general conserved local
redefinition must be resolved into physical improvements and BRST and EOM components together with contact terms before the conclusion is invoked.

\subsection{Positive Wightman measure and the null-shell boundary}
\label{sec:tt-wightman-measure}

Insertion of the joint spectral resolution of the four-momentum operator in
Eq.~\eqref{eq:stress-wightman} gives a tensor-valued measure supported in the
closed forward cone. The full measure is positive on every complex symmetric test tensor $K$ for which the smearing is defined
\begin{equation}
 \langle\Omega|\Theta(K)^\dagger\Theta(K)|\Omega\rangle
 =\int_{\overline V_+}
  \overline{\widetilde K^{\mu\nu}(p)}\,
  \dd M_{\mu\nu,\rho\sigma}(p)\,
  \widetilde K^{\rho\sigma}(p)\geq0.
 \label{eq:wightman-positive-quadratic-form}
\end{equation}
The operator conservation law gives the two distributional Ward identities
\begin{equation}
 p^\mu\dd M_{\mu\nu,\rho\sigma}(p)=0,
 \qquad
 \dd M_{\mu\nu,\rho\sigma}(p)p^\rho=0.
 \label{eq:wightman-measure-ward-identities}
\end{equation}
These are identities of the Wightman measure. Contact terms in time-ordered
products enter the subtraction polynomial and do not modify them.

We first retain the matrix structure of the source response. Let $E(\dd p)$ be the joint projection-valued measure of translations.
The tensor-valued tempered measure
$\dd M_{\mu\nu,\rho\sigma}(p)$ is defined weakly by smearing the two local
operator-valued distributions against test tensors. For a finite family
$J_a\in\mathcal J_{\TT}$ define
\begin{equation}
 \dd M_{ab}(p)
 :=\langle\Omega|\Theta(J_a)^\dagger E(\dd p)
                    \Theta(J_b)|\Omega\rangle.
 \label{eq:body-matrix-tt-measure}
\end{equation}

\begin{theorem}[Covariant disintegration of the physical TT measure]
\label{thm:matrix-tt-wightman-measure}
Assume Wightman temperedness together with the spectral and Poincar\'e-covariance
conditions stated above. Then
\begin{enumerate}[label=(\roman*)]
\item for every finite family of physical TT test sources $\dd M(p)$ is a
positive-semidefinite matrix-valued Radon measure supported in the closed
forward cone. More precisely for $c_a\in\mathbb C$ and every nonnegative
Borel function $f$ for which the expression is finite
\begin{equation}
 \sum_{a,b}\bar c_a\int f(p)\,\dd M_{ab}(p)c_b
 =\left\|\left[\int f(p)E(\dd p)\right]^{1/2}
   \sum_b c_b\Theta(J_b)|\Omega\rangle\right\|^2
 \geq0.
 \label{eq:body-matrix-tt-positivity}
\end{equation}
\item fix the invariant orbit measure
$\dd\Omega_s(p)=(2\pi)^{-3}\theta(p^0)\delta(p^2-s)\dd^4p$ on
$H_s^+=\{p:p^0>0,\ p^2=s\}$. On the open timelike cone there are unique
source-independent positive locally finite measures
$\dd\rho_2(s)$ and $\dd\rho_0(s)$ such that for every symmetric Schwartz
test tensor $K_{\mu\nu}$ supported away from $p=0$
\begin{align}
 \langle K,MK\rangle_{p^2>0}
 =\sum_{A=2,0}\int_{(0,\infty)}\dd\rho_A(s)
  \int_{H_s^+}\dd\Omega_s(p)\,
  \overline{K^{\mu\nu}(p)}
  \Pi^{(A)}_{\mu\nu,\rho\sigma}(p)K^{\rho\sigma}(p),
 \label{eq:body-covariant-spin-disintegration}
\end{align}
where only for $s>0$
\begin{align}
 \pi_{\mu\nu}(p)&=\eta_{\mu\nu}-\frac{p_\mu p_\nu}{s},
 \notag\\
 \Pi^{(2)}_{\mu\nu,\rho\sigma}(p)
 &=\frac12(\pi_{\mu\rho}\pi_{\nu\sigma}
            +\pi_{\mu\sigma}\pi_{\nu\rho})
   -\frac13\pi_{\mu\nu}\pi_{\rho\sigma},
 \qquad
 \Pi^{(0)}_{\mu\nu,\rho\sigma}(p)
 =\frac13\pi_{\mu\nu}\pi_{\rho\sigma}.
 \label{eq:body-lorentz-spin-projectors}
\end{align}
For $J_a,J_b\in\mathcal J_{\TT}$ and Borel sets
$B\Subset(0,\infty)$ the pushforward of Eq.~\eqref{eq:body-matrix-tt-measure}
is therefore
\begin{equation}
 \rho^{J}_{ab}(B)
 =\int_B q^{(2)}_{ab}(s)\,\dd\rho_2(s),
 \quad
 q^{(2)}_{ab}(s)
 :=\int_{H_s^+}\dd\Omega_s(p)\,
 \overline{\widetilde J_a^{\mu\nu}(p)}
 \Pi^{(2)}_{\mu\nu,\rho\sigma}(p)
 \widetilde J_b^{\rho\sigma}(p).
 \label{eq:body-source-factorisation}
\end{equation}
Thus all source dependence is in the positive matrix $q^{(2)}(s)$. The
measure $\dd\rho_2$ is common to every source family.
\item at $s=0$ the pushforward itself rather than
Eq.~\eqref{eq:body-lorentz-spin-projectors} defines the matrix atom
\begin{equation}
 \mathbf Z_0[J]
 :=\bigl[\rho^{J}_{ab}(\{0\})\bigr]_{a,b}\succeq0.
 \label{eq:body-matrix-tt-atom}
\end{equation}
It is defined without evaluating a massive rank-five projector at $p^2=0$.
\end{enumerate}
\end{theorem}

\begin{proof}
Equation~\eqref{eq:body-matrix-tt-positivity} is the spectral theorem applied
to $\sum_b c_b\Theta(J_b)|\Omega\rangle$. On each timelike orbit
conservation leaves the inequivalent $SO(3)$ representations
$\operatorname{Sym}^2_0(\mathbb R^3)$ and $\mathbb R$. Schur's lemma and
positivity therefore give two unique nonnegative scalar measures. TT sources
remove the scalar block and pushforward by $p\mapsto p^2$ preserves matrix
positivity at the boundary. Appendix~\ref{app:tt-spectral} gives the orbit
disintegration and reconstruction in the conventions used here. This is the
tensor K\"all\'en--Lehmann construction
\cite{Kallen1952,Streater1964,Weinberg1995}.
\end{proof}

The two measures in the theorem serve different purposes.
The common tensor measure fixes $\dd\rho_2$ while a test-source family
fixes $\dd\rho^J_{ab}$ through the factors in
Eq.~\eqref{eq:body-source-factorisation}.

\begin{theorem}[Null-shell exclusion for the covariant stress tensor]
\label{thm:null-shell-exclusion}
For a complete tensor satisfying positivity and exact Lorentz covariance together with conservation and the invariant-vacuum assumption the restriction of the
Wightman tensor measure to $H_0^+=\{p:p^0>0,\ p^2=0\}$ has the form
\begin{equation}
 \left.\dd M_{\mu\nu,\rho\sigma}(p)\right|_{H_0^+}
 =\alpha_{\rm null}\,p_\mu p_\nu p_\rho p_\sigma\,\dd\Omega_0(p),
 \qquad \alpha_{\rm null}\geq0.
 \label{eq:actual-null-shell-measure}
\end{equation}
 It vanishes on every conserved test source. A possible measure at $p=0$
 has the form
 \begin{equation}
  \left.\dd M_{\mu\nu,\rho\sigma}(p)\right|_{\{0\}}
  =\beta_0\,\eta_{\mu\nu}\eta_{\rho\sigma}\delta^{(4)}(p),
  \qquad \beta_0\geq0,
  \label{eq:actual-origin-measure}
 \end{equation}
 and vanishes on traceless sources. Consequently the TT matrix atom is
\begin{equation}
 \mathbf Z_0[J]=0\qquad(J_a\in\mathcal J_{\TT}).
 \label{eq:baseline-matrix-atom-zero}
\end{equation}
This is a classification of the Wightman measure itself prior to any
asymptotic one-particle interpretation.
\end{theorem}

\begin{proof}
Proposition~\ref{prop:positive-null-orbit-measure} gives the complete orbit
classification. Conservation annihilates the term in
Eq.~\eqref{eq:actual-null-shell-measure} while tracelessness annihilates
Eq.~\eqref{eq:actual-origin-measure}. Pushing the restricted measure forward
under $p\mapsto p^2$ therefore gives Eq.~\eqref{eq:baseline-matrix-atom-zero}.
The one-particle counterpart is the little-group result of
ref.~\cite{Weinberg2020Massless}. The proposition above works directly at
the level of the positive Wightman measure.
\end{proof}

The theorem concerns the complete covariant stress-tensor measure. Exact
covariance confines its null-shell support to the conserved longitudinal
sector. A physical two-helicity atom is represented instead on a reduced response space or a gauge-potential quotient carrying its own positive
covariance law while nonatomic long-distance correlations remain part of
the complete measure.

Let
\begin{equation}
 G^E_{\mu\nu,\rho\sigma}(Q)
 =\langle\Theta_{\mu\nu}(Q)
          \Theta_{\rho\sigma}(-Q)\rangle_{E,c}
 \label{eq:euclidean-stress-correlator}
\end{equation}
be the connected Euclidean correlator. At $Q^2>0$ its TT coefficient is
\begin{equation}
 C_2(Q^2)
 :=\frac1{N_2}\Pi^{(2)}_{\mu\nu,\rho\sigma}(Q)
                   G^E_{\mu\nu,\rho\sigma}(Q).
 \label{eq:tt-form-factor-definition}
\end{equation}
All kinematic powers are included in this definition. $C_2$ has mass
dimension four. Thus the scalar spectral density includes the kinematic powers
associated with this normalisation.

\paragraph{Boundary construction of the scalar TT measure.}
Fix the orbit normalisation in Theorem~\ref{thm:matrix-tt-wightman-measure}
and the Euclidean normalisation in Eq.~\eqref{eq:tt-form-factor-definition}.
The timelike measures obey for some $\epsilon>0$
\begin{equation}
 \int_{(0,\epsilon)}\frac{\dd\rho_2(s)+\dd\rho_0(s)}{s^2}<\infty.
 \label{eq:physical-tensor-ir-integrability}
\end{equation}
Extend $\rho_2$ to $[0,\infty)$ by assigning zero mass to $\{0\}$.
This positive locally finite extension supplies the subtracted
K\"all\'en--Lehmann representation of $C_2$. In particular
\begin{equation}
 Z_0:=\rho_2(\{0\})=0.
 \label{eq:body-scalar-tt-atom}
\end{equation}
The scalar and source-matrix atoms agree through their common
vanishing
\begin{equation}
 \mathbf Z_0[J]=0,\qquad Z_0=0.
 \label{eq:body-matrix-scalar-atom-correspondence}
\end{equation}
The endpoint extension follows from positivity and local finiteness rather
than from a formal limit of the massive projector. On a compact momentum
patch near the nonzero null cone the positive components
\begin{equation}
 \Pi^{(2)}_{00,00}=\frac{2|\boldsymbol p|^4}{3s^2},\qquad
 \Pi^{(0)}_{00,00}=\frac{|\boldsymbol p|^4}{3s^2}.
 \label{eq:positive-component-ir-bound}
\end{equation}
have a shell integral bounded uniformly below as $s\downarrow0$.
Local finiteness therefore gives
Eq.~\eqref{eq:physical-tensor-ir-integrability}. In the time-ordered
reconstruction external TT contraction selects the timelike spin-two
coefficient and removes the scalar and longitudinal null-shell pieces. Off-shell extensions differ only by local terms. Appendix~\ref{app:tt-spectral}
gives this reconstruction and the contact-term argument. Standard
continuation then yields the Stieltjes transform
\cite{Kallen1952,Lehmann1954,Streater1964,OsterwalderSchrader1973} whose
measure is unique. The null-shell theorem fixes its agreement with the
source-matrix atom.

Choose a subtraction point $Q_*^2>0$ and a nonnegative integer $n_{\rm sub}$ sufficient for UV
convergence. Reflection positivity and analytic continuation give
\begin{equation}
 C_2(Q^2)=P_2(Q^2)
 +(Q_*^2-Q^2)^{n_{\rm sub}}
 \int_0^\infty
 \frac{\dd\rho_2(s)}
 {(s+Q^2)(s+Q_*^2)^{n_{\rm sub}}}.
 \label{eq:subtracted-tt-dispersion}
\end{equation}
The finite polynomial $P_2$ contains contact terms and subtraction
constants. Changing $n_{\rm sub}$ or $Q_*^2$ or a finite local counterterm changes the
analytic split but not $\dd\rho_2$.
We write
\begin{equation}
 C_2^{\rm nl}(Q^2):=C_2(Q^2)-P_2(Q^2)
 \label{eq:nonlocal-tt-form-factor}
\end{equation}
for the dispersive part in the chosen subtraction scheme.

The variable $s$ is the independent invariant-mass variable in the spectral
integral.  The complexified external invariant is denoted by $z$. Its Euclidean value is $z=Q^2$ and its Lorentzian boundary value is $z=-p^2-i0$ for momentum transfer $p$.  Thus analytic continuation acts on
$z$ and never on the integration variable $s$.
The compact modes carry discrete labels such as $l$ or $j$ and their
window scale is $k_{\rm win}$. The unadorned $k$ is reserved for the FRG
momentum scale. The Gaussian parameter $\tau$ has dimension inverse mass
squared. The subtraction order is $n_{\rm sub}$. The compact spectral cutoff is $n_{\rm sp}$ and the finite response dimension is $N$. With the
stated stress-tensor normalisation $C_2$ has mass
dimension four and $Z_0$ dimension six.

\subsection{Extraction and classification of the zero-mass atom}
\label{sec:tt-atom-extraction}

The following extraction formula applies to any positive scalar measure with
the stated integrability and source matching. For the covariant stress tensor
the boundary construction above already gives zero atomic weight. Decompose
the measure as
\begin{equation}
 \dd\rho_2(s)=Z_0\,\delta_0(\dd s)+\dd\rho_{2,c}(s),
 \qquad \rho_{2,c}(\{0\})=0,
 \label{eq:tt-atom-decomposition}
\end{equation}
where $Z_0\geq0$. The subscript $c$ denotes the complement of the zero atom. This remainder includes positive-mass atoms and singular continuous weight as well as any absolutely continuous part. The measure notation
also fixes the full weight of an atom at the integration endpoint.

\begin{proposition}[Euclidean extraction of the TT atom]
\label{prop:tt-atom-extraction}
Let
$\dd\nu(s)=\dd\rho_{2,c}(s)/(s+Q_*^2)^{n_{\rm sub}}$.
Assume that $\nu$ is locally finite at $s=0$ has no atom there and obeys
$\int_1^\infty s^{-1}\dd\nu(s)<\infty$. Then
\begin{equation}
 Z_0=\lim_{Q^2\downarrow0}Q^2C_2^{\rm nl}(Q^2).
 \label{eq:tt-atom-limit}
\end{equation}
In particular $Z_0>0$ is equivalent to an atom of the positive scalar
measure. Matching the measure to the physical source quadratic form identifies this weight with the TT atom. In the covariant
 stress-tensor realisation this matching is the boundary construction above
 and gives $Z_0=0$.
\end{proposition}

\begin{proof}
The atomic term in Eq.~\eqref{eq:subtracted-tt-dispersion} multiplied by
$Q^2$ tends to $Z_0$. For the remaining measure split the integral at $s=1$
in a fixed choice of mass units. On
$(0,1]$ the factor $Q^2/(s+Q^2)$ is bounded by one and tends pointwise to
zero. Local finiteness and the absence of an atom permit dominated
convergence. On $[1,\infty)$ it is bounded by $Q^2/s$ so the assumed UV
bound makes the tail vanish. The remaining subtraction prefactor is bounded
and tends to $(Q_*^2)^{n_{\rm sub}}$. This proves Eq.~\eqref{eq:tt-atom-limit}.
\end{proof}

The limit in Eq.~\eqref{eq:tt-atom-limit} is an infrared projection. A
positive value isolates spectral weight concentrated exactly at zero
invariant mass while every non-atomic contribution is suppressed by the
factor $Q^2/(s+Q^2)$. It therefore distinguishes an atomic propagation
channel from a continuum that merely reaches the origin. The source maps and spatial propagator determine the force carried by this channel.

The atomic weight and the support of the remainder distinguish the
infrared possibilities. When $Z_0>0$ and
$\rho_{2,c}((0,s_{\rm gap}))=0$ for some $s_{\rm gap}>0$ the atom is
isolated and gives an analytic simple pole. If the remaining weight
accumulates at zero the atom persists but there is no punctured
neighbourhood free of other spectral singularities. This includes both
continuous thresholds and accumulating positive-mass atoms.

When $Z_0=0$ a threshold reaching the origin may still produce a
long-distance power law. A gapped remainder has no massless component.
An absolutely continuous threshold can give logarithmic or fractional-power behaviour or a more general boundary singularity. The covariant
stress tensor falls into these nonatomic cases by
Theorem~\ref{thm:null-shell-exclusion}. Positivity and representative
independence are assumptions of this spectral classification rather
than additional spectral cases.

This distinction is already visible in free-field benchmarks. Free scalar and fermion stress tensors are bilinear as are those of free gauge fields. Its first
non-vacuum matrix elements therefore lie in two-particle sectors. Massive
constituents give a threshold at $s=4m^2$ while massless constituents give
a continuum beginning at the origin. In a four-dimensional scale-invariant
theory dimensional analysis fixes the TT density to scale as $s^2$ up to
normalisation. Logarithmic running can modify this scaling away from a fixed point
~\cite{OsbornPetkou1994}. The absence of a TT atom for the covariant tensor
also holds nonperturbatively by Theorem~\ref{thm:null-shell-exclusion}.
Cumulative spectral weights detect the finite-volume limit in which several vanishing levels concentrate into one atom. The precise limit and
its Euclidean extraction are treated in
Appendix~\ref{sec:lattice-test} and
Proposition~\ref{prop:finite-volume-cumulative-atom}.

\section{Source response and changes of resolution}
\label{sec:gaussian-frg}

The observable depends on a source both explicitly and through the
solution of the spectral equation. An eigenvalue of the unsourced operator
identifies an available direction. Nonzero source and readout projections
make that direction visible after the other response channels have been
eliminated. The implicit-function theorem separates these contributions
and fixes the complete linear response.

The equation and measured observable are differentiated with respect to the same source. This fixes the direct and propagated contributions together with the readout before a pole is interpreted as a physical
exchange channel.

This is the transport step between the modal and source data of
Section~\ref{sec:spectral-datum} and the momentum-dependent response
analysed in Section~\ref{sec:atom-kernel}. Compact source
derivatives enter the equation being solved while the observable is the
quantity compared with the stress-tensor correlator. Differentiation of this common rule gives the complete response tuple.

\subsection{Response from the generating equations}
\label{sec:native-same-source}

\begin{proposition}[Response of the sourced equations]
\label{prop:native-same-source}
\label{prop:composite-source-kernel-bridge}
Let $\mathcal E(Y,h)$ and $\mathcal O(Y,h)$ be continuously Fr\'echet
differentiable maps on Banach spaces near a solution
$\mathcal E(Y_*,0)=0$. If $L=D_Y\mathcal E(Y_*,0)$ is a bounded
isomorphism there is a local solution branch $Y(h)$ and its measured
linear response is
\begin{equation}
 C=\frac{\dd\mathcal O(Y(h),h)}{\dd h}\bigg|_0
   =\mathcal O_h-\mathcal O_Y L^{-1}\mathcal E_h.
 \label{eq:native-same-source-response}
\end{equation}
For any bounded isomorphism $M$ from the response space to the equation
space the operators
\begin{equation}
 \mathcal B=I-M^{-1}L,\quad
 \mathsf S=-M^{-1}\mathcal E_h,\quad
 \mathsf R=\mathcal O_Y,\quad C_{\rm dir}=\mathcal O_h
 \label{eq:native-same-source-tuple}
\end{equation}
give the identical response $C=C_{\rm dir}+\mathsf R(I-\mathcal B)^{-1}
\mathsf S$.
\end{proposition}

\begin{proof}
The implicit-function theorem gives $Y_h=-L^{-1}\mathcal E_h$.
Substitution in the chain rule proves
Eq.~\eqref{eq:native-same-source-response}. Since
$I-\mathcal B=M^{-1}L$ the inverse cancels the preconditioner in
Eq.~\eqref{eq:native-same-source-tuple}.
\end{proof}

For the compact branch the closure equation selects the endpoint at which
the accumulated channel strength reaches its prescribed TT value while
the observable records a retained carrier at that endpoint. Set
$\kappa_{\rm TT}=4/27$ and
\begin{equation}
\mathcal E(c,\vartheta)=\frac{A_{n_{\rm sp},R}(\vartheta)}c-\kappa_{\rm TT},
  \qquad \mathcal O_a(c,\vartheta)
  =\frac{\mu_*\sqrt{\widehat\lambda_a+\widehat m_a^2(\vartheta)}}c .
 \label{eq:native-spectral-closure-readout}
\end{equation}
The observable is the modal mass interpreted as the zero-momentum
spectral energy. Its inertial interpretation requires the matter vertex
and the reduced static response.

On the fixed branch described by
Eq.~\eqref{eq:native-tt-homogeneity} with $A=A_{n_{\rm sp},R}>0$ the closure has
exactly one positive root
\begin{equation}
  c_*(\vartheta)=\frac{A(\vartheta)}{\kappa_{\rm TT}},\qquad
  \mathcal E_c=-\frac A{c^2},\quad
  \mathcal E_\vartheta=\frac{A'}c,\quad
  c_*'=\frac{A'}{\kappa_{\rm TT}}.
 \label{eq:native-endpoint-root-response}
\end{equation}
Primes in this calculation denote $\vartheta$ derivatives.
For a positive modal mass $M_a=\mathcal O_a(c_*(\vartheta),\vartheta)$
\begin{equation}
 \frac{M_a'}{M_a}
 =\frac{(\widehat m_a^2)'}{2(\widehat\lambda_a+\widehat m_a^2)}
       -\frac{A'}A .
 \label{eq:native-endpoint-modal-response}
\end{equation}
The second term is common to all such modes. The first is the direct response
of the carrier mass shift. Since $A/c$ decreases continuously from infinity
to zero the crossing is unique. Implicit differentiation gives
$c_*'=-\mathcal E_c^{-1}\mathcal E_\vartheta$ and the chain rule for
$\mathcal O_a$ gives Eq.~\eqref{eq:native-endpoint-modal-response}.

The torsion variation changes both the radius endpoint and the modal
mass at fixed radius. At $\vartheta>0$ the stated shifts and $A'<0$
make $M_a'>0$. The common endpoint contribution is related to the
local energy source obtained by geometric compression.

The nonzero radius Jacobian makes the closure root regular in $c$.
This is the internal endpoint selected by the compact closure convention.
Propagation criticality instead concerns the external-invariant-dependent
operator $L(z;c_*)$ evaluated at that endpoint. The two roots coincide
through the specified source and inverse-response matching which maps the
compact endpoint to the external response branch and keeps the two
variables explicit.

Preconditioning preserves the response and the failure of invertibility of $L$. Algebraic eigenvalue multiplicities can change. For instance
\begin{equation}
 L(z)=\begin{pmatrix}z&0\\0&1\end{pmatrix},\qquad
 M=\begin{pmatrix}0&1\\1&0\end{pmatrix},\qquad
 M^{-1}L(0)=\begin{pmatrix}0&1\\0&0\end{pmatrix}.
 \label{eq:native-preconditioner-jordan}
\end{equation}
The zero eigenvalue of $L(0)$ is simple whereas the preconditioned
matrix has a double zero with a Jordan block. We therefore formulate
the positive-residue theorem for the reciprocal response operator. At
criticality the inverse-function argument is replaced by an analytic
singular-response analysis.

\subsection{Gaussian averaging and external resolution}
\label{sec:gaussian-characterisation}

Gaussian averaging is used as a resolution operation on a response that has
already been defined. Its semigroup law makes finite-resolution comparisons
composable and its multiplier is unity at zero momentum. The sourced
dynamics rather than the averaging kernel determines whether an atom is
present.

Let $D_r:x\mapsto rx$ denote dilation on $\mathbb R^d$. Consider a family
of probability measures $\{\mu_\tau\}_{\tau\geq0}$ that is weakly
continuous at $\tau=0$ obeys
\begin{equation}
 \mu_{\tau_1}*\mu_{\tau_2}=\mu_{\tau_1+\tau_2},
 \qquad \mu_0=\delta_0,
 \label{eq:gaussian-semigroup-law}
\end{equation}
is centred and orthogonally invariant and has a finite nonzero second moment at $\tau=1$. Suppose it also satisfies exact square-root self-similarity
\begin{equation}
 \mu_\tau=(D_{\sqrt\tau})_\#\mu_1.
 \label{eq:gaussian-self-similarity}
\end{equation}
The convolution law describes composition of independent averaging steps. Centering removes drift and square-root self-similarity fixes the relation
between width and scale. Together these conditions specify the averaging
class used here.

Independent averaging steps compose additively while their width scales as the square root of the resolution parameter. These requirements select the Gaussian semigroup within the stated class.

\begin{proposition}[Gaussian averaging law]
\label{prop:gaussian-uniqueness}
Under the preceding assumptions the characteristic function is
\begin{equation}
 \widehat\mu_\tau(q)=\exp(-D\tau |q|^2)
 \label{eq:gaussian-characteristic-function}
\end{equation}
for a constant $D>0$. Thus the law is Gaussian up to the unit chosen for
$\tau$.
\end{proposition}

\begin{proof}
Weak continuity and reflection symmetry together with the convolution law give
$\widehat\mu_\tau(q)=\exp[-\tau\psi(q)]$ with $\psi\geq0$.
Self-similarity makes $\psi$ homogeneous of degree two. Orthogonal
invariance then gives $\psi(q)=D|q|^2$. The nonzero covariance fixes
$D>0$. Appendix~\ref{app:gaussian} supplies the positivity argument
for the characteristic function and the Fourier inversion.
\end{proof}

This is the Gaussian case of the classical classification of
self-similar convolution semigroups \cite[Chapter~3]{Sato2013}.
The generated-operation construction in
ref.~\cite{Guo2026Arithmetic} supplies the broader mathematical setting
for this averaging law. Proposition~\ref{prop:gaussian-uniqueness} and
Appendix~\ref{app:gaussian} give the self-contained form used here.

With a rescaling of $\tau$ we set $D=1$. The corresponding heat semigroup is
\begin{equation}
 (\cH_\tau f)(x)
 =\frac1{(4\pi\tau)^{d/2}}
  \int\dd^dy\,e^{-|x-y|^2/(4\tau)}f(y),
 \qquad
 \widehat{\cH_\tau f}(Q)=e^{-\tau Q^2}\widehat f(Q).
 \label{eq:heat-semigroup-action}
\end{equation}
If the coordinate-space width is denoted by $\sigma$ our convention is
$\tau=\sigma^2/2$ for a Fourier factor
$e^{-\sigma^2Q^2/2}$.

The quadratic exponent follows already from orthogonal invariance and
square-root self-similarity. The nonzero second moment fixes its
nondegenerate normalisation. If square-root self-similarity is replaced
by $\alpha$-self-similarity at the same time as the finite-second-moment
condition is dropped the other stable semigroups have
$\mu_\tau=(D_{\tau^{1/\alpha}})_\#\mu_1$ with $0<\alpha<2$.
The averaging law fixes external resolution. The functional RG regulator $R_k$ controls the integration of fluctuations. Gaussian resolution therefore preserves
and resolves the spectral content of the defined response whereas the
existence of a massless channel is decided by its unsmeared dynamics.

\paragraph{Preservation of the atom under external smearing.}
\label{sec:external-gaussian}

Smearing each stress-tensor insertion with $\cH_{\tau/2}$ multiplies the
Euclidean two-point form factor by
\begin{equation}
 C_{2,\tau}^{\rm ext}(Q^2)
 =e^{-\tau Q^2}C_2(Q^2).
 \label{eq:external-gaussian-correlator}
\end{equation}
The averaging acts on the external insertions of the fixed physical
response.

The same multiplier gives the following atom-preservation identity at every
finite $\tau\geq0$
\begin{equation}
 \lim_{Q^2\downarrow0}Q^2
 e^{-\tau Q^2}C_2^{\rm nl}(Q^2)=Z_0.
 \label{eq:gaussian-atom-limit}
\end{equation}
The multiplier is entire and nonzero with value one at the origin. Thus
if $C_2^{\rm nl}$ is regular at $Q^2=0$ the product is regular there. If
$C_2^{\rm nl}=Z_0/Q^2+C_{\rm rem}$ with
$Q^2C_{\rm rem}\to0$ then
\begin{equation}
 e^{-\tau Q^2}C_2^{\rm nl}(Q^2)
 =\frac{Z_0}{Q^2}+C_{\rm rem}(Q^2)-\tau Z_0+o(1).
 \label{eq:gaussian-pole-expansion-paper1}
\end{equation}
the residue is unchanged. A regular form factor remains regular after the
same multiplication.

The filtered expression in Eq.~\eqref{eq:external-gaussian-correlator} need
not itself be a Stieltjes function and
$e^{-\tau Q^2}P_2(Q^2)$ is not a contact polynomial. Positivity and contact
subtraction are therefore imposed on the unfiltered correlator first.
Equation~\eqref{eq:external-gaussian-correlator} preserves the infrared
limit. At finite $\tau$ the operator $\cH_\tau$ is injective on tempered distributions. For $\tau>0$ its inverse is unbounded on the natural Hilbert spaces.
Its use therefore entails a loss of resolution at finite precision
which is quantified separately from the functional truncation error.

\subsection{The background-source flow and the TT response space}
\label{sec:background-source-frg}

For the QFT realisation introduce in Euclidean signature a source
$h_{\mu\nu}$ for the centred
renormalised stress tensor and write
\begin{equation}
 Z_k[j,h]=\int\mathcal D\phi\,
 \exp\left[-S[\phi]-\Delta S_k[\phi]
      +j\cdot\phi+h\cdot\Theta_R\right],
 \qquad W_k=\log Z_k,
 \label{eq:paper1-composite-generating-functional}
\end{equation}
where
$h\cdot\Theta_R=\int\dd^4x\,h^{\mu\nu}\Theta_{\mu\nu,R}$ and
$\Delta S_k=\tfrac12\phi\cdot R_k\cdot\phi$ is schematic superfield
notation. Counterterms polynomial in $h$ are included from the outset. They
renormalise coincident composite insertions.
The auxiliary source $h$ probes the connected response of the stress
operator and provides the QFT implementation of the same source-resolved
logic used above. It remains a nondynamical background variable. A
collective field supplies a candidate carrier in a separate response
representation. Matching the operator and its source and observable maps relates that representation to the stress correlator.

Thus the FRG flow is a QFT realisation of the same source--equation--
observable triple. It evolves the vertex data used to construct the
response while the background source remains a nondynamical probe.

The modified Legendre transform is taken only with respect to $j$
\begin{equation}
 \Gamma_k[\varphi,h]
 =j\cdot\varphi-W_k[j,h]-\Delta S_k[\varphi],
 \qquad \varphi=\frac{\delta W_k}{\delta j}.
 \label{eq:paper1-modified-legendre-transform}
\end{equation}
At fixed $h$ it obeys the exact Wetterich equation~\cite{Wetterich1993}
\begin{equation}
 \partial_t\Gamma_k[\varphi,h]
 =\frac12\STr\left[G_k[\varphi,h]\,\partial_tR_k\right],
 \qquad
 G_k=(\Gamma_k^{(2,0)}+R_k)^{-1},
 \qquad t=\log(k/k_0).
 \label{eq:paper1-wetterich}
\end{equation}
Here the first superscript counts fundamental-field derivatives and the
second counts $h$ derivatives. Scale-dependent composite fields and
their correlation functions can also be treated by the flow construction
of ref.~\cite{FloerchingerWetterich2009}. Here we retain the fixed composite
source and differentiate its generating functional directly.

For a source-independent regulator differentiating
Eq.~\eqref{eq:paper1-wetterich} twice with respect to source components
$h_i,h_j$ gives
\begin{align}
 \partial_t\Gamma^{(0,2)}_{k,ij}
 =\frac12\STr\bigl[&G_k\Gamma^{(2,1)}_{k,j}G_k
                   \Gamma^{(2,1)}_{k,i}G_k\partial_tR_k
 +G_k\Gamma^{(2,1)}_{k,i}G_k
                   \Gamma^{(2,1)}_{k,j}G_k\partial_tR_k
 \notag\\[-2mm]
 &-G_k\Gamma^{(2,2)}_{k,ij}G_k\partial_tR_k\bigr].
 \label{eq:paper1-composite-two-source-flow}
\end{align}
The equation is part of the full vertex hierarchy. For a source-dependent
regulator its derivatives include the corresponding derivatives of $R_k$.

The connected source response includes propagation through the fundamental-field sector. Differentiating Eq.~\eqref{eq:paper1-modified-legendre-transform} gives the Schur identity
\begin{equation}
 W^{(0,2)}_{k,hh}
 =-\Gamma^{(0,2)}_{k,hh}
  +\Gamma^{(1,1)}_{k,h\varphi}
   (\Gamma_k^{(2,0)}+R_k)^{-1}
   \Gamma^{(1,1)}_{k,\varphi h},
 \label{eq:source-schur-identity}
\end{equation}
with signs fixed by Eq.~\eqref{eq:paper1-modified-legendre-transform}.
At $k=0$ after renormalisation and continuation the TT projection of
$W^{(0,2)}_{hh}$ is the quantity that must be matched to
Eq.~\eqref{eq:subtracted-tt-dispersion}.
Appendix~\ref{sec:frg-kernel-test} relates this endpoint to the
operator and source convergence required of finite vertex expansions.

\paragraph{The FRG realisation of the response tuple.}
At vanishing fundamental-field source $j$ and for the source-independent
regulator used in Eq.~\eqref{eq:paper1-composite-two-source-flow} define the
exact stationarity map
\begin{equation}
 \mathcal E_k(\varphi,h)
 :=\Gamma_{k,\varphi}[\varphi,h]+R_k\varphi=0,
 \qquad
 \mathcal A_k:=D_\varphi\mathcal E_k
 =\Gamma_{k,\varphi\varphi}+R_k=G_k^{-1}.
 \label{eq:canonical-stationarity-map}
\end{equation}
Choose a noncritical source-free reference branch $\varphi_*$ at
$h=0$. Assume that $\mathcal E_k$ is continuously Fr\'echet differentiable
near this branch and that its Jacobian $\mathcal A_k$ is a bounded
isomorphism on a specified response space. The corresponding observable is
$\mathcal O_k(\varphi,h)=-\Gamma_{k,h}[\varphi,h]$ evaluated on the
stationary branch. All response maps are evaluated on this branch. External momentum labels Fourier blocks of the
linearised response. It need not parametrise the translation-invariant
vacuum field itself. Let $\mathcal A_{0,k}(z)$ be an analytic
boundedly invertible preconditioner on the same space and set
\begin{align}
 \mathfrak F_k(\varphi;h)
 &:=\varphi-\mathcal A_{0,k}^{-1}\mathcal E_k(\varphi,h),
 &\mathcal B_k&:=D_\varphi\mathfrak F_k
 =I-\mathcal A_{0,k}^{-1}\mathcal A_k,
 \notag\\
 \mathsf S_k&:=-\mathcal A_{0,k}^{-1}\Gamma_{k,\varphi h},
 &\mathsf R_k&:=-\Gamma_{k,h\varphi},
 &C_{k,\mathrm{dir}}&:=-\Gamma_{k,hh}.
 \label{eq:canonical-response-tuple}
\end{align}
Then an infinitesimal physical tensor-source variation obeys
\begin{equation}
 \delta\varphi=(I-\mathcal B_k)^{-1}\mathsf S_k\,\delta h,
 \qquad
 W_{k,hh}=C_{k,\mathrm{dir}}
 +\mathsf R_k(I-\mathcal B_k)^{-1}\mathsf S_k.
 \label{eq:canonical-response-resolvent}
\end{equation}
Thus the same source-dependent effective action fixes the kernel and both coupling maps as well as the direct term. Their response reproduces the Schur identity~\eqref{eq:source-schur-identity} exactly.

The preconditioner $\mathcal A_{0,k}$ is a bookkeeping isomorphism on the
chosen response space. Its choice changes the factorisation while preserving the measured response. At a singular stationarity Jacobian continuation from the
noncritical domain uses an analytic response family satisfying the
analytic Fredholm hypotheses. This condition belongs to the response
family and is independent of the preconditioner. For a finite vertex
ansatz $\mathcal E_{k,N}$ and $\mathcal O_{k,N}$ give the response of
the corresponding closure. Comparison with $W_{k,hh}$ uses common-space
operator and source-map bounds. Direct differentiation of
Eq.~\eqref{eq:canonical-stationarity-map} at fixed vanishing fundamental
source gives
$\mathcal A_k\,\delta\varphi=-\Gamma_{k,\varphi h}\delta h$.
Since
$I-\mathcal B_k=\mathcal A_{0,k}^{-1}\mathcal A_k$ and the
Jacobian is invertible on the stated noncritical domain the first identity in
Eq.~\eqref{eq:canonical-response-resolvent} follows. Substitution into the
total $h$ derivative of $W_{k,h}=-\Gamma_{k,h}$ gives its second identity
which is Eq.~\eqref{eq:source-schur-identity}. For the extended variables
Proposition~\ref{prop:native-same-source} supplies the general identity.
Block elimination yields the reduced response with its source and direct
terms given in Appendix~\ref{app:composite-frg}.

\paragraph{TT response space and complementary channels.}
\label{sec:tt-response-space}

A vertex expansion contains physical TT insertions together with trace and longitudinal components. BRST and EOM representatives and auxiliary variables complete the response space. Let
$P$ be a bounded projection of the response space $\cX$ onto its TT block and
$\overline P=I-P$. The two blocks reduce the full response operator
$\cB$ precisely when
\begin{equation}
 \overline P\cB P=P\cB\overline P=0.
 \label{eq:tt-response-closure}
\end{equation}
Whenever $I-\overline P\cB\overline P$ is invertible the exact projected resolvent
uses the Schur--Feshbach kernel \cite{Feshbach1958}
\begin{equation}
 \cB_{\TT}^{\rm eff}
 =P\cB P
 +P\cB\overline P
  (I-\overline P\cB\overline P)^{-1}
  \overline P\cB P.
 \label{eq:tt-feshbach-kernel}
\end{equation}
Retaining only $P\cB P$ loses the second term whenever it is nonzero.
One-sided invariance can make that term vanish without reducing both blocks.
The source and observable maps together with the direct term reduce at the same time as
given explicitly in Eqs.~\eqref{eq:app-complete-feshbach-tuple} and
\eqref{eq:app-complete-reduced-response}. Neither the source nor the
observable map is assumed to vanish on the $\overline P$ block. In this reduction
the effective maps $\mathsf S_{\mathrm{eff}}$ and
$\mathsf R_{\mathrm{eff}}$ act on tensor sources before contraction with
a physical TT polarisation.

The reduced response acts between the source spaces as
\begin{equation}
 \mathcal J_{\TT}
 \xrightarrow{\ \mathsf S_{\TT}\ }\cX_{\TT}
 \xrightarrow{\ (I-\cB_{\TT}^{\rm eff})^{-1}\ }\cX_{\TT}
 \xrightarrow{\ \mathsf R_{\TT}\ }\mathcal J_{\TT}^{*}.
 \label{eq:typed-tt-response-chain}
\end{equation}
The first map realises the stress-tensor source in the chosen response
coordinates. The last map reconstructs the measured TT correlator.
Both maps transform together under a change of operator
normalisation so that the measured quadratic form is preserved.
The singularities of this measured response determine the visible critical modes.

\section{Visible critical modes and spectral residues}
\label{sec:atom-kernel}

The compact spectral ladder has a simple internal crossing at which a collective direction becomes soft. This enhanced susceptibility produces spatial propagation when the same sourced equations supply the corresponding external-momentum dependence. The derivative of the inverse response with respect to the external invariant fixes the pole strength. The source and readout overlaps determine whether matter can excite and observe the critical direction.
For the compact spectral model of Ref.~\cite{Guo2026CoarseGrainingFramework}
consider the mass-shift deformation $t_\chi=e^{-2\chi}$ for $\chi\geq0$.
We use the trace units
$L_{\rm tr}=c/\mu_*$ and $k_{\rm win}=\mu_*\sqrt\pi/c$ of
Section~\ref{sec:native-tt-spectral-branch}.  The radius coordinate $c$
and the closure coefficient $\kappa_{\rm TT}=4/27$ are dimensionless. The ladder below depends on their dimensionless spectral data and is
independent of the reference mass unit $\mu_*$.
The deformation holds the matter eigenvalues and volume fixed together with the Gaussian window. It replaces each $\widehat m_j^2$ by $t_\chi\widehat m_j^2$ and scales
the two entries $8\mu_*^2/c^2$ and $6\mu_*^2/c^2$ of the TT vertex by
$t_\chi$.
Thus it is a specified spectral deformation distinct from a common
radius rescaling of every operator. With the data of
Eq.~\eqref{eq:native-tt-coefficient} define
\begin{align}
 \mathcal J(\chi)&=\frac1{\pi^3}\sum_jd_jw_j
 \int_0^{B_j}\frac{X_j(u)^2e^{-X_j(u)/\pi}}
 {[X_j(u)+t_\chi\widehat m_j^2]^2}\,\dd u,\notag\\
 D(t)&=8t\left[\frac{8t}{1-e^{-8t/\pi}}+6t\right],\qquad
 x(\chi)=\frac{\mathcal J(\chi)}{D(t_\chi)}.
 \label{eq:native-chi-spectral-ladder}
\end{align}
The frequency substitution $u=c\omega/\mu_*$ gives
$\Pi_\chi=\mu_*^4c^{-4}\mathcal J(\chi)$ and
$v_{\rm TT,\chi}=\mu_*^{-4}c^4/D(t_\chi)$. Their product
$x=v_{\rm TT,\chi}\Pi_\chi$
is therefore independent of $c$ while
$\mathcal J(0)=32A_{n_{\rm sp},R}$.

Its critical point is simple. Every summand of $\mathcal J$ is
nondecreasing in $\chi$ and its limit is positive. For $t>0$ the function
$8t/(1-e^{-8t/\pi})$ is strictly increasing since $e^y>1+y$.
Consequently $D'(t)>0$ and $x'(\chi)>0$ and
$D(t_\chi)\sim8\pi t_\chi$ as $\chi\to\infty$. Equivalently $t_\chi\downarrow0$ and $x(\chi)\to\infty$.
For the eight rows of Appendix~\ref{sec:app-compact-spectral-branch}
$\sum_jd_jw_j=2080$. Bounding each rational kernel by one and
$\int_0^\infty e^{-u^2/\pi}\dd u
=\sqrt\pi\int_0^\infty e^{-v^2}\dd v=\pi/2$ gives
$x(0)<65/(7\pi^2)<1$ where $t_0=1$. There is exactly one $\chi_c>0$ with
$x(\chi_c)=1$. At $\vartheta=1/50$ independent quadrature of these
same rows gives $x(0)\simeq0.1007$ and $\chi_c\simeq0.6572$ with $x'(\chi_c)\simeq3.058$.
The preceding inequalities establish existence and simplicity
independently of the quadrature.

The scalar ladder response $\Pi_\chi/(1-x(\chi))$ therefore obeys
\begin{equation}
 \lim_{\delta\downarrow0}\delta
 \frac{\Pi_{\chi_c-\delta}}{1-x(\chi_c-\delta)}
 =\frac{\mu_*^4\mathcal J(\chi_c)}{c^4x'(\chi_c)}>0.
 \label{eq:native-chi-detuning-residue}
\end{equation}
This pole is in the internal detuning $\delta=\chi_c-\chi$.
The radius endpoint fixes $c$. The external variable enters the sourced propagation response whose pole determines spatial exchange.
Appendix~\ref{sec:app-native-scale-identity} records the corresponding
scale derivatives at the radius endpoint.

The crossing $x(\chi_c)=1$ marks unit feedback in the complete ladder and
mixing kernel so a collective direction appears in the reduced operator.
The sourced equations continue this direction to an external branch whose
nonzero $z$-slope fixes the propagating residue. The Fredholm formulation
states this mechanism independently of a diagrammatic grouping.

\subsection{Visible tensor modes of the compact source}
\label{sec:native-compact-tt-visibility}

There is an explicit TT projection of the scalar source in
Proposition~\ref{prop:native-geometric-source}. Identify the round
quotient with $SO(3)$ and write its adjoint matrix as $R_{ib}(U)$.
The nine real modes
$u_{ib}=(3/V_3)^{1/2}R_{ib}$ are orthonormal and have scalar
eigenvalue $8/L^2$. For a real symmetric trace-free matrix $q$ form
$h_q^+=V_3^{-1/2}q_{ab}\theta^a\theta^b$ in the left-invariant
orthonormal coframe and $h_q^-$ in the right-invariant coframe.
Both are transverse and traceless. They span orthogonal
five-dimensional subspaces with rough tensor eigenvalue $6/L^2$.

The scalar operator source and its energy source on this shell are
\begin{align}
 H[h_q^+]&=\frac4{L^2\sqrt{V_3}}I_3\otimes q,&
 H[h_q^-]&=\frac4{L^2\sqrt{V_3}}q\otimes I_3,\notag\\
 D H_{\rm sc}[h_q^+]&=\alpha_2 I_3\otimes q,&
 D H_{\rm sc}[h_q^-]&=\alpha_2 q\otimes I_3,\qquad
 \alpha_2=\frac2{M_2L^2\sqrt{V_3}},
 \label{eq:native-compact-tt-vertices}
\end{align}
where $H_{\rm sc}=(-\Delta_g+\mathcal R_g/8)^{1/2}$ and
$M_2=\sqrt{35}/(2L)$. The curvature variation vanishes on these TT
directions. All source factors therefore follow from the same scalar
operator and its common source normalisation.
Appendix~\ref{sec:app-compact-tt-visibility} proves the normalisation
transversality and the two matrix formulas.

For a real symmetric occupation matrix $\mathsf N_{\rm occ}$ on this
nine-dimensional shell let $\mathsf N_{{\rm occ},L}$ and
$\mathsf N_{{\rm occ},R}$ be its partial traces. Its ten TT source
components are the matrix pair
\begin{equation}
 j_{\rm hom}(\mathsf N_{\rm occ})
 =\alpha_2\big((\mathsf N_{{\rm occ},L})_0,
                    (\mathsf N_{{\rm occ},R})_0\big),\qquad
 B_0=B-\tfrac13\operatorname{tr}(B)I_3.
 \label{eq:native-compact-tt-source}
\end{equation}
They span both tensor subspaces so this source is visible
on a ten-dimensional TT space. An isotropic occupation
$\mathsf N_{\rm occ}=I_9/9$ has zero TT projection despite positive
energy $M_2$.
The constant scalar mode also has positive energy and zero TT
projection. These examples show why a static matter response must retain
the trace and constraint sectors.

The $(2,1)$ one-form band used in the compact normalisation also has
an explicit source on this tensor space. In the right-invariant frame
let $\mathcal Q=\operatorname{Sym}_0(\mathbb R^3)$ and set
$a_{q,v}=(3/V_3)^{1/2}(qRv)_a\theta_R^a$ for $q\in\mathcal Q$ and $v\in\mathbb R^3$. These forms span the $(2,1)$ band with Hodge
eigenvalue $16/L^2$. On its identification with
$\mathcal Q\otimes\mathbb R^3$ the compressed source for
$h_t^-=V_3^{-1/2}t_{ab}\theta_R^a\theta_R^b$ is
\begin{equation}
 H_{\rm lr}[h_t^-]=\frac{16}{L^2\sqrt{V_3}}
       \mathcal A_t\otimes I_3,\qquad
 \mathcal A_tq=(tq+qt)_0 .
 \label{eq:native-long-root-tt-source}
\end{equation}
The source is obtained by differentiating the one-form spectral
operator and its norm with respect to the tensor perturbation. Appendix~\ref{sec:app-native-long-root-source}
derives it and its nonzero source on each real decomposable mode.
The equal occupation of the full band again has zero TT source.

These pairings determine the source on compact tensor eigenspaces
whose rough eigenvalue is positive at finite $L$. Projection onto the critical space of the sourced response determines their propagation residue.

\subsection{A finite-dimensional critical space}
\label{sec:native-critical-space}

The residue follows from the Schur complement of an isolated critical
subspace as in analytic operator perturbation theory~\cite{Kato1995}.
Keeping the source maps in this reduction identifies exactly which
critical directions contribute to the measured response.

\begin{theorem}[Positive residue on a finite critical space]
\label{thm:native-finite-rank-residue}
\label{prop:native-rational-residue}
Let $L(z)$ be bounded and analytic on a Hilbert space $\mathcal X$ near
zero with $L(\overline z)^*=L(z)$. Suppose that $L(0)$ has a
finite-dimensional kernel with orthogonal projection $P$ and that its restriction to $\overline P\mathcal X$ with $\overline P=I-P$ is boundedly invertible. Assume
\begin{equation}
 G=P L'(0)P\big|_{P\mathcal X}>0.
 \label{eq:native-critical-slope}
\end{equation}
Then $L(z)$ is invertible for sufficiently small nonzero $z$ and
\begin{equation}
 L(z)^{-1}=\frac{P G^{-1}P}{z}+O(1).
 \label{eq:native-critical-inverse}
\end{equation}
For a Hilbert source space $\mathcal U$ if $\mathsf S(z):\mathcal U\to\mathcal X$ and
$\mathsf R(z):\mathcal X\to\mathcal U$ are analytic and satisfy $\mathsf R(x)=\mathsf S(x)^*$ for real $x$ and if $C_{\rm dir}$ is analytic
then the residue of $C=C_{\rm dir}+\mathsf R L^{-1}\mathsf S$ is
\begin{equation}
 Z_{\rm nat}=\mathsf S(0)^*P G^{-1}P\mathsf S(0)\succeq0,
 \qquad \operatorname{rank}Z_{\rm nat}
        =\operatorname{rank}(P\mathsf S(0)).
 \label{eq:native-finite-rank-residue}
\end{equation}
The finite rational case is the same Schur calculation in coordinates.
Let $H,K$ be real symmetric rational matrices with $H\succeq0$. Let $S$ be a rational source and let $U$ be any basis of $\ker H$. Put
$G_U=U^{\mathsf T}KU$.  If $G_U\succ0$ then
\begin{equation}
 S^{\mathsf T}(H+zK)^{-1}S
 =\frac{Z_{\rm crit}}{z}+O(1),\qquad
 Z_{\rm crit}=S^{\mathsf T}U G_U^{-1}U^{\mathsf T}S\succeq0,
 \label{eq:native-rational-residue}
\end{equation}
with
$\operatorname{rank}Z_{\rm crit}=\operatorname{rank}(U^{\mathsf T}S)$.
If $W_H$ contains a maximal independent set of columns of $H$ the
complementary gap obeys
\begin{equation}
 \lambda_{\min}(H|_{(\ker H)^\perp})\geq
 \left\{\operatorname{Tr}\left[
 (W_H^{\mathsf T}HW_H)^{-1}W_H^{\mathsf T}W_H\right]\right\}^{-1}>0,
 \label{eq:native-rational-complement-gap}
\end{equation}
with the empty-kernel and $H=0$ cases interpreted separately.
\end{theorem}

\begin{proof}
On $P\mathcal X\oplus\overline P\mathcal X$ the complementary block is
analytically invertible and the Schur complement is $zG+O(z^2)$
\cite{Kato1995}.  Hence only the $P$ block has a $z^{-1}$ term.
Inserting the source and readout gives the displayed residue whose
factorisation through $G^{-1/2}P\mathsf S(0)$ proves positivity and rank.
The rational formulas are the same calculation in a nonorthonormal kernel
basis. The last inequality is the trace bound for the inverse positive
complementary matrix.
\end{proof}

Equation~\eqref{eq:native-finite-rank-residue} separates three physical
ingredients. The operator $G$ is the propagation slope on the critical space. The projected map $P\mathsf S(0)$ carries the applied source into that space. Its adjoint gives the readout. A zero overlap makes
an algebraic critical mode invisible whereas positivity of $G$ makes the
visible residue a positive quadratic form.

For rational matrices elimination and positive Schur pivots determine the kernel and residue together with the gap bound exactly~\cite{Guo2026CoarseGrainingFramework}.
The kernel identity fixes the critical
space while positive eigenvalues contribute to the regular response. For $\epsilon>0$ one has $\lim_{z\downarrow0}z/(\epsilon+z)=0$.
If $K\succ0$ on the full space the positive-form construction in
Eq.~\eqref{eq:native-positive-resolvent-measure} also gives a positive
measure for the same finite affine response.

Reciprocity makes the slope Hermitian but does not make it positive.
For example $L(z)=-zI$ has a negative residue despite reciprocity.
Equation~\eqref{eq:native-critical-slope} must be calculated from the
propagation rule. The projection $P$ is taken after the physical
constraints have been handled. A gauge or normalisation direction is
not made observable by this theorem. For an unbounded form realisation
the bounded result applies only after a justified form-space reduction.
Eq.~\eqref{eq:native-positive-resolvent-measure} gives the corresponding
exact affine construction.

Appendix~\ref{sec:app-native-finite-rank-example} evaluates the residue
for explicit finite matrices including a critical direction orthogonal
to every source. Further finite calculations in
Ref.~\cite{Guo2026CoarseGrainingFramework} retain the direct term and the
source-dependent complementary contribution. For approximations to an
infinite-dimensional response the operator and source convergence
estimates of Appendix~\ref{app:fredholm} control the limiting critical
space and residue.

For matter source columns $\mathsf S_a(0)$ the critical mixed residue is
the Gram matrix of $G^{-1/2}P\mathsf S_a(0)$ whose rank is the dimension
of the spanned critical source space. If an energy-source
identity gives $P\mathsf S_a(0)=q_a^{\rm src}v_0$ with the same $v_0$
for every allowed matter state it instead gives
\begin{equation}
 Z_{ab}=q_a^{\rm src}q_b^{\rm src}
       \langle v_0,G^{-1}v_0\rangle
 \label{eq:native-common-energy-residue}
\end{equation}
for real energy charges. For stationary conserved matter Laue's
identity derives this common direction from the tensor vertex and the
reduced source map.

\subsection{Response factorisation}
\label{sec:tt-response-factorisation}

On a Euclidean ray $Q=\sqrt z\,\widehat Q$ with $z>0$ choose a real unit
TT polarisation $e$ transverse to the fixed direction $\widehat Q$.
In a covariant renormalisation scheme rotational covariance gives
$e^*G^E(Q)e=C_2(z)$ independently of this choice of $e$.
Thus contracting the effective source and observable maps gives
$\mathsf S_{\TT}=\mathsf S_{\mathrm{eff}}e$ and
$\mathsf R_{\TT}=e^*\mathsf R_{\mathrm{eff}}$. The direct term is
contracted in the same way. In this section these are scalar-channel maps
$\mathbb C\to\cX_{\TT}$ and $\cX_{\TT}\to\mathbb C$.
Their analytic continuation near $z=0$ is a hypothesis of the theorem
below. If a truncation breaks rotational
symmetry polarisation dependence must be retained until restoration is
controlled. For several source amplitudes the response is instead a
matrix and the norm estimates of Appendix~\ref{app:fredholm} apply to it.

The polarisation multiplicity also matters for the kernel theorem.
If exact rotational covariance identifies the spin-two response block as
$V_2\otimes\mathcal X_{\rm mult}$ with
$\cB=I_{V_2}\otimes\cB_{\rm red}$ restriction to a fixed polarisation
leaves $\cB_{\rm red}$ on the multiplicity space. Here $\dim V_2=5$
at nonzero Euclidean momentum. Simplicity refers to the eigenvalue of this reduced family on the multiplicity space.
This reduction follows from the intertwining property and Schur's lemma. It must be verified in the chosen response representation. When it is
unavailable the full critical multiplicity is retained and a finite-rank
analysis replaces the simple-eigenvalue theorem.

Let $\cX_{\TT}$ be the response space selected either by the closed TT block
or by the Schur--Feshbach reduction in
Eq.~\eqref{eq:tt-feshbach-kernel} with any justified polarisation reduction
made as above. Near a source-free solution $X_*$
consider
\begin{equation}
 X=\mathfrak F(X;z,\eta)+\mathsf S_{\TT}(z;\eta)J,
 \qquad
 \cB_{\TT}(z;\eta):=D_X\mathfrak F(X_*;z,\eta).
 \label{eq:tt-self-consistent-response}
\end{equation}
Here $J$ is the scalar amplitude of an infinitesimal background-source
variation along $e$. The fundamental-field source is $j$ in
Eq.~\eqref{eq:paper1-composite-generating-functional}. Equation~\eqref{eq:tt-self-consistent-response}
is the linear-source form of the construction in
Proposition~\ref{prop:composite-source-kernel-bridge}. Its kernel and direct term $\mathcal C_{\rm reg}$ come from the same source functional as $\mathsf S_{\TT}$ and $\mathsf R_{\TT}$.
The parameter $z$ is the complexified external invariant with $z=Q^2$ on
the Euclidean axis and $z=-p^2-\iu0$ at the Lorentzian boundary. The symbol
$\eta$ denotes a control parameter which may include couplings or an RG
scale. Differentiation with respect to $z$ is not RG evolution.
Only the first derivative at $J=0$ is used. Terms beyond linear response
are suppressed in Eq.~\eqref{eq:tt-self-consistent-response}.

If a measured response has regular direct term
$\mathcal C_{\rm reg}$ and observable map $\mathsf R_{\TT}$
differentiation of
Eq.~\eqref{eq:tt-self-consistent-response} with respect to $J$ gives
\begin{equation}
 \mathcal C_{\rm resp}(z;\eta)=\mathcal C_{\rm reg}(z;\eta)
 +\mathsf R_{\TT}(z;\eta)
  [I-\cB_{\TT}(z;\eta)]^{-1}
  \mathsf S_{\TT}(z;\eta).
 \label{eq:tt-response-resolvent}
\end{equation}
The symbol $\mathcal C_{\rm resp}$ denotes the response produced by the
specified physical source and readout. It equals the stress coefficient $C_2$ only
after those maps have been matched to the stress observable of
Section~\ref{sec:covariant-stress-realisation}.

This distinction is already visible in a reduced two-helicity branch.
Let $\mathcal H_{\rm hel}\simeq\mathbb C^2$ and suppose its field
covariance has
\[
 \mathcal D_{\rm hel}(z)=\frac{Z^{-1}I_2}{z}
      +\mathcal D_{\rm hel,reg}(z),\qquad Z>0.
\]
Dimensionful maps
$\mathsf S_T(z):\mathcal U_T\to\mathcal H_{\rm hel}$ and
$\mathsf R_T(z):\mathcal H_{\rm hel}\to\mathcal U_T$ match a physical TT
stress source space $\mathcal U_T$ to this field branch. Let
$\sigma_T$ be the positive scalarisation fixed by the same TT
polarisation convention as $C_2$ (an invariant normalised polarisation
trace in the finite-dimensional case). The operator-valued response and
its atomic coefficient are
\begin{align}
 \mathbf C_T^{\rm nl}(z)
 &=\mathsf R_T(z)\mathcal D_{\rm hel}(z)\mathsf S_T(z)
   +\mathbf C_{T,\rm reg}^{\rm nl}(z),\notag\\
 \mathbf Z_T
 &=\mathsf R_T(0)\,Z^{-1}I_2\,\mathsf S_T(0),\notag\\
 C_{2,\rm red}^{\rm nl}(z)
   &=\sigma_T[\mathbf C_T^{\rm nl}(z)],\qquad
 Z_0^{(\rm red)}=\sigma_T[\mathbf Z_T].
 \label{eq:stress-helicity-dimensional-bridge}
\end{align}
In the dimensionless-field convention used for the effective realisation
below $Z$ has mass dimension two. The product of
$\mathsf R_T$ and $\mathsf S_T$ carries the remaining dimensions required
for $C_{2,\rm red}$ and $Z_0^{(\rm red)}$ to have mass dimensions four
and six
respectively.  Under reciprocity $\mathsf R_T=\mathsf S_T^*$ on the real
axis so $\mathbf Z_T$ is positive semidefinite and the positive functional
$\sigma_T$ gives $Z_0^{(\rm red)}\geq0$. Attaching two matter
vertices of strength $g_\Phi$ gives the exchange weight $g_\Phi^2/Z$.  Equation~\eqref{eq:stress-helicity-dimensional-bridge} gives the stress-response atom whereas $g_\Phi^2/Z$ is the matter-exchange
weight and is invariant under a renormalisation of the interpolating
field.
When the compact branch is matched to this stress response on the same
scalar TT quadratic form the normalisations obey
\begin{equation}
 \frac{M_P^2}{\widehat Z}
 =Z_{\rm prop}=Z_0^{(\rm red)}
 =\sigma_T\!\left[
   \mathsf R_T(0)\,Z^{-1}I_2\,\mathsf S_T(0)\right].
 \label{eq:ladder-stress-residue-match}
\end{equation}
\begin{proposition}[Critical branch and propagation residue]
\label{prop:native-critical-propagation}
\label{prop:external-matched-branch}
Fix the radius at
$c_*=A_{n_{\rm sp},R}(\vartheta)/\kappa_{\rm TT}$ and determine
$\chi_c$ by $x(\chi_c)=1$. Suppose the sourced propagation equations
give analytic functions $\chi(z)$ and $\Pi_{\rm ext}(z)$ together with $C_{\rm dir}(z)$ near $z=0$ with $z=Q^2$ and
\begin{equation}
 \begin{aligned}
  \chi(0)&=\chi_c,&
  a&=-\left.\frac{\dd\chi}{\dd z}\right|_{z=0}>0,\\
  \Pi_{\rm ext}(0)&=\frac{\mu_*^4\mathcal J(\chi_c)}{c_*^4},&
  C_{\rm red}(z)&=C_{\rm dir}(z)+\frac{\Pi_{\rm ext}(z)}{1-x(\chi(z))}.
 \end{aligned}
 \label{eq:native-critical-propagation-data}
\end{equation}
Then the external scalar response has the positive residue
\begin{equation}
 C_{\rm red}(Q^2)=\frac{Z_{\rm prop}}{Q^2}+O(1),\qquad
 Z_{\rm prop}=\frac{\mu_*^4\mathcal J(\chi_c)}
 {c_*^4\,a\,x'(\chi_c)}>0.
 \label{eq:native-critical-propagation-residue}
\end{equation}
The same branch can be normalised directly against the massless tensor
inverse. Let $M_P>0$ be the reference mass used to make the external
invariant dimensionless. The effective tensor realisation fixes this
reference through its kinetic coefficient. With
$\Pi(\chi)=\mu_*^4\mathcal J(\chi)/c_*^4$ and
$\widehat z=Q^2/M_P^2$ set
\begin{equation}
 \ell(\chi):=\frac{1-x(\chi)}{\Pi(\chi)},\qquad
 \ell(\chi(\widehat z))=\widehat Z\widehat z,
 \qquad \widehat Z>0.
 \label{eq:external-ladder-matching}
\end{equation}
The branch through $\chi_c$ is real analytic and unique and
\begin{equation}
 \chi(0)=\chi_c,\qquad
 -\left.\frac{\dd\chi}{\dd\widehat z}\right|_{\widehat z=0}
   =\frac{\widehat Z\Pi(\chi_c)}{x'(\chi_c)}>0,\qquad
 \frac{\Pi(\chi(\widehat z))}{1-x(\chi(\widehat z))}
       =\frac1{\widehat Z\widehat z}.
 \label{eq:external-matched-slope}
\end{equation}
Here $\widehat z$ is dimensionless.  In the stress-response
normalisation $C_{\rm lad}:=\Pi/(1-x)$ has mass dimension four so
$[\widehat Z]=-4$ and its $Q^2$-pole residue is
\begin{equation}
 Z_{\rm prop}=\frac{M_P^2}{\widehat Z}.
 \label{eq:ladder-dimensional-residue}
\end{equation}
\end{proposition}

\begin{proof}
The preceding crossing calculation fixes the unique $\chi_c$ and
$x'(\chi_c)>0$. The propagation branch has
$\chi(z)=\chi_c-az+O(z^2)$ so
$1-x(\chi(z))=a x'(\chi_c)z+O(z^2)$. Expanding the numerator
and the analytic direct term proves the stated residue.  Moreover
$\ell'(\chi_c)=-x'(\chi_c)/\Pi(\chi_c)<0$ so the analytic
inverse-function theorem gives the matched branch and its derivative.
\end{proof}

Equation~\eqref{eq:native-critical-propagation-data} specifies the
matching to be evaluated from the propagation equations. Its derivative
$a$ uses the external invariant at fixed radius and matter content.
For a matrix-valued tensor response the corresponding calculation uses the physical critical projection and momentum derivative together with the source overlap on the critical subspace. Its spectral interpretation additionally
requires a positive measure for the same sourced response.

Thus the relevant kernel is the dimensionless Jacobian of the complete
self-consistent map. The Bethe--Salpeter bound-state equation
\cite{Salpeter1951} provides a two-particle representation. In the
two-particle-irreducible effective-action formulation
\cite{CornwallJackiwTomboulis1974} differentiation of the stationarity
equation yields the corresponding response kernel. Its product with
the two-particle propagator represents this Jacobian when the sourced
equations and observable map are retained consistently.

For an $N$-dimensional calculation with response space
$\cX_{\TT,N}$ we require a bounded comparison embedding
$\iota_N:\cX_{\TT,N}\to\cX_{\TT}$ and projection
$P_N:\cX_{\TT}\to\cX_{\TT,N}$ satisfying
$P_N\iota_N=I_{\cX_{\TT,N}}$. Its response is determined by
\begin{equation}
 \mathfrak T_N=
 \big(\cB_{\TT,N},\mathsf S_{\TT,N},\mathsf R_{\TT,N},
       \mathcal C_{\mathrm{reg},N}\big).
 \label{eq:declared-tt-truncation}
\end{equation}
The renormalised operator basis $\{\cO_a^{(2)}\}_{a=1}^N$ together with the effective action $\Gamma_{k,N}$ and regulator $R_k$ define the projected
stationarity equations $\mathcal E_{k,N}$ and observable
$\mathcal O_{k,N}$. Proposition~\ref{prop:composite-source-kernel-bridge}
then determines the kernel and direct term together with the source and observable maps
with preconditioner $\mathcal A_{0,k,N}$ and subtraction polynomial
$P_{2,N}$. The comparison maps place these finite operators in the
common space of Appendix~\ref{app:fredholm} where their approximation
errors are measured.

\subsection{Internal averaging and closure defect}
\label{sec:tt-closure-defect}

External smearing changes the source and observable maps in
Eq.~\eqref{eq:tt-response-resolvent} but not the spectrum of a fixed
$\cB_{\TT}$. An internally averaged response can change the kernel only
through a specified coarse-grained self-consistent map. Let
$\mathcal C_\tau:\cX_{\TT}\to\cX_{\TT,\tau}$ and define
\begin{equation}
 E_\tau(X)
 :=\mathcal C_\tau\mathfrak F(X)
  -\mathfrak F_\tau(\mathcal C_\tau X).
 \label{eq:paper1-closure-defect}
\end{equation}
Write $\overline Y_*(z)=\mathcal C_\tau X_*(z)$ and
$H_\tau(Y;z)=Y-\mathfrak F_\tau(Y;z)$. Since $X_*$ is a source-free
solution its transported point has residual
\begin{equation}
 H_\tau(\overline Y_*;z)=E_\tau(X_*;z).
 \label{eq:paper1-transported-branch-residual}
\end{equation}
Thus the transported point is an actual coarse solution only when this
residual vanishes. If $E_\tau=0$ in a neighbourhood of $X_*$
differentiation gives the exact
intertwining relation
\begin{equation}
 \mathcal C_\tau\cB(X_*)
 =\cB_\tau(\mathcal C_\tau X_*)\mathcal C_\tau.
 \label{eq:paper1-kernel-intertwining}
\end{equation}
For an approximate closure let $Y_{\tau,*}$ be an established
coarse solution and let
$\delta_{Y,\tau}:=\|Y_{\tau,*}-\overline Y_*\|$ have a verified upper
bound. Suppose that $D_Y\mathfrak F_\tau$ is Lipschitz with constant
$L_\tau$ on a ball containing both points. On the specified invariant response
comparison spaces where $\mathcal C_\tau$ is a bounded isomorphism the
actual response Jacobian $\cB_\tau(Y_{\tau,*})$ then obeys
\begin{equation}
 \begin{aligned}
 &\|\cB_\tau(Y_{\tau,*})
       -\mathcal C_\tau\cB(X_*)\mathcal C_\tau^{-1}\|\\
 &\qquad\leq
 \|D_XE_\tau(X_*)\|\,\|\mathcal C_\tau^{-1}\|
       +L_\tau\delta_{Y,\tau}.
 \end{aligned}
 \label{eq:paper1-kernel-defect-bound}
\end{equation}
Indeed differentiation of Eq.~\eqref{eq:paper1-closure-defect} compares
Jacobians at $X_*$ and $\overline Y_*$. Adding and subtracting
$\cB_\tau(\overline Y_*)$ gives the additional displacement term. Exact
reference-branch matching $E_\tau(X_*;z)=0$ permits the choice
$Y_{\tau,*}=\overline Y_*$ and sets this term to zero.
Proposition~\ref{prop:app-actual-branch-certificate} supplies one sufficient
test for an approximate branch. At a singular stationarity point exact matching or a separate critical-branch construction supplies the branch. The inverse norm also
makes explicit why strong smoothing can amplify a small closure error.
Appendix~\ref{app:fredholm} includes the resulting root and residue bounds.

\subsection{An isolated Fredholm pole}
\label{sec:isolated-kernel-certificate}

\begin{theorem}[Simple pole of the TT response kernel]
\label{thm:tt-kernel-certificate}
At a fixed physical endpoint $\eta=\eta_c$ let $U$ be a connected
neighbourhood of $z=0$ on which $\cB_{\TT}(z)$ is compact and analytic with $I-\cB_{\TT}(z)$ invertible somewhere in $U$. Suppose $U$ is free of multiparticle thresholds and other nonmeromorphic singularities.
Assume that the direct term and both source and readout maps in
Eq.~\eqref{eq:tt-response-resolvent} are analytic on $U$. Then
\begin{enumerate}[label=(\roman*)]
 \item if $1\notin\spec\cB_{\TT}(0)$ this kernel does not generate a pole
 at $z=0$.
 \item if an algebraically simple eigenvalue satisfies
 \begin{equation}
  \lambda(0)=1,
  \qquad \lambda'(0)\neq0,
  \label{eq:tt-simple-crossing}
 \end{equation}
 then $(I-\cB_{\TT})^{-1}$ has a simple pole.
 \item with right and left eigenvectors normalised by
 $\langle\ell|r\rangle=1$ that pole is visible in
 $\mathcal C_{\rm resp}$ if and only if
 \begin{equation}
  \mathsf R_{\TT}(0)|r\rangle\neq0,
  \qquad
  \langle\ell|\mathsf S_{\TT}(0)\neq0;
  \label{eq:tt-overlap-conditions}
 \end{equation}
 \item its algebraic residue is
 \begin{equation}
  Z_{\rm alg}
  =-\frac{\mathsf R_{\TT}(0)|r\rangle
          \langle\ell|\mathsf S_{\TT}(0)}{\lambda'(0)}.
  \label{eq:tt-algebraic-residue}
 \end{equation}
\end{enumerate}
\end{theorem}

The proof uses the rank-one Laurent expansion of the analytic Fredholm
resolvent~\cite{Kato1995}. Appendix~\ref{app:fredholm} gives the
expansion with both source factors. A first-order transverse crossing in the
control parameter $\eta$ has
$\partial_\eta\lambda(0;\eta_c)\neq0$. Higher-order parameter crossings
are covered by the same argument whenever the external-invariant slope
remains nonzero. For example
$\lambda(z;\eta)=1-z+\eta^3$ has $\partial_\eta\lambda(0;0)=0$ but
$\partial_z\lambda=-1$.
In an FRG realisation a finite-$k$ crossing concerns the regulated
response. The physical correlator is obtained in the renormalised
$k\to0$ limit. Its physical correlation scale remains fixed in this limit.

The remaining step is to identify the Laurent coefficient with the
weight of the physical measure. This requires the sourced response rather
than just the kernel spectrum to agree with the reconstructed correlator.

\begin{theorem}[TT Fredholm--spectral matching]
\label{thm:tt-fredholm-spectral-bridge}
In addition to Theorem~\ref{thm:tt-kernel-certificate} assume that
\begin{enumerate}[label=(\roman*)]
 \item the endpoint response $\mathcal C_{\rm resp}$ in
 Eq.~\eqref{eq:tt-response-resolvent} is the unfiltered connected
 correlator of the same physical source observable that defines the
 response equation.
 \item this observable has an independently established positive scalar
 spectral measure $\dd\rho_{\rm resp}$ on $[0,\infty)$ with atom
 $Z_0^{(\rm resp)}=\rho_{\rm resp}(\{0\})$ matched to the same physical source quadratic form with common quotient and continuation and subtraction conventions.  For the covariant stress-tensor realisation
 $\mathcal C_{\rm resp}=C_2$ and these facts are supplied by the boundary
 construction in Section~\ref{sec:tt-atom-extraction} and
 Eq.~\eqref{eq:subtracted-tt-dispersion}.  For a directly reduced
 two-helicity realisation contraction with a fixed physical helicity
 source scalarises $\mathcal C_{\rm resp}=\mathcal D_{\rm hel}$ on its
 Fock space. The maps in
 Eq.~\eqref{eq:stress-helicity-dimensional-bridge} then produce the
 corresponding stress response.
 \item the direct term and complementary resolvent contributions are
 analytic at $z=0$ after subtraction.
 \item the remaining spectral measure satisfies the hypotheses of
 Proposition~\ref{prop:tt-atom-extraction}.
\end{enumerate}
Then
\begin{equation}
 Z_{\rm alg}
 =\lim_{Q^2\downarrow0}Q^2\mathcal C_{\rm resp}^{\rm nl}(Q^2)
 =Z_0^{(\rm resp)}\geq0.
 \label{eq:tt-fredholm-spectral-identity}
\end{equation}
In particular a matched scalar Laurent coefficient is real and
nonnegative.
\end{theorem}

\begin{proof}
Theorem~\ref{thm:tt-kernel-certificate} gives
$\mathcal C_{\rm resp}^{\rm nl}(Q^2)=Z_{\rm alg}/Q^2+A(Q^2)$ with
$A$ analytic under
assumption (iii). Multiplication by $Q^2$ and passage to the origin gives
$Z_{\rm alg}$. Assumptions (i) and (ii) together with (iv) identify the same limit with
the Wightman atom through Proposition~\ref{prop:tt-atom-extraction}.
Positivity follows from the independently matched measure in assumption (ii).
\end{proof}

The theorem is the point at which an operator pole acquires a physical
spectral meaning. Analytic Fredholm theory fixes the Laurent coefficient. The independent positive-measure construction identifies the same
coefficient as spectral weight. Together they connect the critical
operator branch to the physical source measure.

For the complete covariant stress tensor the match includes Theorem~\ref{thm:null-shell-exclusion} and gives $Z_{\rm alg}^{(\Theta)}=Z_0=0$. Thus a matched exact limit has zero visible residue in this tensor class.
For another response object the positive measure and null-shell match
are supplied by that object's spectral construction.

\subsection{From compact spectral data to a propagation residue}
\label{sec:compact-spectrum-interface}

Compact internal spectra provide bases and masses together with degeneracies and window weights for a finite response calculation. In the compact
realisation the fields and their sources and observables are constructed
in the same operator algebra and state. The spectral operations of
Ref.~\cite{Guo2026Arithmetic} and the covariant transport of
Ref.~\cite{Guo2026CompactSpectrumI} fix the objects used in that
construction. A limit of their responses uses the common domains and
norm estimates in the proposition below. For specified semilinear
generators ref.~\cite{Guo2026GeneratedDynamics} controls the discarded
nonlinear source on a common time interval using a stronger spectral
bound and convergence in a weaker norm. Finite response calculations also approximate the operator and its readout. Appendix~\ref{app:fredholm} compares these objects on common response spaces. Bounds on the same sourced equations and their variations supply the kernel and source-map errors.
For example the
prescription in Ref.~\cite{Guo2026CompactSpectrumII} uses
\begin{equation}
 K_{\TT}^{\rm win}(k_{\rm win}^2,m^2)
   =\frac{k_{\rm win}^4}{(k_{\rm win}^2+m^2)^2},
 \qquad y=\frac{m^2}{k_{\rm win}^2+m^2},
 \qquad \max_{0\leq y\leq1}y(1-y)^2=\frac4{27}.
 \label{eq:compact-scale-window}
\end{equation}
Here $m$ is an internal spectral mass and $k_{\rm win}$ is the window
scale.
$z$ remains the external invariant in the response equation. The sourced
equations determine how these internal data enter the
external-momentum operator in Eq.~\eqref{eq:tt-response-resolvent}.
The following proposition states sufficient conditions for a positive
scalar physical residue from a finite approximation with a controlled error.

\begin{proposition}[Physical residue from compact-spectrum approximations]
\label{prop:compact-spectrum-interface}
Let a finite calculation supply internal eigenvalues $m_n^2$ and degeneracies $d_n$ together with scale-window factors
$K_n(k_{\rm win}^2,m_n^2)$ and suppose that an
extremum or fixed-point equation selects a model closure. Suppose the
associated response has an independently constructed positive spectral
measure whose null-shell sector carries the measured TT quadratic form.
Let these data enter the response tuple
Eq.~\eqref{eq:declared-tt-truncation} with analytic kernel and direct term together with analytic source and readout maps derived from common sourced equations. Suppose
the external variable $z$ satisfies the simple-crossing and visibility
hypotheses of Theorem~\ref{thm:tt-kernel-certificate} and evaluate each
finite residue $\mathcal Z_N$ at its displaced critical root $z_N$.
Assume further that the common-space estimates of
Appendix~\ref{app:fredholm} control errors in the basis and vertices together with momentum discretisation and complementary-channel and closure errors and identify the limiting
unfiltered correlator with its positive Wightman measure with any
zero-threshold continuum separated from the atom. If for some finite
$N$ those estimates provide
\[
 |\mathcal Z_N-Z_0^{(\rm resp)}|\leq\Delta_N,
 \qquad \operatorname{Re}\mathcal Z_N-\Delta_N>0,
\]
then $Z_0^{(\rm resp)}>0$. Here $\Delta_N$ is the controlled upper bound given by
Eq.~\eqref{eq:app-certified-cross-root-residue} assembled from the telescoping kernel bounds and the derivative and source-map errors in
Eq.~\eqref{eq:app-residue-comparison-errors} together with the projection and root-displacement bounds and the slope estimates of
Propositions~\ref{prop:app-telescoping-error-budget} and
\ref{prop:app-root-residue-certificate}. The simple-eigenvalue hypothesis
is a sufficient route. Higher critical multiplicity is treated by the
finite-rank theorem above.
\end{proposition}

\begin{proof}
The sourced response tuple and the external crossing give the finite pole
and its residue at $z_N$.
The common-space and cross-root estimates bound its difference from the
exact Laurent coefficient by $\Delta_N$.
Theorem~\ref{thm:tt-fredholm-spectral-bridge} identifies that coefficient
with the real spectral weight $Z_0^{(\rm resp)}$. Hence
$Z_0^{(\rm resp)}\geq\operatorname{Re}\mathcal Z_N
-|\mathcal Z_N-Z_0^{(\rm resp)}|
\geq\operatorname{Re}\mathcal Z_N-\Delta_N>0$.
\end{proof}

A strict lower margin separates a positive limiting residue from a vanishing sequence. For example $C_N(z)=N^{-2}/z$ has positive finite residues tending to zero. The window
extremum in Eq.~\eqref{eq:compact-scale-window} fixes an internal scale
weight. The sourced propagation equation separately determines the
external derivative and the two overlaps.

Higher critical multiplicity is possible. For example $\cB(z)=(1-z)I_2$ with
$\mathsf S=(1,1)^{\mathsf T}$ and $\mathsf R=(1,1)$ gives
$\mathsf R(I-\cB)^{-1}\mathsf S=2/z$ although its critical eigenvalue
has algebraic multiplicity two as allowed by
Theorem~\ref{thm:native-finite-rank-residue}.

\subsection{Thresholds and spectral accumulation}
\label{sec:kernel-thresholds}

If $\dd\rho_{2,c}$ has support arbitrarily close to $s=0$ the analytically
continued correlator has other spectral singularities arbitrarily close to
the candidate pole. Continuous thresholds and accumulating positive-mass
atoms both prevent an isolated meromorphic neighbourhood. For a continuous
threshold the boundary behaviour may be logarithmic or of fractional-power type or more general. When spectral accumulation removes an isolated analytic
neighbourhood the direct measure criterion still determines the atomic
contribution. A limiting-absorption or
weighted-resolvent analysis may also separate it from the remaining spectrum.

Positivity also constrains possible moment identities. A condition of
the form
\begin{equation}
 \int_0^\infty s\,\dd\rho_2(s)=0
 \label{eq:excluded-positive-superconvergence}
\end{equation}
forces a nonnegative measure to be supported at zero so it is
incompatible with a nonzero positive continuum. Signed subtracted
sum rules concern a different spectral object. Similarly Ward
identities constrain tensor structures on the timelike cone. They do not
provide spectral weight or force an eigenvalue crossing. On the null cone
full positivity together with exact covariance and conservation gives Theorem~\ref{thm:null-shell-exclusion}.

The linear realisation fixes the field covariance and stress-response atom together with the kernel exchanged between matter sources.

\subsection{A linear gauge-tensor realisation}
\label{sec:external-gauge-realization}

The external representation uses the minimal gauge-invariant
Fierz--Pauli realisation. Constraint reduction supplies two propagating
helicities while the conserved-source numerator and spatial inverse
provide the data against which the compact response is matched. Let
$\Phi_{\mu\nu}$ be a dimensionless symmetric field on Minkowski space
with signature $(+---)$ and
\begin{align}
 S_{\rm FP}&=Z\int\dd^4x\left[
 \tfrac12\partial_\lambda\Phi_{\mu\nu}\partial^\lambda\Phi^{\mu\nu}
 -\partial_\mu\Phi^{\mu\nu}\partial^\rho\Phi_{\rho\nu}
 +\partial_\mu\Phi^{\mu\nu}\partial_\nu\Phi
 -\tfrac12\partial_\lambda\Phi\partial^\lambda\Phi\right],\notag\\
 S_J&=g_\Phi\int\dd^4x\,\Phi_{\mu\nu}T^{\mu\nu},\qquad
 \Phi=\eta^{\mu\nu}\Phi_{\mu\nu},\quad Z>0.
 \label{eq:external-fp-action}
\end{align}
The effective action is local with two derivatives and the linear gauge
symmetry $\delta\Phi_{\mu\nu}=\partial_\mu\xi_\nu+\partial_\nu\xi_\mu$.
The background is the source-free flat vacuum. The coefficient $Z$ and
the common matter vertex $g_\Phi$ specify the normalisation. This is the
linear spin-two action underlying the self-coupling construction of
Ref.~\cite{Deser1970}.

For Eq.~\eqref{eq:external-fp-action} the standard constraint reduction
with a conserved source gives two physical helicities.  Three residues
must be distinguished.  The reduced field covariance and its spectral measure together with the kernel exchanged between two matter sources are
\begin{equation}
 \begin{aligned}
 \mathcal L_{\rm hel}(Q^2)&=ZQ^2I_2,&
 \mathcal D_{\rm hel}(Q^2)&=\frac{1}{ZQ^2}I_2,\\
 \dd\rho_{\Phi}(s)&=\frac{1}{Z}I_2\delta_0(\dd s),&
 \mathcal K_{\rm matt}(Q^2)&=g_\Phi^2\mathcal D_{\rm hel}(Q^2).
 \end{aligned}
 \label{eq:external-gauge-residue}
\end{equation}
The field atom is positive because $Z>0$ and the matter exchange has
weight $g_\Phi^2/Z$.  The dimensionful stress interpolating maps in
Eq.~\eqref{eq:stress-helicity-dimensional-bridge} convert this field
covariance into the spectral weight of a stress--tensor two-point
function.
The complete exchange
between conserved sources in four dimensions is
\begin{equation}
 \mathcal A[T,U](Q)=\frac{g_\Phi^2}{ZQ^2}
 \left(T_{\mu\nu}^*U^{\mu\nu}-\tfrac12T^*U\right).
 \label{eq:external-conserved-exchange}
\end{equation}
Its static energy-monopole potential is
$V_{ab}(r)=-g_\Phi^2E_aE_b/(8\pi Zr)$.
These are the standard Fierz--Pauli reduction and conserved-source
inverse~\cite{Deser1970}.  The gauge kernel and four constraints leave two
TT configuration degrees of freedom on which the quadratic form is
$ZQ^2I_2$ and the Hamiltonian is positive.  The corresponding reflection
form is a Gram form giving the positive helicity measure above.  The
trace-reversed de Donder inverse supplies the $-T^*U/2$ numerator. For
static energy monopoles it is $E_aE_b/2$ and the transform of
$|\mathbf Q|^{-2}$ yields the displayed potential.

We use the field-normalisation freedom to set $g_\Phi=1$ and write the
resulting kinetic coefficient as $Z=M_P^2$. Only the invariant combination
$g_\Phi^2/Z$ enters the exchange. Comparing the static potential above with
$-G_NE_aE_b/r$ therefore gives $G_N=(8\pi M_P^2)^{-1}$.
The spectral calculation fixes the dimensionless residue while one
dimensional calibration fixes $M_P$ in physical units. The field-covariance
residue is $M_P^{-2}$. The physical field measure
in Eq.~\eqref{eq:external-gauge-residue} is on the gauge-reduced helicity
space.
The off-shell conserved-source numerator and the two on-shell helicities
are evaluated in their respective spaces.

Equation~\eqref{eq:external-ladder-matching} specifies the identification
of the ladder source normalisation with the reduced external field. The
massless tensor inverse is proportional to $Q^2$ so the relation matches
the internal inverse ladder to the leading infrared coordinate of the
external inverse propagator including its zero and normalisation.
The gauge action supplies the zero-momentum condition and propagation
coefficient. The internal crossing supplies its unique coordinate and
the conversion slope.  Write
$\overline Z=M_P^4\widehat Z$.  In the convention used for the quoted
spectral data the endpoint mass is identified with the Planck unit
$\mu_*=M_P$. Then $\overline Z=1$ gives the dimensionless slope
$-\left.\dd\chi/\dd\widehat z\right|_0\simeq0.1082$.
The matching is an input of the effective linear realisation whose
quadratic external pole is given in
Eq.~\eqref{eq:external-gauge-residue}. The propagated field is
the macroscopic collective tensor represented by this reduced branch.
The finite evaluations in Ref.~\cite{Guo2026CoarseGrainingFramework} use these gauge
identities and the reduced helicity covariance.

\section{Matter coupling and long-range response}
\label{sec:gravity-gates}

In the flat-background limit of the linear two-helicity realisation
define the matter-normalised collective field
$\widetilde\Phi_{\mu\nu}:=g_\Phi\Phi_{\mu\nu}$ from
Eq.~\eqref{eq:external-fp-action}. Its vertex is
\begin{equation}
 S_{\rm coupling}=\int\dd^4x\,T_{\rm mat}^{\mu\nu}
                    \widetilde\Phi_{\mu\nu}.
 \label{eq:native-matter-tensor-coupling}
\end{equation}
Here $T_{\rm mat}$ is the total matter energy--momentum tensor including
interaction contributions. The coupling fixes the matter vertex and the spectral
calculation determines the mediated response. The compact comparison
geometry enters through the internal spectral operators.

In the compact construction the full one-family matter module has
sixteen directions while its fifteen-dimensional charged quotient
records the internal charge response
\cite{Guo2026CompactSpectrumI}. The species trace used to fix the
endpoint is taken on that charged sector. The energy--momentum source
here is obtained by varying the complete matter action including the
neutral direction and the interactions acting on the full module.
Consequently the same tensor vertex carries component energies and
their interaction energy into the exchange amplitude.

The interaction is read out in a definite order. Conservation first
reduces the long-wavelength source to its total energy. The common response maps that source into the critical channel. The spatial inverse determines the potential and force. The helicity factorisation
and soft limits then select the gravitational tensor structure. A positive
spectral atom supplies the massless channel while the reduced static
energy and the relative source directions determine the sign of the force.

\subsection{Energy sources and the complete static response}
\label{sec:native-static-exchange}

Let $\mathsf t_a^{\mu\nu}(\mathbf x)$ denote the vacuum-subtracted expectation
of $T_{\rm mat}^{\mu\nu}$ in a stationary localised matter configuration.
The tensor source is kept complete including its stresses.
At leading exchange order each configuration is evaluated in its
unsourced stationary state. The following integrated stress identity
is Laue's theorem~\cite{Giulini2018Laue} with the long-wavelength
estimate stated in the normalisation of the matter vertex.

\begin{proposition}[Energy source, factorised residue, and static exchange]
\label{prop:native-matter-energy-monopole}
\label{prop:native-static-energy}
\label{prop:native-energy-factorisation}
Suppose $\mathsf t_a$ is symmetric and time independent and is conserved in the
fixed flat coordinates and
$\int(1+|\mathbf x|)\|\mathsf t_a(\mathbf x)\|\,\dd^3x<\infty$.
In its rest frame set $E_a=\int\mathsf t_a^{00}\,\dd^3x$ and
$(e_{00})^{\mu\nu}=\delta^\mu_0\delta^\nu_0$. Then
\begin{equation}
 \int\mathsf t_a^{ij}\,\dd^3x=0,\qquad
 \widehat{\mathsf t}_a(\mathbf p)=E_a e_{00}+r_a(\mathbf p),\qquad
 \|r_a(\mathbf p)\|\leq |\mathbf p|
            \int |\mathbf x|\|\mathsf t_a(\mathbf x)\|\,\dd^3x.
 \label{eq:native-matter-energy-monopole}
\end{equation}
Thus the zero-momentum vertex in
Eq.~\eqref{eq:native-matter-tensor-coupling} has the same energy
normalisation for every such configuration including bound systems
when their complete conserved tensor is used.

Suppose in addition that on a fixed momentum ray the reduced source is
$\widehat j_a(\mathbf p)=\mathsf V(\mathbf p)^*
\widehat{\mathsf t}_a(\mathbf p)$ where
$\|\mathsf V(\mathbf p)^*\|\leq C_V$ and
$\|\mathsf V(\mathbf p)^*-\mathsf V_0^*\|\leq K_V|\mathbf p|$.
For the critical projection $P$ set
$v_0=P\mathsf V_0^*e_{00}$ and
$I_a=\int |\mathbf x|\|\mathsf t_a(\mathbf x)\|\,\dd^3x$. Then
\begin{equation}
 \left\|P\widehat j_a(\mathbf p)-E_a v_0\right\|
 \leq |\mathbf p|\bigl(C_V I_a+K_V|E_a|\bigr).
 \label{eq:native-source-moment-bound}
\end{equation}
If the critical inverse of Theorem~\ref{thm:native-finite-rank-residue}
has positive propagation form $G$ on $P\mathcal X$ its mixed residue is
\begin{equation}
 Z_{ab}^{\rm mat}=E_aE_b\,\mathfrak b_{\rm mat},\qquad
 \mathfrak b_{\rm mat}=\langle v_0,G^{-1}v_0\rangle\geq0,
 \label{eq:native-matter-common-residue}
\end{equation}
and $\mathfrak b_{\rm mat}>0$ exactly when $v_0\neq0$.

Moreover suppose the constraint reduction has a real Hilbert energy space
$\mathcal V$ with a bounded symmetric coercive form $\mathfrak l$ and the
source-dependent energy on that space is
\begin{equation}
 E_{\rm stat}[\gamma;J]
 =\tfrac12\mathfrak l[\gamma,\gamma]-J[\gamma]
   +E_m+E_{\rm dir},\qquad J\in\mathcal V^*.
 \label{eq:native-static-energy}
\end{equation}
Let $\mathcal G:\mathcal V^*\to\mathcal V$ be the inverse of its form
operator. The unique stationary point is $\gamma=\mathcal GJ$ and for
$J=J_a+J_b$ the mediated cross term is
\begin{equation}
 E_{ab}^{\rm med}=-J_a[\mathcal GJ_b].
 \label{eq:native-static-cross-energy}
\end{equation}
Direct matter cross terms if present are added separately.
\end{proposition}

\begin{proof}
Choose a smooth cutoff $\chi_R$ equal to one on the ball of radius $R$
and zero outside the ball of radius $2R$. Conservation gives
\[
 \int\chi_R\mathsf t_a^{ij}\,\dd^3x
   =-\int x^i(\partial_k\chi_R)\mathsf t_a^{kj}\,\dd^3x.
\]
The right-hand side tends to zero by integrability since
$|x^i\partial_k\chi_R|$ is bounded on its receding annulus.
The mixed components integrate to the total momentum which vanishes in the rest frame. The remaining component integrates to $E_a$. Subtracting the
zero-momentum Fourier transform and using
$|e^{-\iu\mathbf p\cdot\mathbf x}-1|\leq|\mathbf p|\,|\mathbf x|$
proves the bound with the interaction and binding contributions
included in the conserved tensor.
Applying the two bounds on $\mathsf V^*$ to
$\widehat{\mathsf t}_a=E_ae_{00}+r_a$ gives
Eq.~\eqref{eq:native-source-moment-bound}. Multiplication of the critical
inverse by $|\mathbf p|^2$ and passage to the source limit then gives
Eq.~\eqref{eq:native-matter-common-residue}. Positivity and strict positivity
follow from the positive form $G^{-1}$ on the visible critical space.
The static statement follows from the Riesz representation theorem and completion
of the square in the coercive form
\[
 \tfrac12\mathfrak l[\gamma-\mathcal GJ,\gamma-\mathcal GJ]
    -\tfrac12J[\mathcal GJ].
\]
Symmetry of the inverse gives Eq.~\eqref{eq:native-static-cross-energy}.
\end{proof}

For a slowly varying static field the corresponding position-space difference between $\int\mathsf t_a^{\mu\nu}(\mathbf x)
\widetilde\Phi_{\mu\nu}(\mathbf x)\,\dd^3x$ and
$E_a\widetilde\Phi_{00}(0)$ is bounded by
$\|\nabla\widetilde\Phi\|_\infty
\int |\mathbf x|\|\mathsf t_a\|\,\dd^3x$.
The conserved system includes any external supports that carry momentum. The proposition applies to the resulting isolated stationary source.
Stationary configurations on a compact domain instead retain the full
modal vertex since the cutoff argument uses a noncompact matter domain.
The complete conserved tensor incorporates stresses and interactions together with binding energy into the same monopole charge. The Ward identity therefore
selects total energy as the long-wavelength source.

The overall minus sign has a direct variational origin. The positive
quadratic form is the cost of a field configuration. Minimising it in the
presence of a source lowers the energy by one half of the source--response
pairing. Attraction follows when two physical sources have positive cross
pairing through $\mathcal G$ as in the common energy-factorised channel
below. General source orientations retain the sign of
$J_a[\mathcal GJ_b]$.

After reducing the sourced constraints write the field seen by
matter as $\widetilde\Phi=\mathsf V\gamma$ in the retained coordinates. The
static sign convention corresponding to the coupling above is
$E_{\rm coupling}=-J[\gamma]$ with $J_a=\mathsf V^*\mathsf t_a$.
The mediated tensor response is consequently
\begin{equation}
 \mathcal D_{\rm stat}=\mathsf V\mathcal G\mathsf V^*,\qquad
 E_{ab}^{\rm med}=-\langle\mathsf t_a,\mathcal D_{\rm stat}\mathsf t_b\rangle,
 \label{eq:native-existing-vertex-exchange}
\end{equation}
where $\mathcal G$ is the inverse of the reduced energy form in
Proposition~\ref{prop:native-static-energy}. Source-dependent eliminated
components also produce the direct term of
Eq.~\eqref{eq:native-schur-tuple}. The reduction retains both this term
and the stresses. For a translation-invariant response
Eq.~\eqref{eq:native-matter-energy-monopole} selects its $00,00$
contraction at leading long wavelength. The full kernel including its constraint contributions determines the sign and spatial decay. A change $\widetilde\Phi\mapsto a\widetilde\Phi$ multiplies its covariance by $a^2$
and its vertex by $a^{-1}$ leaving the exchange unchanged.

For the compact spectrum of Section~\ref{sec:native-tt-spectral-branch}
the modal energy is
$M_a(c,\vartheta)=\mu_*\widehat M_a(\vartheta)/c$ where
$\widehat M_a:=\sqrt{\widehat\lambda_a+\widehat m_a^2}$.
Apply the same geometric compression $c\mapsto c e^{-\eta}$ to
all carriers with their dimensionless content fixed. Directly
\begin{equation}
 \left.\partial_\eta M_a(c e^{-\eta},\vartheta)\right|_0=M_a(c,\vartheta),
 \qquad
 \left.\partial_\eta\sum_a n_a M_a(c e^{-\eta},\vartheta)\right|_0
       =\sum_a n_aM_a(c,\vartheta).
 \label{eq:native-uniform-energy-source}
\end{equation}
The occupations are held fixed and the sum is finite or absolutely
convergent. This internal scale derivative measures the integrated
modal energy at a fixed reference branch.

For the complete interacting spectral energy $\mathcal E_{\rm mat}$
the same statement extends to its degree-one homogeneous part. If
$\mathcal E_{\rm mat}(c,\xi)=\mu_*c^{-1}\widehat{\mathcal E}_{\rm mat}(\xi)$
with dimensionless $\widehat{\mathcal E}_{\rm mat}$ then compression differentiates its full energy including binding.
When dimensionless variables $\xi$ change under the source variation
their contribution $D_\xi\mathcal E_{\rm mat}\,\partial_\eta\xi$
must also be retained. Equation~\eqref{eq:native-endpoint-modal-response}
is an explicit example of such an additional contribution.
These spectral derivatives refer to the specified internal comparison.
For a nonlinear source map its second derivative also enters the
chain rule. This internal susceptibility is distinct from the
propagator of $\widetilde\Phi$.

The scalar branch supplies a local representative of the same energy
source. Appendix~\ref{sec:app-local-scalar-energy} differentiates the
conformally transported square-root operator on its common form domain.
For a stationary occupation matrix commuting with that operator the
result in a joint eigenbasis is
\begin{equation}
 \begin{gathered}
 DE_{\rm occ}(0)[\eta]=\int_M\eta\rho_{E,\rm occ}\,\dd v_g,
 \qquad \rho_{E,\rm occ}=\sum_i n_iM_i|u_i|^2\geq0,\\
 \int_M\rho_{E,\rm occ}\,\dd v_g=E_{\rm occ}(0).
 \end{gathered}
 \label{eq:native-local-scalar-charge}
\end{equation}
The matrix version proved there also covers degeneracies and coherent
entries. This density enters the reduced source map $\mathsf V^*$ while
the common matter vertex remains
Eq.~\eqref{eq:native-matter-tensor-coupling}.

Proposition~\ref{prop:native-energy-factorisation}
then shows how its long-wavelength projection acquires a common critical direction. The matter state enters through $E_a$ while
$\mathfrak b_{\rm mat}$ belongs to the shared propagating channel.
Directional dependence of $P$ and $G$ together with that of $\mathsf V_0$ is retained until the
isotropy of this contraction is established. The compact density has this
matter interpretation after it is transported to the same source domain.

The quadratic equations can include constraint multipliers.
For example a block with entries $K,C^*,C,0$ is a saddle system rather
than a positive form on all its variables. The construction in
Eq.~\eqref{eq:native-positive-resolvent-measure} applies to its positive
reduced form on the physical quotient.
Complementary sources and
direct terms are retained through Eq.~\eqref{eq:native-schur-tuple}.
Static matter enters through the complete conserved source. Its TT
projection is one response channel while the energy monopole is selected
by the zero-momentum Ward identity.

The sign of the mixed contraction depends on the relative source
directions. For $\mathcal G=I$ opposite vectors have negative mutual
contraction. A positive energy-factorised residue instead gives a
negative mediated cross energy for positive energy charges.

For a centred Gaussian field of covariance $\mathcal G$ completing the
square gives $W[J]-W[0]=J[\mathcal GJ]/2$ and $\langle\gamma\rangle_J=\mathcal GJ$. This is the same source pairing as in the reduced energy. This field measure is distinct from Gaussian scale averaging.

\subsection{Spatial propagation and long-range exchange}
\label{sec:native-spatial-spectrum}

Spatial propagation depends on the external momentum fibres of the
response operator. Suppose separation of variables gives $L(\mathbf p)=A_{\rm int}+\alpha_{\rm sp}|\mathbf p|^2I$ with $\alpha_{\rm sp}>0$ and
$A_{\rm int}\geq0$. A visible internal kernel then gives the fibre term
$P_{\ker A_{\rm int}}/(\alpha_{\rm sp}|\mathbf p|^2)$ whose distributional
inverse supplies spatial propagation on infinite volume even though the
full operator has no normalisable zero eigenvector.

The full $L^2(\mathbb R^3)$ Laplacian has no normalisable zero eigenvector
but its distributional inverse is
\begin{equation}
 G_{\rm sp}(x)=\frac{1}{4\pi\alpha_{\rm sp}|x|},\qquad
 -\alpha_{\rm sp}\Delta G_{\rm sp}=\delta_0.
 \label{eq:native-spatial-green}
\end{equation}
with the coefficient fixed by the flux through a small sphere. It acts on
compact sources in the homogeneous gradient energy space for which
$\int|\widehat J(\mathbf p)|^2/|\mathbf p|^2\,\dd^3p<\infty$. The fibrewise
$1/|\mathbf p|^2$ singularity therefore produces spatial propagation even
though zero is not an $L^2$ eigenvalue of the full-space operator.

For the free-carrier spectral rule
\eqref{eq:native-original-spectral-propagation} the static kernel is
obtained mode by mode. Put $a_{ab,l}=\langle g_a,P_lg_b\rangle$.
The three-dimensional Fourier transform at zero external frequency is
\begin{equation}
 G_{ab}^{\rm free}(r)=\frac1{4\pi r}\sum_l a_{ab,l}e^{-M_lr},
 \qquad r>0.
 \label{eq:native-original-static-kernel}
\end{equation}
Absolute summability of $a_{ab,l}$ justifies the sum and its derivative
away from zero. Since $M_l\geq m_\delta>0$ it gives the explicit bounds
\begin{equation}
 \begin{aligned}
 |G_{ab}^{\rm free}(r)|&\leq
       \frac{\|g_a\|\|g_b\|}{4\pi r}e^{-m_\delta r},\\
 |\partial_rG_{ab}^{\rm free}(r)|&\leq
       \frac{\|g_a\|\|g_b\|}{4\pi r^2}
                    (1+m_\delta r)e^{-m_\delta r}.
 \end{aligned}
 \label{eq:native-original-static-bounds}
\end{equation}
The second bound uses that $(1+Mr)e^{-Mr}$ decreases with $M\geq0$.
These are the spatial consequences of the scalar mass tower.
They fix the elementary propagator used in composite spectral
calculations. The collective tensor uses its resummed response and
the matter vertex of Eq.~\eqref{eq:native-matter-tensor-coupling}.

On a connected compact space without boundary integration of
$-\Delta\gamma=J$ gives $\int J=0$. Its inverse is consequently defined
only on a compatible source subspace up to a constant. Replacing a
positive total energy source by its zero-mean part changes the physical
problem unless a background or constraint equation supplies exactly
that compensation. Finite-space distance windows and their boundary
conditions are separate from the infinite-volume limit above.

Consider an observable whose matched physical TT measure
$\dd\rho_{\rm resp}$ contains an atom $Z_0^{(\rm resp)}\delta_0$.
For $J^{\mu\nu}\in\mathcal J_{\TT}$ its quadratic source response is
\begin{equation}
 W[J]-W[0]
 =\frac12\int\frac{\dd^4Q}{(2\pi)^4}\,
 J^{\mu\nu}(-Q)\Pi^{(2)}_{\mu\nu,\rho\sigma}(Q)
 \mathcal C_{\rm resp}(Q^2)J^{\rho\sigma}(Q)+\cdots,
 \label{eq:source-source-generating-functional}
\end{equation}
with
\begin{equation}
 \mathcal C_{\rm resp}^{\rm nl}(Q^2)=\frac{Z_0^{(\rm resp)}}{Q^2}
       +\mathcal C_{{\rm resp},\rm rem}(Q^2),
 \qquad Q^2\mathcal C_{{\rm resp},\rm rem}(Q^2)\longrightarrow0.
 \label{eq:tt-massless-response}
\end{equation}

\begin{proposition}[Long-range TT response and static tail]
\label{prop:long-range-tt-response}
\label{prop:native-green-tail}
Under the stated positive spectral representation a positive atom
$Z_0^{(\rm resp)}$ contributes
\begin{equation}
 \int\frac{\dd^4Q}{(2\pi)^4}\frac{e^{\iu Q\cdot x}}{Q^2}
 =\frac1{4\pi^2|x|^2}
 \label{eq:four-dimensional-massless-green}
\end{equation}
to the connected response of physical TT sources. If the remaining static
mixed measure $\nu$ on $(0,\infty)$ satisfies
$\int e^{-r_0\sqrt s}\,\dd|\nu|(s)<\infty$ then
\begin{equation}
 G_{\rm stat}(r)=\frac{Z_0^{(\rm resp)}}{4\pi r}+g_{\rm rem}(r),
 \qquad
 g_{\rm rem}(r)=\frac1{4\pi r}\int e^{-r\sqrt s}\,\dd\nu(s),
 \label{eq:native-green-remainder}
\end{equation}
where $g_{\rm rem}=o(r^{-1})$ and $g_{\rm rem}'=o(r^{-2})$.
External Gaussian averaging at finite $\tau$ leaves the atomic coefficient
unchanged.
\end{proposition}

\begin{proof}
On conserved traceless sources the TT projector is the identity. Fourier
transformation of the atomic term gives Eq.~\eqref{eq:four-dimensional-massless-green}
and its zero-frequency restriction gives $1/(4\pi r)$. The Yukawa transform
of the remaining measure is Eq.~\eqref{eq:native-green-remainder}. Dominated
convergence proves the potential limit while
\begin{equation}
 4\pi r^2g_{\rm rem}'(r)
  =-\int(1+r\sqrt s)e^{-r\sqrt s}\,\dd\nu(s)
 \label{eq:native-green-derivative}
\end{equation}
and $(1+u)e^{-u}\leq C e^{-u/2}$ give the force limit. Finally
\begin{equation}
 \int_\tau^\infty\frac{e^{-r^2/(4t)}}{(4\pi t)^2}\,\dd t
 =\frac{1-e^{-r^2/(4\tau)}}{4\pi^2r^2}
 \label{eq:smeared-atomic-spatial-kernel}
\end{equation}
shows directly that finite Gaussian resolution changes only the short-distance
profile.
\end{proof}

The total-variation hypothesis includes mixed spectral measures which may be
signed even when the matrix measure is positive. Source multipoles and direct terms together with other spatial corrections enter the total remainder with their own
derivative bounds.

When source projection and constraint reduction together with spatial inversion give a common static energy channel their result takes the form
\begin{equation}
 E_{ab}(r)
 =-\frac{\beta_{\rm eff}^2}{4\pi\alpha_{\rm sp}}
       \frac{q_a^{\rm src}q_b^{\rm src}}{r}
      +\mathcal R_{ab}(r),\qquad
 \mathcal R_{ab}'(r)=o(r^{-2}).
 \label{eq:native-energy-force}
\end{equation}
For the tensor vertex and an isotropic critical contraction in
Eq.~\eqref{eq:native-matter-common-residue} $q_a^{\rm src}=E_a$ and
$\beta_{\rm eff}^2/\alpha_{\rm sp}=\mathfrak b_{\rm mat}$. The coefficient is fixed
by the same reduced propagation and source maps.
If quasistatic work satisfies
$\delta W=-\delta E_{ab}$ its radial force is
 $-\beta_{\rm eff}^2q_a^{\rm src}q_b^{\rm src}/(4\pi\alpha_{\rm sp} r^2)
 +o(r^{-2})$.
The leading force is attractive for positive product
$q_a^{\rm src}q_b^{\rm src}$.

The acceleration depends also on the inertial mass. For a stable
matter branch with
\begin{equation}
 E_a(\mathbf P)=q_a^{\rm src}+\frac{|\mathbf P|^2}{2m_{I,a}}
                         +o(|\mathbf P|^2),\qquad m_{I,a}>0,
 \label{eq:native-inertial-response}
\end{equation}
the ratio $q_a^{\rm src}/m_{I,a}$ enters the acceleration. In the Lorentz-covariant
realisation of Proposition~\ref{prop:native-matter-energy-monopole}
this ratio follows from the same total tensor. For the family obtained
by boosting the isolated stationary state a boost $\Lambda$ changes
the equal-time volume by $1/\Lambda^0{}_0$. Stationarity and the
integrated stress identity give
\begin{equation}
 P_a'{}^\mu=\frac{\Lambda^0{}_\alpha\Lambda^\mu{}_\beta}
                         {\Lambda^0{}_0}
       \int\mathsf t_a^{\alpha\beta}\,\dd^3x
       =E_a\Lambda^\mu{}_0,\qquad
 E_a(\mathbf P)=\sqrt{E_a^2+|\mathbf P|^2}.
 \label{eq:native-energy-inertia-match}
\end{equation}
For $E_a>0$ this proves $m_{I,a}=E_a=q_a^{\rm src}$ for the Lorentz-covariant
boosted family used in this realisation. The equality follows from its
dispersion relation together with the common conserved source map.
For other dispersion relations Eq.~\eqref{eq:native-inertial-response}
defines the relevant inertial mass. At finite coupling the total energy
also includes energy transfer and the mediator contribution.
A field redefinition $\gamma'=\xi\gamma$ sends
$\alpha_{\rm sp}$ to $\alpha_{\rm sp}/\xi^2$ and $\beta_{\rm eff}$ to
$\beta_{\rm eff}/\xi$. The measured combination $\beta_{\rm eff}^2/\alpha_{\rm sp}$ in
Eq.~\eqref{eq:native-energy-force} is unchanged. A zero-threshold continuum
may add a further power-law tail while local subtraction terms remain at
contact support. The atomic coefficient is isolated by the spectral limit
and the derivative estimate above controls the force independently of the
potential. The result is a source-visible long-range tensor response. Its
helicity content and mixed-matter factorisation determine its gravitational interpretation.

\subsection{Helicity, factorisation, and the physical coupling}
\label{sec:helicity-gate}

The Euclidean projector in Eq.~\eqref{eq:euclidean-tt-projector} has rank
five and describes the continued transverse traceless spin-2 sector at
nonzero invariant momentum. The physical helicity space of a massless
spin-2 particle is two-dimensional. A helicity-two interpretation
identifies the atom in the
Lorentzian continuation with a positive-norm massless representation
of helicities $+2$ and $-2$. Longitudinal and negative-norm
representatives then decouple from physical amplitudes. For the
gravitational exchange considered here the same pole must couple
to matter in mixed correlators and the response to a generic
conserved source must have no additional massless scalar residue.

The helicity step identifies the physical polarisation content of the
covariant pole. Factorisation then tests whether the same polarisation
space and source vertex occur for every matter state while the soft limit
tests the resulting universal charge. These stages supply the polarisation
and coupling data beyond the scalar residue.

The trace sector matters for coupling to ordinary matter. Between generic conserved sources a
massless graviton in four dimensions has numerator
\begin{equation}
 \mathcal P^{\rm GR}_{\mu\nu,\rho\sigma}
 =\frac12(\eta_{\mu\rho}\eta_{\nu\sigma}
          +\eta_{\mu\sigma}\eta_{\nu\rho})
  -\frac12\eta_{\mu\nu}\eta_{\rho\sigma}
  +\text{terms proportional to }p_\mu,
 \label{eq:graviton-conserved-source-numerator}
\end{equation}
where the last terms vanish on conserved sources. The nonzero-mass TT
projector instead contains the coefficient $-1/3$ in four dimensions.
Their contractions agree on traceless sources but not on general matter.
The full gravitational tensor structure therefore involves the Ward
identities and scalar sector as well as the pole in the matched response. Their
combined massless amplitude has the Einstein form when its numerator
is Eq.~\eqref{eq:graviton-conserved-source-numerator}.

\paragraph{Factorisation across matter channels.}
Let $a$ and $b$ label matter states or gauge-invariant matter operators. A
force interpretation requires the connected mixed amplitude to factorise at
the same pole
\begin{equation}
 \mathcal A_{ab}(p)
 =\frac{\mathcal V_a^{\mu\nu}(p)
        \mathcal P^{\rm phys}_{\mu\nu,\rho\sigma}(p)
        \mathcal V_b^{\rho\sigma}(-p)}{-p^2-\iu0}
  +\mathcal A_{ab}^{\rm reg}(p),
 \label{eq:mixed-amplitude-factorisation}
\end{equation}
with positive physical residue. This factorisation is a statement about
mixed source derivatives of the same generating functional.
The common denominator identifies a shared singular scale. The common
physical projector and the two source vertices identify one exchanged
channel.

If the helicity and soft-coupling conditions hold denote the canonically normalised field by $\gamma^{\rm can}_{\mu\nu}$. The nonlinear effective metric and background source $h_{\mu\nu}$ of Section~\ref{sec:background-source-frg} have their respective roles. Choose the conventional normalisation
\begin{equation}
 \mathcal L_{\rm int}=-\frac{\kappa_{\rm phys}}2
 \gamma^{\rm can}_{\mu\nu}T^{\mu\nu}_{\rm matter}.
 \label{eq:physical-graviton-coupling}
\end{equation}
This is the convention
$\widetilde\Phi_{\mu\nu}=-(\kappa_{\rm phys}/2)
\gamma^{\rm can}_{\mu\nu}$
for the tensor vertex in Eq.~\eqref{eq:native-matter-tensor-coupling}.
Then for conserved sources
\begin{equation}
 \mathcal A_{ab}(p)
 \underset{p^2\to0}{\sim}
 \frac{\kappa_{\rm phys}^2}{4}
 \frac{T_a^{\mu\nu}T_{b,\mu\nu}
       -\tfrac12T_aT_b}{-p^2-\iu0},
 \qquad
 \kappa_{\rm phys}^2=32\pi G_N.
 \label{eq:newton-from-mixed-residue}
\end{equation}
This equation defines $G_N$ in the stated convention. A rescaling of
the interpolating operator changes its Wightman weight $Z_0^{(\rm resp)}$ and
compensates that change in the matter vertices. The mixed residue in
Eq.~\eqref{eq:mixed-amplitude-factorisation} and hence $G_N$ is invariant.

\subsection{The Weinberg--Witten obstruction}
\label{sec:ww-gate}

The Weinberg--Witten theorem separates the spectral and particle questions.
The positive vacuum measure concerns the states created by the operator
whereas the theorem tests whether an ordinary massless particle can carry
the translation charge measured by the same covariant stress tensor. In the
form relevant here a positive TT atom interpreted as an asymptotic
massless helicity-two particle in the same positive-metric Hilbert space is
excluded whenever a Lorentz-covariant conserved stress tensor carries its
nonzero four-momentum and the remaining hypotheses on states and matrix
elements hold~\cite{WeinbergWitten1980}. The conclusion is independent of
whether the proposed particle is elementary or composite.

This particle obstruction and the vacuum spectral statement test different
matrix elements. For the covariant stress-tensor input class of
Section~\ref{sec:covariant-stress-realisation} the null-shell calculation
gives zero TT atomic weight directly. For an independently constructed
response the corresponding particle test uses its stress-tensor matrix elements and translation charge together with the remaining hypotheses of the
standard theorem.

For any other response applicability is determined by its infrared states and the Lorentz covariance of their stress-tensor matrix elements together with the conserved translation charge acting on them including the forward
matrix-element limit. Its infrared realisation therefore specifies the state space and charge together with the matrix elements that replace any standard-particle
hypothesis. When ordinary asymptotic states are absent the spectral
long-range response is still defined while Weinberg's $S$-matrix soft
theorem applies after an infrared scattering framework has been
constructed.
These restrictions are discussed for emergent gravity in
ref.~\cite{Jenkins2009}. The corresponding higher-spin bounds and their
relation to soft theorems are given in refs.~\cite{Porrati2008,Porrati2012}.

\subsection{Soft universality and the Einstein form}
\label{sec:einstein-gate}

Assume now that an infrared scattering framework exists and that the same
positive-norm helicity-two pole factorises in amplitudes involving every
matter species. Lorentz invariance and decoupling of spurious
polarisations imply the leading soft identity. Weinberg's argument then
forces the coupling to be proportional to each external four-momentum and
hence universal~\cite{Weinberg1965}. This supplies universality at leading soft order. Its extension to a local low-energy theory uses the consistency of the nonlinear gauge symmetry.

Once the exchanged channel has the helicity and universal soft structure
of a massless spin-two field the remaining local question is the
nonlinear completion of its gauge symmetry. For a local derivative expansion with ghost-free propagation and a single interacting massless spin-two field with at most two derivatives the consistent
Poincar\'e-invariant completion is Einstein gravity up to field
redefinitions and boundary terms~\cite{Deser1970,Boulanger2000}. The
low-energy action may therefore be written
\begin{equation}
 S_{\rm IR}
 =\frac1{16\pi G_N}\int\dd^4x\sqrt{-\gamma}\,
      (R[\gamma]-2\Lambda)
 +S_{\rm matter}[\gamma,\Psi]
 +\sum_i c_i\mathcal I_i^{(4+)},
 \label{eq:conditional-einstein-action}
\end{equation}
where $\mathcal I_i^{(4+)}$ are local higher-derivative operators after
additional massive fields have been integrated out. Light matter fields
are retained in $S_{\rm matter}$. The vacuum density and higher-derivative
terms are fixed by additional coefficients of the effective action.
An exact Minkowski background satisfies the vanishing renormalised
tadpole condition including the cosmological term.

The chain from the spectral response to this effective action is now explicit.
A positive TT atom supplies the long-range channel of
Proposition~\ref{prop:long-range-tt-response}. The helicity and conserved-source
numerator select the massless spin-two exchange. Mixed-amplitude factorisation
and the soft identity make its coupling universal. When the physical
realisation lies outside an applicable Weinberg--Witten hypothesis and admits
the local ghost-free single-field two-derivative completion used above the
resulting infrared action is Eq.~\eqref{eq:conditional-einstein-action}.

The pole and its helicity content determine the propagating field. The matter vertex and conserved-source numerator fix its exchange normalisation. A common metric--matter action carries these data into the nonlinear gauge identity and constraints with the same translation charge and cosmological normalisation.

\subsection{A normalised Einstein--matter realisation}
\label{sec:common-einstein-matter}

At leading two-derivative order the completion above can be normalised so
that the same coefficients govern the spectral pole and matter vertex together with the macroscopic metric. The compact spectral system continues to fix the
internal channel and its source normalisation.
For $g_\Phi\ne0$ introduce a metric $\gamma_{\mu\nu}$ and a positive
coefficient $M_*^2$ by
\begin{equation}
 \gamma_{\mu\nu}=\eta_{\mu\nu}-2g_\Phi\Phi_{\mu\nu},\qquad
 M_*^2=Z/g_\Phi^2.
 \label{eq:common-metric-normalization}
\end{equation}
We retain signature $(+---)$ and the curvature convention
$R=-\gamma^{\mu\nu}\mathcal R_{\mu\nu}$.
For real matter fields $\varphi^A$ with field-space metric $\delta_{AB}$
the Einstein--scalar action is
\begin{align}
 S[\gamma,\varphi]&=-\frac{M_*^2}{2}
   \int\dd^4x\sqrt{-\gamma}\,\gamma^{\mu\nu}
      \mathcal R_{\mu\nu}(\gamma)+S_m[\gamma,\varphi],\notag\\
 S_m[\gamma,\varphi]&=\int\dd^4x\sqrt{-\gamma}
  \left[\frac12\gamma^{\mu\nu}\delta_{AB}
          \partial_\mu\varphi^A\partial_\nu\varphi^B-V(\varphi)\right].
 \label{eq:common-einstein-scalar-action}
\end{align}
Assume $V(0)=\partial_A V(0)=0$ and a nonnegative mass matrix
$\partial_A\partial_BV(0)$. These conditions select the flat stationary
matter background and a nonnegative quadratic matter spectrum.
Define the Hilbert source by
\begin{equation}
 T^{\mu\nu}=-\frac{2}{\sqrt{-\gamma}}
             \frac{\delta S_m}{\delta\gamma_{\mu\nu}}.
 \label{eq:common-hilbert-source}
\end{equation}
All species couple to the same metric. Other covariant matter
functionals obey the corresponding Noether identity with their tensor
or spinor Euler derivatives.

The family replication of Ref.~\cite{Guo2026CompactSpectrumII} gives a
concrete quadratic source. At zero Higgs background
$K_{n_f}[h]=I_{n_f}\otimes K_1[h]$ on the common source-dependent domain
so $W_{n_f}[h]=n_fW_1[h]$ up to a source-independent constant and
\begin{equation}
 \partial_a\partial_bW_{n_f}
 =n_f\left(\langle S_{1,a}S_{1,b}\rangle_c
                 -\langle S_{1,ab}\rangle\right),
 \label{eq:common-family-quadratic-source}
\end{equation}
where the derivatives include the domain and measure identification. The connected quadratic response and its direct term together with its TT compression thus
carry the same family multiplicity. The compact representation and its
three-family application are given in
refs.~\cite{Guo2026CompactSpectrumI,Guo2026CompactSpectrumII}.

\paragraph{Classical completion and the cubic identity.}
The action~\eqref{eq:common-einstein-scalar-action} has quadratic tensor
part~\eqref{eq:external-fp-action} and leading matter vertex
$g_\Phi\int\Phi_{\mu\nu}T_0^{\mu\nu}\dd^4x$ where $T_0$ is the flat
Hilbert tensor. For compactly supported gauge parameters its cubic
expansion obeys
\begin{equation}
 \delta_0 S_3+\delta_1 S_2=0,
 \label{eq:common-cubic-ward}
\end{equation}
with the matter and tensor terms included in each $S_n$.
The full classical equations obey the corresponding off-shell
Noether identity and propagate the gravitational constraints.
Writing $h_{\mu\nu}=\gamma_{\mu\nu}-\eta_{\mu\nu}$ the quadratic curvature
action is $M_*^2/4$ times the Fierz--Pauli bracket. Substitution of
$h=-2g_\Phi\Phi$ gives $M_*^2g_\Phi^2=Z$ while
$\delta S_m=-\frac12\int T_0^{\mu\nu}\delta\gamma_{\mu\nu}\dd^4x$ gives
the matter vertex with the same sign and normalisation. Deser's
self-coupling construction supplies Eq.~\eqref{eq:common-cubic-ward} and
the corresponding nonlinear gauge identity~\cite{Deser1970}.

More generally define
$\mathcal E^{\mu\nu}=(\sqrt{-\gamma})^{-1}
\delta S/\delta\gamma_{\mu\nu}$ and
$\mathcal E_A=(\sqrt{-\gamma})^{-1}\delta S/\delta\varphi^A$.
For the displayed matter action these are
\[
 \begin{aligned}
 T_{\mu\nu}&=\delta_{AB}\partial_\mu\varphi^A\partial_\nu\varphi^B
 -\gamma_{\mu\nu}\left(\tfrac12\gamma^{\rho\sigma}\delta_{AB}
       \partial_\rho\varphi^A\partial_\sigma\varphi^B-V\right),\\
 \mathcal E_A&=-\Box_\gamma\varphi_A-\partial_A V.
 \end{aligned}
\]
Integrating the metric Lie derivative by parts gives
\begin{equation}
 2\nabla_\mu\mathcal E^\mu{}_{\nu}
   =\mathcal E_A\partial_\nu\varphi^A,
 \qquad
 2\mathcal E^{\mu\nu}=M_*^2\mathcal G^{\mu\nu}-T^{\mu\nu}.
 \label{eq:common-offshell-noether}
\end{equation}
Here $\mathcal G_{\mu\nu}=\mathcal R_{\mu\nu}
-\gamma_{\mu\nu}\gamma^{\rho\sigma}\mathcal R_{\rho\sigma}/2$.
On the matter equations the metric Euler derivative is covariantly
conserved. Its lapse and shift equations are the four gravitational
constraints. The identity supplies their homogeneous propagation
equations. Their preservation with the four diffeomorphism freedoms
leaves two gravitational configuration degrees of freedom near the
flat vacuum. In particular the quadratic pole and matter normalisation
are those of Eq.~\eqref{eq:external-gauge-residue}.
The action fixes the source and propagation operator of the classical
effective realisation. Eq.~\eqref{eq:external-ladder-matching} fixes its
compact critical coordinate. The metric and matter fields are varied
independently in this effective realisation. When a generated construction
prescribes relations between these fields variation is taken after
pulling the action back to the admissible generated variables.
Finite rational evaluations of the density
identities are given in Ref.~\cite{Guo2026CoarseGrainingFramework}.

On an asymptotically flat Cauchy surface the differentiable translation
generator is
\begin{equation}
 H[N,N^i]=\int_\Sigma(N\mathcal H+N^i\mathcal H_i)\dd^3x
            +B_\infty[N,N^i].
 \label{eq:common-constraint-generator}
\end{equation}
Compact lapse and shift generate gauge transformations. For an asymptotic
time translation the energy on the constraint surface is the boundary
charge. With $\gamma^{(3)}_{ij}=\delta_{ij}+s_{ij}$
\begin{equation}
 P^0_{\rm ADM}=\frac{M_*^2}{2}\lim_{r\to\infty}
 \int_{S_r}(\partial_j s_{ij}-\partial_i s_{jj})n^i\dd S.
 \label{eq:common-adm-energy}
\end{equation}
The Hilbert source~\eqref{eq:common-hilbert-source} supplies the matter
contribution while Eq.~\eqref{eq:common-adm-energy} is the translation
charge on the physical constraint space~\cite{ReggeTeitelboim1974}. This
is the charge realisation entering the Weinberg--Witten comparison.

\paragraph{Cosmological and stationary limits.}
\label{sec:common-cosmological-static-limits}

The same action determines the background and its constraints together with the response about that background. Allow a constant vacuum density by
adding $-\int\sqrt{-\gamma}\,\rho_\Lambda\dd^4x$ to
Eq.~\eqref{eq:common-einstein-scalar-action}. Let $\rho$ and $p_{\rm mat}$ be the
energy density and isotropic pressure of the remaining covariant matter
action. On a spatially flat FLRW solution
\begin{equation}
 3M_*^2H^2=\rho+\rho_\Lambda,\qquad
 -2M_*^2\dot H=\rho+p_{\rm mat},\qquad
 \dot\rho+3H(\rho+p_{\rm mat})=0.
 \label{eq:common-flrw-background}
\end{equation}
The last equation is the matter Noether identity. Together with the second
equation it propagates the first including through $H=0$.

For tensor propagation let the background solve
Eq.~\eqref{eq:common-flrw-background} with constant $M_*^2>0$ and write
\[
 \dd s^2=\dd t^2-a(t)^2(e^h)_{ij}\dd x^i\dd x^j,
 \qquad \partial_i h_{ij}=h_{ii}=0.
\]
For unit-norm spatial polarizations set
$h_{ij}=\sum_{\lambda=1}^2h_\lambda e^\lambda_{ij}$ and
$e^\lambda_{ij}e^{\lambda'}_{ij}=\delta_{\lambda\lambda'}$.
The quadratic tensor action and its linear anisotropic-stress vertex are
\begin{align}
 S_T^{(2)}={}&\frac{M_*^2}{8}\sum_\lambda\int a^3
 \left[\dot h_\lambda^2-a^{-2}|\nabla h_\lambda|^2\right]
 \dd t\,\dd^3x\notag\\
 &+\frac12\sum_\lambda\int a^3 h_\lambda\Pi_\lambda
 \dd t\,\dd^3x,
 \label{eq:common-flrw-tensor-action}
\end{align}
where $\Pi_{ij}$ is the physical transverse traceless anisotropic
stress and $\Pi_\lambda=e^\lambda_{ij}\Pi_{ij}$. Consequently
\begin{equation}
 \ddot h_\lambda+3H\dot h_\lambda-a^{-2}\Delta h_\lambda
   =\frac{2}{M_*^2}\Pi_\lambda.
 \label{eq:common-flrw-tensor-equation}
\end{equation}
In the flat background with $a=1$ the substitution $h_\lambda=2g_\Phi\phi_\lambda$ recovers
the reduced action and residue in Eq.~\eqref{eq:external-gauge-residue}.
For stationary weak
perturbations of the flat background with nonrelativistic sources and
negligible anisotropic stress the
same metric equations give
\begin{equation}
 \Delta\Psi=\frac{\rho}{2M_*^2}=4\pi G_N\rho,
 \qquad G_N=\frac1{8\pi M_*^2}=\frac{g_\Phi^2}{8\pi Z}.
 \label{eq:common-stationary-poisson}
\end{equation}
The ADM constraint reduction gives the positive tensor kinetic term and
for a decaying stationary solution with negligible anisotropic stress
$\Psi=\Upsilon$. The Green function is therefore
$-G_N\int\rho(\mathbf y)|\mathbf x-\mathbf y|^{-1}\dd^3y$ with leading
interaction $-G_NE_aE_b/r$. Thus the tensor action and the static limit use
the same $M_*^2$ while the retarded tensor equation continues the
flat-background spectral pole.

\paragraph{Constant condensate and cold source.}
\label{sec:common-condensate-cold}

A radial scalar and an irrotational pressureless current test the same
effective action in its scalar and cold-matter sectors. With the conventions
of Eq.~\eqref{eq:common-einstein-scalar-action} take
\begin{equation}
 S_{r\Sigma}=\int\dd^4x\sqrt{-\gamma}\left[
  \frac12\gamma^{\mu\nu}\partial_\mu\varphi\partial_\nu\varphi
  -U(\varphi)+\frac{\varrho}{2}
       (\gamma^{\mu\nu}\partial_\mu T\partial_\nu T-1)\right].
 \label{eq:common-radial-cold-action}
\end{equation}
The multiplier $\varrho$ imposes a unit timelike velocity
$u_\mu=\partial_\mu T$. Its current equation is
$\nabla_\mu(\varrho u^\mu)=0$ and its stress is
$T^{\Sigma}_{\mu\nu}=\varrho u_\mu u_\nu$. Its normalisation and gradient
form imply $u^\mu\nabla_\mu u_\nu=0$. This is the irrotational sector of the
canonical dust description~\cite{BrownKuchar1995}. Thus $T$ is the
velocity potential of the matter flow and $T=t$ selects proper-time
coordinates adapted to that flow. For a generated matter phase such a
unit-gradient clock exists locally when the normalised phase one-form is
closed~\cite[Proposition ``Congruence of a generated matter phase'']
{Guo2026CompactSpectrumII}.
The phase-clock construction of ref.~\cite{Guo2026CompactSpectrumI}
Section~10 calibrates timelike translations through their energy-dependent
phase. The velocity potential here is the flow-adapted clock of this
particular matter realisation.

Let $U(\varphi_*)=U'(\varphi_*)=0$ and
$U''(\varphi_*)=m_\delta^2>0$. The constant vacuum energy is included
in $\rho_\Lambda$ of Eq.~\eqref{eq:common-flrw-background}.
For the compact scalar of Ref.~\cite{Guo2026CompactSpectrumI} these
conditions follow from $F_y(\sigma_*,y_*)=0$ and $Z_*>0$ with
$m_\delta^2=4y_*F_{yy}(\sigma_*,y_*)/Z_*$. Coupling
Eq.~\eqref{eq:common-radial-cold-action} to the Einstein action and the
constant density $\rho_\Lambda$ gives on a spatially flat FLRW solution
with $T=t$ and $\varphi=\varphi_*$ and $\varrho=\rho(t)>0$
\begin{equation}
 3M_*^2H^2=\rho+\rho_\Lambda,\qquad
 -2M_*^2\dot H=\rho,\qquad \dot\rho+3H\rho=0.
 \label{eq:common-condensate-background}
\end{equation}
Consider a real nonzero spatial Fourier mode of physical wavenumber
$\kappa=|\mathbf k|/a$ on an interval with $H\ne0$. In scalar spatially
flat gauge write $N=1+n$ and $T=t+\pi_T$ together with $\varrho=\rho+d$ and $b=\partial_iN^i$. Before the auxiliary variables are eliminated the
quadratic Lagrangian per unit $a^3$ is
\begin{align}
 \mathcal L_2={}&-3M_*^2H^2n^2-2M_*^2Hnb
 +\frac12\dot{\delta\varphi}^{,2}
 -\frac12(\kappa^2+m_\delta^2)\delta\varphi^2\notag\\
 &+d(\dot\pi_T-n)+\frac\rho2(\dot\pi_T-n)^2
 -\frac\rho2\kappa^2\pi_T^2+\rho b\pi_T .
 \label{eq:common-condensate-constraints}
\end{align}
The shift equation gives $n=B\pi_T$ and the lapse equation fixes $b$.
The cold canonical momentum is
$\pi_\rho=d+\rho(\dot\pi_T-n)$. Preservation of the primary multiplier
constraint sets $d=\pi_\rho$. After substituting the lapse and shift
solutions and integrating $a^3\pi_\rho\dot\pi_T$ by parts the remaining
$\pi_T$ dependence is
$-[\dot\pi_\rho+(3H+B)\pi_\rho]\pi_T-D\pi_T^2/2$.
Its algebraic equation yields the reduced quadratic action
\begin{align}
 S^{(2)}_{r\Sigma}&=\int\dd t\,a^3\left\{
  \frac12\dot{\delta\varphi}^{\,2}
  -\frac12(\kappa^2+m_\delta^2)\delta\varphi^2
  +\frac{[\dot\pi_\rho+(3H+B)\pi_\rho]^2}{2D}\right\},\notag\\
 B&=\frac{\rho}{2M_*^2H},\qquad
 D=\rho\kappa^2+\frac{3\rho^2}{2M_*^2}>0.
 \label{eq:common-condensate-reduced-action}
\end{align}
The kinetic matrix $\diag(1,D^{-1})$ is positive for $M_*^2,\rho>0$.
The radial mode has unit propagation speed and mass $m_\delta$ whereas
the cold mode has vanishing intrinsic short-wavelength sound speed. In
the regime $\kappa^2\gg\rho/M_*^2$ writing
$\pi_\rho=\rho\delta$ gives the standard growing-density equation
$\ddot\delta+2H\dot\delta-\rho\delta/(2M_*^2)=0$ at leading order.

The dynamical-amplitude realisation of
Ref.~\cite[Appendix~G.6, Proposition ``Congruence of a generated matter
phase'' and Proposition ``Occupied common background and physical
propagation'']{Guo2026CompactSpectrumII} contains a matter pair carrying a phase and
an auxiliary collective radius. Its regular algebraic radius equation
is solved in the common action and the lapse and shift constraints are
removed from its quadratic variation. The resulting physical scalar
kinetic and gradient matrices are positive. Their characteristic speeds and the
phase-clock rate have explicit positive bounds on a finite interval.
This gives a coupled realisation in which matter evolution and the collective amplitude together with the auxiliary radius satisfy the same source
balance.

\section{Conclusions}
\label{sec:conclusions}

A tensor pole contributes to matter exchange through the source that
creates it and the observable that reads it.
On an isolated critical subspace the propagation slope and projected
source maps determine the residue.
Exact channel elimination preserves the complete response.
Common-space estimates control the residue at the shifted critical
root of a finite approximation.
Positive spectral matching identifies this coefficient with the weight
of an invariant-mass atom.

The null-shell theorem classifies the complete microscopic stress tensor.
Its positive Lorentz-covariant conserved measure is longitudinal on the
massless orbit and has zero TT atomic weight.
The reduced two-helicity response has its own physical state space
and inherits its source maps from the effective spectral realisation.
The compact construction gives an internal spectral zero and explicit
matter vertices for this response.
Its charged trace fixes a response domain and endpoint.
The complete matter action supplies the conserved energy vertex.
The gapped radial example gives a different spectral outcome.
Its visible positive-mass mode has finite static susceptibility.
In both cases source overlap fixes the measured weight.

For localised stationary matter conservation and Laue's identity reduce
the leading long-wavelength vertex to the total energy including binding
contributions.
A common critical direction then gives an energy-factorised mixed residue.
Spatial resolvent estimates determine the leading potential and force
with controlled finite-size corrections.
A collective response with the required helicity content and universal
soft coupling admits the Einstein--matter completion under the stated
local nonlinear consistency assumptions.
The effective geometry represents the macroscopic collective tensor.

\appendix
\section{Gaussian averaging and spectral atoms}
\label{app:gaussian}

Gaussian resolution smooths an already defined response while preserving
the weight of its infrared atom. The characteristic exponent fixes the
averaging law and its Fourier multiplier gives the inverse on the range
and the identity~\eqref{eq:gaussian-atom-limit}.

\paragraph{Semigroup characterisation.}

Let $\varphi_\tau(q)=\int e^{\iu q\cdot x}\,\dd\mu_\tau(x)$. The convolution
law gives
\begin{equation}
 \varphi_{\tau_1+\tau_2}(q)
 =\varphi_{\tau_1}(q)\varphi_{\tau_2}(q),
 \qquad \varphi_0(q)=1.
 \label{eq:app-characteristic-semigroup}
\end{equation}
Orthogonal invariance includes reflection and makes $\varphi_\tau$ real.
For every integer $n$ the semigroup identity gives $\varphi_\tau(q)=\varphi_{\tau/(2n)}(q)^{2n}\geq0$. A zero at finite $\tau$ would give zeros arbitrarily close to $\tau=0$ by the semigroup law. Weak continuity excludes this possibility. Hence
\begin{equation}
 \varphi_\tau(q)=e^{-\tau\psi(q)}
 \label{eq:app-characteristic-exponent}
\end{equation}
with a continuous nonnegative $\psi$.

Square-root self-similarity also gives
$\varphi_\tau(q)=\varphi_1(\sqrt\tau q)$ so
$\psi(rq)=r^2\psi(q)$ for $r>0$. Orthogonal invariance implies that $\psi$
depends only on $|q|$ and quadratic homogeneity then yields
$\psi(q)=D|q|^2$. The covariance is
\begin{equation}
 \int x_i x_j\,\dd\mu_\tau(x)=2D\tau\,\delta_{ij},
 \label{eq:app-gaussian-covariance}
\end{equation}
so the nonzero second moment gives $D>0$ and fixes $D$ once the
scale convention is chosen. For $\tau>0$ Fourier inversion gives
\begin{equation}
 \dd\mu_\tau(x)
 =\frac{\dd^dx}{(4\pi D\tau)^{d/2}}
  \exp\left[-\frac{|x|^2}{4D\tau}\right].
 \label{eq:app-gaussian-density}
\end{equation}
\paragraph{Finite resolution and the infrared germ.}

The operator $\cH_\tau$ on $L^2(\mathbb R^d)$ acts by Fourier multiplication by
$e^{-\tau Q^2}$. For finite $\tau$ the multiplier is strictly positive so
\begin{equation}
 \Ker\cH_\tau=\{0\}.
 \label{eq:app-heat-injective}
\end{equation}
Its inverse has multiplier $e^{\tau Q^2}$ and is unbounded for every
$\tau>0$. Fourier data supported on $n<|Q|<n+1$ are amplified by at least
$e^{\tau n^2}$. The same local Fourier argument gives injectivity on
tempered distributions. Thus finite Gaussian resolution loses stability at
high momentum while retaining all Fourier directions.

For the infrared statement let $F$ be meromorphic near zero and let $w$ be
analytic with $w(0)=1$. If
\begin{equation}
 F(z)=\sum_{j=-m}^{\infty}a_jz^j,
 \qquad a_{-m}\neq0,
 \label{eq:app-laurent-form-factor}
\end{equation}
then multiplication by $w$ preserves the pole order and leading
coefficient. Taking $w(z)=e^{-\tau z}$ proves the meromorphic case of
Eq.~\eqref{eq:gaussian-atom-limit}. Its direct limit also covers an atom at
a continuous threshold. Since the same multiplication sends a contact
polynomial to the generally nonpolynomial entire function
$e^{-\tau Q^2}P_2(Q^2)$ contact subtraction is performed first. Positivity
and the Stieltjes representation therefore refer to the unfiltered
response.

\section{The TT quotient, contact terms, and spectral measure}
\label{app:tt-spectral}

The physical source quotient is formed before the massless limit. Local
counterterms and changes of stress-tensor representative then separate
from the nonlocal measure whose timelike disintegration and null-shell
restriction give Eq.~\eqref{eq:subtracted-tt-dispersion} without assigning
a massive projector to the null orbit.

\subsection{Renormalisation and the physical source quotient}

Multiple stress-tensor insertions are renormalised by allowing local
counterterms in the background source. In four dimensions their purely
geometric part has the schematic form
\begin{equation}
 S_{\rm ct}[g]
 =\int\dd^4x\sqrt g\,
 \bigl(c_0+c_1R+c_2R^2+c_3R_{\mu\nu}R^{\mu\nu}
       +c_4R_{\mu\nu\rho\sigma}R^{\mu\nu\rho\sigma}+\cdots\bigr).
 \label{eq:metric-source-counterterms}
\end{equation}
Two metric derivatives at $g_{\mu\nu}=\delta_{\mu\nu}$ are tensor
polynomials in $Q$ and contribute only to $P_2(Q^2)$. They have neither a
positive-mass discontinuity nor a $1/Q^2$ term. Matter-dependent local
counterterms have the same status after operator mixing is closed. An
anomaly-free source Ward identity makes the nonlocal correlator transverse
whereas a trace anomaly changes the trace and local terms.

For operators with the same exact quantum numbers write
$\cO_{R,i}=Z_{ij}\cO_{B,j}$ and $h_{B,i}=Z_{ji}h_{R,j}$ and form the
physical quotient only after this mixing problem is closed. BRST-exact and
renormalised EOM insertions then give their standard separated-point contact
identities~\cite{CollinsScalise1994,JoglekarLee1976,Barnich2000}. The
quotient is formed at nonexceptional momentum before continuation to the
infrared limit. Likewise the scalar improvement
$I_{\mu\nu}=(\partial_\mu\partial_\nu-\eta_{\mu\nu}\Box)X$ has
$J^{\mu\nu}I_{\mu\nu}=0$ for every conserved traceless source. It may alter
the trace and contact terms but leaves the TT response unchanged.

For a finite physical operator family bold symbols denote its matrix-valued
Wightman measure and atom. These are distinct from the scalar orbit measure
$\dd\rho_2$ and from a source-family pushforward. Under a finite nonsingular
change of physical operator basis $\cO'=M\cO$ the matrix-valued measure and
its atom transform by congruence
\begin{equation}
 \dd\boldsymbol\rho'_2=M\,\dd\boldsymbol\rho_2\,M^\dagger,
 \qquad \boldsymbol Z'_0=M\boldsymbol Z_0M^\dagger.
 \label{eq:app-residue-congruence}
\end{equation}
Thus positivity and rank are invariant under nonsingular basis changes.

\subsection{Spectral resolution on physical sources}

Let $E(\dd p)$ be the joint projection-valued measure of translations. For
a finite family of TT sources $J_i$ define
\begin{equation}
 \dd M_{ij}(p)
 :=\langle\Omega|\Theta(J_i)^\dagger E(\dd p)
                    \Theta(J_j)|\Omega\rangle.
 \label{eq:app-translation-spectral-measure}
\end{equation}
For every $c_i\in\mathbb C$ and nonnegative Borel function $f$
\begin{equation}
 \sum_{i,j}\bar c_i\int f(p)\,\dd M_{ij}(p)c_j\geq0.
 \label{eq:app-matrix-wightman-positivity}
\end{equation}
The spectral condition restricts the support to the closed forward cone.
This matrix is the source-smeared image of one common tensor-valued tempered
measure. The source-independent scalar measure is obtained by restricting the common tensor measure to the open timelike cone. The map
$p\mapsto s=p^2$ decomposes that cone into the Lorentz orbits $H_s^+$.
After the orbit measure $\dd\Omega_s$ is fixed uniqueness of disintegration
leaves a matrix coefficient transforming under the stabiliser $SO(3)$.
Conservation removes time components in the rest frame and
\begin{equation}
 \operatorname{Sym}^2(\mathbb R^3)
 =\operatorname{Sym}^2_0(\mathbb R^3)\oplus\mathbb R.
 \label{eq:app-little-group-decomposition}
\end{equation}
The two summands are inequivalent irreducible representations. Schur's
lemma therefore makes the measure scalar on each block and forbids a mixed
block. These two positive coefficients are precisely the common measures
$\dd\rho_2$ and $\dd\rho_0$ in
Eq.~\eqref{eq:body-covariant-spin-disintegration}. The chosen normalisation of $\dd\Omega_s$ fixes their uniqueness independently of the test source.

For a family of TT sources the spin-zero block vanishes and the remaining
matrix pushforward is
$\dd\rho^J_{ab}(s)=q^{(2)}_{ab}(s)\dd\rho_2(s)$ on $(0,\infty)$ as in
Eq.~\eqref{eq:body-source-factorisation}. Changing the test family therefore changes $q^{(2)}$ while preserving
the common scalar measure. Polarisation of diagonal quadratic forms recovers
all off-diagonal entries.

The null orbit has a separate stabiliser and requires its own measure classification. For the complete covariant tensor it is longitudinal and
vanishes on conserved sources. Positivity and local finiteness of its
$00,00$ component near any nonzero null momentum yield
Eq.~\eqref{eq:physical-tensor-ir-integrability}. Assigning zero weight to
$\{0\}$ therefore extends $\rho_2$ to a positive locally finite measure.

To reconstruct the time-ordered distribution continue each timelike
shell using the polynomial numerator built from
$\eta_{\mu\nu}-p_\mu p_\nu/s$ at fixed $s>0$. Two polynomial extensions
that agree on the shell differ by a factor $p^2-s$ which cancels the
propagator and leaves local tensor terms. The infrared inverse moment
above controls the $s^{-2}$ coefficients. UV divergences are removed by
the fixed local subtractions. External TT contraction of this
reconstruction kills the spin-zero and null-shell pieces and
selects the timelike spin-two coefficient. A possible origin measure is
removed on traceless test sources before any momentum projection. Wick
rotation then gives Eq.~\eqref{eq:subtracted-tt-dispersion} whose
nonlocal Lorentzian discontinuity recovers $\dd\rho_2$. The commutator is
the positive-energy measure minus its negative-energy reflection and
determines the retarded response. Positivity in the atom criterion belongs
to the positive-energy Wightman measure.

\subsection{The little group on the null shell}
\label{sec:app-null-shell}

Let $G=SO^+(1,3)$ and $E=\operatorname{Sym}^2(\mathbb C^4)$ and denote
by $D$ the induced finite-dimensional tensor representation of $G$ on
$E$. The following proposition supplies the orbit statement used in
Theorem~\ref{thm:null-shell-exclusion}.

\begin{proposition}[Positive covariant measures on the null orbit]
\label{prop:positive-null-orbit-measure}
Let $M$ be a positive $\operatorname{End}(E)$-valued Radon measure on
$H_0^+$ such that for every Borel set $A\subset H_0^+$
\begin{equation}
 M(gA)=D(g)M(A)D(g)^\dagger,
 \qquad g\in G,
 \label{eq:app-null-measure-covariance}
\end{equation}
and suppose that the two Ward identities
Eq.~\eqref{eq:wightman-measure-ward-identities} hold. Then $M$ has a
smooth positive density relative to $\dd\Omega_0$ and
\begin{equation}
 \dd M_{\mu\nu,\rho\sigma}(p)
 =\alpha_{\rm null}p_\mu p_\nu p_\rho p_\sigma\dd\Omega_0(p),
 \qquad \alpha_{\rm null}\geq0.
 \label{eq:app-classified-null-measure}
\end{equation}
If a positive covariant measure supported at the origin is added it has
the form
\begin{equation}
 \dd M_{\mu\nu,\rho\sigma}^{(0)}(p)
 =\beta_0\eta_{\mu\nu}\eta_{\rho\sigma}\delta^{(4)}(p),
 \qquad \beta_0\geq0.
 \label{eq:app-classified-origin-measure}
\end{equation}
\end{proposition}

\begin{proof}
\medskip\noindent Orbit regularity.\ 
Fix a positive Hermitian inner product on $E$ and introduce the scalar
control measure $\tau(A)=\operatorname{tr}M(A)$. If $\tau(A)=0$ positivity
gives $M(A)=0$. The scalar entries of $M$ are therefore absolutely
continuous with respect to $\tau$. The Radon--Nikodym theorem gives
$M=K\,\dd\tau$ with $K(p)\succeq0$ for $\tau$-almost every $p$.
For fixed $g$ the smallest and largest eigenvalues of
$D(g)^\dagger D(g)$ give positive constants $m_g,M_g$ such that
\begin{equation}
 m_g\tau(A)\leq\tau(gA)\leq M_g\tau(A).
 \label{eq:app-null-control-quasi-invariance}
\end{equation}
Thus $\tau$ is quasi-invariant. The uniqueness of the quasi-invariant
measure class on the homogeneous space $H_0^+\simeq G/E(2)$ makes every
nonzero $\tau$ equivalent to the standard invariant orbit measure
$\dd\Omega_0$~\cite[ch.~2]{Folland2016Harmonic}. If $\tau=0$ the result
already holds with $\alpha_{\rm null}=0$. Otherwise
$M=H\,\dd\Omega_0$ for a locally integrable positive matrix density $H$.

Define the compensated action on tensor-valued measures by
\begin{equation}
 (\Pi(g)M)(A)
 :=D(g)M(g^{-1}A)D(g)^\dagger.
 \label{eq:app-compensated-null-action}
\end{equation}
Equation~\eqref{eq:app-null-measure-covariance} says that $\Pi(g)M=M$.
If $\varphi\in C_c^\infty(G)$ has Haar integral one then
$\Pi(\varphi)M:=\int_G\varphi(g)\Pi(g)M\,\dd g=M$. The action map
$G\times H_0^+\to H_0^+$ is a submersion so in a local section this is the
usual convolution of a distribution with a compactly supported smooth
function. The convolution-smoothing theorem
\cite{Folland2016Harmonic,Hormander1990LPDO1} therefore gives a smooth
density. Thus $H$ has a representative satisfying
\begin{equation}
 H(gp)=D(g)H(p)D(g)^\dagger
 \label{eq:app-null-density-covariance}
\end{equation}
pointwise. This establishes regularity along the null orbit. The transverse
invariant-mass dependence remains encoded in the scalar spectral measure.

\medskip\noindent The little group and the Ward identity.\ 
Choose a null tetrad $\{\ell,n,e_1,e_2\}$ with $\ell\mathbin{\cdot}n=1$ and $e_i\mathbin{\cdot}e_j=-\delta_{ij}$. All other inner products vanish. At $p=\ell$ factor the positive matrix as
$H(\ell)=VV^\dagger$ with $V$ of full column rank. For an element $h$ of
the little group Eq.~\eqref{eq:app-null-density-covariance} implies
$D(h)V=VU(h)$. Multiplication by a left inverse of $V$ shows that
$U(h)U(h)^\dagger=I$ and uniqueness makes $U$ a representation of the
little group. A null rotation is unipotent in the finite tensor
representation $D$. Its restriction to $\operatorname{Ran}V$ is therefore
both unipotent and unitary and hence is the identity.

To determine the common fixed space expand a symmetric contravariant
tensor as
\begin{align}
 u={}&A\ell\ell+Bnn+C(\ell n+n\ell)
       +\sum_i D_i(\ell e_i+e_i\ell)\notag\\
    &+\sum_i E_i(ne_i+e_i n)+\sum_{i,j}F_{ij}e_i e_j,
       \qquad F_{ij}=F_{ji}.
 \label{eq:app-null-tensor-expansion}
\end{align}
Here juxtaposition denotes the ordinary tensor product with both orders
written explicitly for mixed symmetric terms. Apply the relations $N_k\ell=0$ and $N_kn=e_k$ together with $N_ke_i=\delta_{ki}\ell$ to both indices. The coefficients of $ne_k$ and $e_ke_i$ together with those of $\ell\ell$ and $\ell e_i+e_i\ell$ in $N_ku=0$ give
\begin{equation}
 B=0,\qquad E_i=0,\qquad D_i=0,\qquad F_{ij}=-C\delta_{ij}.
 \label{eq:app-null-fixed-tensors}
\end{equation}
The null-rotation fixed tensors are therefore
$u=A\ell\ell+C\eta$.

Let $C_\ell:E\to\mathbb C^4$ be contraction with $\ell$ in the first
index. The measure Ward identity and $M=H\,\dd\Omega_0$ first give
$C_pH(p)=0$ almost everywhere. Smoothness promotes this identity to every
point of the orbit. In particular
\begin{equation}
 0=C_\ell H(\ell)C_\ell^\dagger
   =(C_\ell V)(C_\ell V)^\dagger,
 \label{eq:app-ward-on-positive-factor}
\end{equation}
so $C_\ell V=0$. Conservation therefore holds for every column of $V$ and
removes the metric term from its null-rotation invariant form. Hence
$H(\ell)=\alpha_{\rm null}|\ell\ell\rangle
\langle\ell\ell|$ with $\alpha_{\rm null}\geq0$.
Transitivity of $G$ and Eq.~\eqref{eq:app-null-density-covariance} give
Eq.~\eqref{eq:app-classified-null-measure} on the whole orbit.

\medskip\noindent The origin.\ 
A positive Radon measure supported at one point is $K\delta^{(4)}(p)$ for
a positive matrix $K$. Derivatives of the delta distribution do not occur
in a measure. Covariance gives $K=D(g)KD(g)^\dagger$. Factor
$K=V_0V_0^\dagger$ on its range. The same argument as above produces a
continuous finite-dimensional unitary representation $U_0$ of $G$ on
$\operatorname{Ran}V_0$. Its differential is a homomorphism from the
noncompact simple Lie algebra $\mathfrak{so}(1,3)$ to the compact Lie
algebra $\mathfrak u(r)$. An injection would pull back a positive $\operatorname{Ad}$-invariant inner product and make $\mathfrak{so}(1,3)$ compact. Noncompactness therefore excludes an injection. Simplicity then
makes the differential zero and connectedness of $G$ makes $U_0$
trivial. Rotational invariance sets the mixed components of each column
of $V_0$ to zero and makes its spatial block scalar. Invariance under one
boost fixes the remaining coefficient leaving a multiple of $\eta$.
This proves Eq.~\eqref{eq:app-classified-origin-measure}.
\end{proof}

This orbit classification is independent of the limiting massive source
factors $q^{(2)}(s)$. The one-particle theorem of
ref.~\cite{Weinberg2020Massless} is its representation-theoretic analogue.
Proposition~\ref{prop:positive-null-orbit-measure} is formulated directly
for the Wightman measure.

\subsection{Infrared limit as a quadratic-form statement}

For several physical operators write $\mathrm d\boldsymbol\rho_2$ and
$\boldsymbol C_2$ for the matrix-valued invariant measure and response.
These operator indices differ from the test-source indices of
Eq.~\eqref{eq:app-translation-spectral-measure}. The atom extraction is first established on
each quadratic form $v^\dagger\boldsymbol C_2v$. With
\begin{equation}
 \dd\nu_v(s)
 =\frac{v^\dagger\dd\boldsymbol\rho_{2,c}(s)v}{(s+Q_*^2)^{n_{\rm sub}}},
 \label{eq:app-quadratic-continuum-measure}
\end{equation}
the dominated-convergence argument of Proposition~\ref{prop:tt-atom-extraction} applies to each quadratic form.
Polarisation of the resulting quadratic-form limit gives the matrix limit.
Positivity belongs to the physical matrix quadratic form. Individual
off-diagonal coefficients need not be positive measures.

Local integrability at the origin is essential to the
dominated-convergence argument. Without this integrability the infrared limit requires an estimate of Tauberian or limiting-absorption type. For the covariant tensor considered here
Eq.~\eqref{eq:physical-tensor-ir-integrability} supplies the required
integrability.

\section{Composite-source flow and connected response}
\label{app:composite-frg}

The Schur identity differentiates the source together with the propagating variables and observable in one generating functional. The Legendre transform and source derivatives fix the signs of the reduced source and observable maps in Section~\ref{sec:background-source-frg}.

\subsection{Source derivatives of the Wetterich equation}

Condense the symmetric tensor source and its spacetime argument into an
index $i$ and define
\begin{equation}
 A_i:=\frac{\delta\Gamma_k^{(2,0)}}{\delta h_i}
      =\Gamma_{k,i}^{(2,1)},
 \qquad
 A_{ij}:=\frac{\delta^2\Gamma_k^{(2,0)}}
                 {\delta h_i\delta h_j}
      =\Gamma_{k,ij}^{(2,2)}.
 \label{eq:app-source-vertices}
\end{equation}
For a source-independent regulator
$\delta_iG_k=-G_kA_iG_k$. Applying this identity twice to the
Wetterich equation gives
\begin{align}
 \partial_t\Gamma_{k,ij}^{(0,2)}
 =\frac12\STr\bigl[&G_kA_jG_kA_iG_k\partial_tR_k
 +G_kA_iG_kA_jG_k\partial_tR_k
 \notag\\[-1mm]
 &-G_kA_{ij}G_k\partial_tR_k\bigr],
 \label{eq:app-two-source-flow}
\end{align}
with the superfield signs contained in the supertrace and vertex ordering.
The first two terms contain separate insertions while $A_{ij}$ is the
renormalised coincident double insertion required by the contact and Ward
identities. Their flows involve higher field vertices. Finite closures of
this hierarchy are compared with the full response in
Appendix~\ref{app:fredholm}.

\subsection{Legendre transform and resolvent response}

Write
\begin{equation}
 \Gamma_k[\varphi,h]
 =j\cdot\varphi-W_k[j,h]-\Delta S_k[\varphi],
 \qquad \varphi=W_{k,j}.
 \label{eq:app-legendre-transform}
\end{equation}
At fixed $\varphi$ the first derivative is $\Gamma_{k,h}=-W_{k,h}$ and
$\delta_hj=-W_{k,jj}^{-1}W_{k,jh}$ while
$W_{k,jj}=G_k=(\Gamma_{k,\varphi\varphi}+R_k)^{-1}$.
Differentiation therefore gives
$\Gamma_{k,hh}=-W_{k,hh}+W_{k,hj}G_k^{-1}W_{k,jh}$.
Using $\Gamma_{k,h\varphi}=-W_{k,hj}G_k^{-1}$ and its transposed
mixed derivative yields
\begin{equation}
 W_{k,hh}
 =-\Gamma_{k,hh}
  +\Gamma_{k,h\varphi}G_k\Gamma_{k,\varphi h},
 \label{eq:app-source-schur-final}
\end{equation}
The second term propagates through the fundamental-field sector. It
vanishes when the mixed source--field Hessians vanish. The converse
holds for a strictly positive propagator on the coupled sector.

At vanishing fundamental source the effective average field satisfies
\begin{equation}
 \mathcal E_k(\varphi,h)
 =\Gamma_{k,\varphi}+R_k\varphi=0.
 \label{eq:app-stationarity-equation}
\end{equation}
Choose an analytic invertible preconditioner $\mathcal A_{0,k}$ and rewrite
this equation as the fixed-point problem
\begin{equation}
 \varphi=\mathfrak F_k(\varphi;h),
 \qquad
 \mathfrak F_k
 =\varphi-\mathcal A_{0,k}^{-1}\mathcal E_k.
 \label{eq:app-preconditioned-fixed-point}
\end{equation}
Linearising in a tensor-source variation gives the response operators in
Eq.~\eqref{eq:canonical-response-tuple}. In particular
\begin{align}
 (I-\mathcal B_k)^{-1}\mathsf S_k
 &=-\mathcal A_k^{-1}\Gamma_{k,\varphi h},
 \notag\\
 \mathsf R_k(I-\mathcal B_k)^{-1}\mathsf S_k
 &=\Gamma_{k,h\varphi}G_k\Gamma_{k,\varphi h}.
 \label{eq:app-preconditioner-cancellation}
\end{align}
The preconditioner cancels from the response though the compactness and
approximation estimates refer to the chosen kernel representation.
Proposition~\ref{prop:native-same-source} applies the same construction
to enlarged response coordinates. For a finite projected hierarchy it
gives the response of the closure with errors controlled as in
Appendix~\ref{app:fredholm}.

\subsection{Complete channel reduction}

For a renormalised block $\cO_{R,i}=Z_{ij}\cO_{B,j}$ the sources transform
contragrediently so the flow produces a matrix response. Internal mixing survives the external TT projection. The renormalised block includes the BRST and EOM representatives required by its counterterms. The physical quotient is formed after renormalisation.

Let the complete response variables be $X=(X_P,X_{\overline P})$ where
$P$ is the physical TT block and $\overline P=I-P$ its complement.
Linearisation gives
\begin{equation}
 \begin{pmatrix}\delta X_P\\ \delta X_{\overline P}\end{pmatrix}
 =
 \begin{pmatrix}
  \cB_{PP}&\cB_{P\overline P}\\
  \cB_{\overline PP}&\cB_{\overline P\overline P}
 \end{pmatrix}
 \begin{pmatrix}\delta X_P\\ \delta X_{\overline P}\end{pmatrix}
 +\begin{pmatrix}\mathsf S_P\\ \mathsf S_{\overline P}\end{pmatrix}\delta J.
 \label{eq:app-response-block-system}
\end{equation}
Write the full measured response as
$\delta\mathcal O=C_{\mathrm{dir}}\delta J+\mathsf R_P\delta X_P
+\mathsf R_{\overline P}\delta X_{\overline P}$.
Assume $D:=I_{\overline P}-\cB_{\overline P\overline P}$ is invertible.
The reduced kernel and source map together with the observable map and direct term are
\begin{align}
 \cB_{\mathrm{eff}}
 &:=\cB_{PP}+\cB_{P\overline P}D^{-1}\cB_{\overline PP},
 \notag\\
 \mathsf S_{\mathrm{eff}}
 &:=\mathsf S_P+\cB_{P\overline P}D^{-1}\mathsf S_{\overline P},
 \notag\\
 \mathsf R_{\mathrm{eff}}
 &:=\mathsf R_P+\mathsf R_{\overline P}D^{-1}\cB_{\overline PP},
 \notag\\
 C_{\mathrm{dir,eff}}
 &:=C_{\mathrm{dir}}+\mathsf R_{\overline P}D^{-1}\mathsf S_{\overline P}.
 \label{eq:app-complete-feshbach-tuple}
\end{align}
Wherever the remaining inverse exists the full correlator is
\begin{align}
 C
 &=C_{\mathrm{dir}}+\mathsf R(I-\cB)^{-1}\mathsf S
 \notag\\
 &=C_{\mathrm{dir,eff}}
   +\mathsf R_{\mathrm{eff}}
    (I_P-\cB_{\mathrm{eff}})^{-1}\mathsf S_{\mathrm{eff}}.
 \label{eq:app-complete-reduced-response}
\end{align}
Solving the second block equation and substituting it into the first and into
$\delta\mathcal O$ gives Eqs.~\eqref{eq:app-complete-feshbach-tuple}--
\eqref{eq:app-complete-reduced-response}. Analyticity is preserved while
$D$ remains invertible and the remaining inverse has the Laurent expansion
of Section~\ref{sec:atom-kernel} under its Fredholm hypotheses.

Equivalently before preconditioning write the linearised equations as
\begin{equation}
 L=\begin{pmatrix}A&B\\ C&D\end{pmatrix},\qquad
 L\begin{pmatrix}x_P\\x_{\overline P}\end{pmatrix}
 =\begin{pmatrix}S_P\\S_{\overline P}\end{pmatrix}h,\qquad
 \mathcal O=C_{\rm dir}h+R_Px_P+R_{\overline P}x_{\overline P}.
 \label{eq:native-schur-blocks}
\end{equation}
If $D$ is invertible direct substitution gives
\begin{align}
 L_{\rm eff}&=A-BD^{-1}C,&
 S_{\rm eff}&=S_P-BD^{-1}S_{\overline P},\notag\\
 R_{\rm eff}&=R_P-R_{\overline P}D^{-1}C,&
 C_{\rm dir,eff}&=C_{\rm dir}+R_{\overline P}D^{-1}S_{\overline P}.
 \label{eq:native-schur-tuple}
\end{align}
The signs differ from the iteration-kernel entries in
Eq.~\eqref{eq:app-complete-feshbach-tuple} because the off-diagonal blocks
of $I-\mathcal B$ are negative. The complete response is the same. On an
analytic invertibility domain
\begin{equation}
 L_{\rm eff}'=A'-B'D^{-1}C
       +BD^{-1}D'D^{-1}C-BD^{-1}C'.
 \label{eq:native-schur-slope}
\end{equation}
Here $(D^{-1})'=-D^{-1}D'D^{-1}$. Differentiating the source and observable
maps in Eq.~\eqref{eq:native-schur-tuple} retains the same dependence on the
eliminated block.

In the FRG representation the renormalised physical response is obtained
as $k\to0$. This regulator limit is distinct from the Gaussian spectral
weight and the finite correlation scale of the compact construction.
On the Euclidean axis TT projection is performed at $Q^2>0$ and the
continuation uses
\begin{equation}
 z=Q^2\quad\longrightarrow\quad z=-p^2-\iu0.
 \label{eq:app-continuation-convention}
\end{equation}
The discontinuity of the nonlocal boundary value determines the Wightman
measure. For finite Euclidean data the uncertainty of analytic continuation adds to the regulator and truncation errors.

\section{Fredholm pole and truncation control}
\label{app:fredholm}

Solution branches and pole residues are compared on a common space. The spectral input is analytic perturbation theory
\cite[Ch.~VII]{Kato1995,Chatelin1983}. The residue bounds retain the motion
of the critical root and both source maps.

\subsection{Laurent expansion at a simple eigenvalue}

Let $\cB(z)$ be analytic and compact on a neighbourhood of the origin and
suppose $\lambda(z)$ is an isolated algebraically simple eigenvalue with
$\lambda(0)=1$ and $\lambda'(0)\neq0$. Its Riesz projector is
\begin{equation}
 \Pi(z)=|r(z)\rangle\langle\ell(z)|
 =\frac1{2\pi\iu}\oint_\gamma
  (\zeta-\cB(z))^{-1}\,\dd\zeta,
 \label{eq:app-riesz-projector}
\end{equation}
where $\gamma$ encloses no other spectrum and the analytic left and right
eigenvectors are normalised by $\langle\ell|r\rangle=1$. Analytic
Fredholm theory gives
\begin{equation}
 [I-\cB(z)]^{-1}
 =\frac{\Pi(z)}{1-\lambda(z)}+\mathcal R(z)
 =-\frac{\Pi(0)}{\lambda'(0)z}+\order{1},
 \label{eq:app-resolvent-decomposition}
\end{equation}
with $\mathcal R$ analytic near zero. Insertion between
$\mathsf R_{\TT}$ and $\mathsf S_{\TT}$ proves
Eqs.~\eqref{eq:tt-overlap-conditions} and
\eqref{eq:tt-algebraic-residue}. If either overlap vanishes the operator
resolvent is singular but the measured TT correlator is regular in that
direction.

If $1\notin\spec\cB(0)$ the response is regular near the origin.
The complete Schur reduction in
Eqs.~\eqref{eq:app-complete-feshbach-tuple}--
\eqref{eq:app-complete-reduced-response} preserves this response while
including shifts caused by complementary channels.

\subsection{Solution branches, critical roots and residues}

Let $\cX$ be the exact response Banach space. Let $\mathcal Y$ be a fixed
Hilbert space of source amplitudes with its dual identified by its inner
product. The scalar criterion uses $\mathcal Y=\mathbb C$.
The Riesz projector $\Pi(z)$ acts on $\cX$ whereas the tensor projector
$\Pi^{(2)}(Q)$ acts on external Lorentz indices. For each finite space $\cX_N$ use the maps in
Eq.~\eqref{eq:declared-tt-truncation} and lift all finite objects to the same
comparison spaces
\begin{equation}
 \widetilde{\cB}_N(z):=\iota_N\cB_N(z)P_N,
 \qquad
 \widetilde{\mathsf S}_N(z):=\iota_N\mathsf S_N(z),
 \qquad
 \widetilde{\mathsf R}_N(z):=\mathsf R_N(z)P_N.
 \label{eq:app-lifted-truncation-maps}
\end{equation}
Thus $\widetilde{\cB}_N\in\mathcal L(\cX)$ and $\widetilde{\mathsf S}_N\in\mathcal L(\mathcal Y,\cX)$ with $\widetilde{\mathsf R}_N\in\mathcal L(\cX,\mathcal Y)$. On a compact set $K$ in
the analytic $z$-domain define the common-space kernel error
\begin{equation}
 \epsilon_{\cB,N}
 :=\sup_{z\in K}
 \|\widetilde{\cB}_N(z)-\cB(z)\|_{\mathcal L(\cX)}.
 \label{eq:app-operator-norm-convergence}
\end{equation}
The convergence hypothesis $\epsilon_{\cB,N}\to0$ therefore refers to
operators on the same space through the fixed comparison maps.
Let $\gamma$ enclose the simple exact eigenvalue uniformly on $K$ and exclude zero and all other exact spectrum and let
$\Pi_N$ be the contour projector of $\widetilde{\cB}_N$ on the same contour.
Set
\begin{equation}
 M_\gamma
 :=\sup_{z\in K,\,\zeta\in\gamma}
   \|(\zeta-\cB(z))^{-1}\|,
 \qquad L_\gamma:=\operatorname{length}(\gamma).
 \label{eq:app-resolvent-constants}
\end{equation}
If $M_\gamma\epsilon_{\cB,N}<1$ the resolvent identity and a Neumann series give
\begin{equation}
 \|\Pi_N(z)-\Pi(z)\|
 \leq\frac{L_\gamma}{2\pi}
 \frac{M_\gamma^2\epsilon_{\cB,N}}
      {1-M_\gamma\epsilon_{\cB,N}}.
 \label{eq:app-riesz-stability-bound}
\end{equation}
The same Neumann bound holds along
$\cB+t(\widetilde{\cB}_N-\cB)$ $0\leq t\leq1$ so the contour
projectors have constant rank one. Consequently an isolated simple
eigenspace is stable under operator-norm convergent truncations.

Residue convergence also involves the following errors on the same spaces
\begin{align}
 \epsilon_{\partial\cB,N}
 &:=\sup_{z\in K}
 \|\partial_z\widetilde{\cB}_N(z)-\partial_z\cB(z)\|_{\mathcal L(\cX)},
 \notag\\
 \epsilon_{\mathsf S,N}
 &:=\sup_{z\in K}
 \|\widetilde{\mathsf S}_N(z)-\mathsf S(z)\|_{\mathcal L(\mathcal Y,\cX)},
 \notag\\
 \epsilon_{\mathsf R,N}
 &:=\sup_{z\in K}
 \|\widetilde{\mathsf R}_N(z)-\mathsf R(z)\|_{\mathcal L(\cX,\mathcal Y)},
 \notag\\
 \epsilon_{\mathcal C_{\rm reg},N}
 &:=\sup_{z\in K}
 \|\mathcal C_{{\rm reg},N}(z)-\mathcal C_{\rm reg}(z)\|_{\mathcal L(\mathcal Y)}.
 \label{eq:app-residue-comparison-errors}
\end{align}
The derivative and source-map errors control the residue. Convergence of
the regular part also uses $\epsilon_{\mathcal C_{\rm reg},N}\to0$.
All are evaluated on the actual solution branches. For nonnormal kernels
the resolvent constant $M_\gamma$ can be large even for separated
eigenvalues~\cite{Chatelin1983,TrefethenEmbree2005}.

\begin{proposition}[Controlled critical branch]
\label{prop:app-actual-branch-certificate}
\label{prop:app-root-residue-certificate}
Let $H_\tau(Y;z)=Y-\mathfrak F_\tau(Y;z)$ be the coarse stationarity
residual and $\overline Y_*(z)$ a transported reference branch. One
sufficient condition for a nearby solution at each $z$ is that $H_\tau$
is continuously differentiable near the closed Banach-space ball
$\overline B(\overline Y_*,r_Y)$ and for a bounded isomorphism $A$
\begin{equation}
 \begin{gathered}
 \|A^{-1}H_\tau(\overline Y_*)\|\leq b,
 \qquad
 \sup_{Y\in\overline B(\overline Y_*,r_Y)}
       \|I-A^{-1}D_YH_\tau(Y)\|\leq q<1,\\
 b+qr_Y\leq r_Y.
 \end{gathered}
 \label{eq:app-branch-contraction-bounds}
\end{equation}
Then the solution in the ball is unique and
\begin{equation}
 \|Y_{\tau,*}-\overline Y_*\|\leq\frac{b}{1-q}.
 \label{eq:app-certified-branch-displacement}
\end{equation}
Uniform bounds and analytic data give an analytic solution branch.
Alternatively exact matching $H_\tau(\overline Y_*(z);z)=0$ gives the
branch directly including at singular stationarity points.

On a solution branch established by either route or by a separate
existence argument suppose the preceding response families and rank-one
contour projectors are defined on a neighbourhood of the closed disk
$D_r=\{z:|z|\leq r\}\subset K$. Let $\lambda(0)=1$ and set $\alpha=|\lambda'(0)|>0$. Assume
\begin{equation}
 m_r:=\inf_{z\in D_r}
 \left|\frac{1-\lambda(z)}{z}\right|>0,
 \label{eq:app-critical-root-lower-bound}
\end{equation}
where the quotient at zero is its analytic value $-\lambda'(0)$.
Assume the uniform bound
$\sup_{D_r}|\lambda_N-\lambda|\leq e_{\lambda,N}<m_rr$.
Then $1-\lambda_N$ has exactly one zero $z_N$ in the disk counted with
multiplicity and
\begin{equation}
 |z_N|\leq d_{z,N}:=e_{\lambda,N}/m_r<r.
 \label{eq:app-certified-critical-root}
\end{equation}
For example admissible choices from the common-space operator estimates are
\begin{align}
 e_{\lambda,N}
 &:=\epsilon_{\cB,N}(\Pi_*+d_{\Pi,N})
       +2B_*d_{\Pi,N},\\
 e_{\lambda',N}
 &:=\epsilon_{\partial\cB,N}(\Pi_*+d_{\Pi,N})
       +2B'_*d_{\Pi,N},
 \label{eq:app-eigenvalue-and-slope-errors}
\end{align}
where $d_{\Pi,N}$ is the right side of
Eq.~\eqref{eq:app-riesz-stability-bound} and the constants are $\Pi_*:=\sup_{D_r}\|\Pi\|$ and $B_*:=\sup_{D_r}\|\cB\|$ together with $B'_*:=\sup_{D_r}\|\partial_z\cB\|$.
The second line bounds $\sup_{D_r}|\lambda_N'-\lambda'|$.

Define the numerators
$\mathcal N=\mathsf R\Pi\mathsf S$ and
$\mathcal N_N=\widetilde{\mathsf R}_N\Pi_N
\widetilde{\mathsf S}_N$ and set $R_*:=\sup_{D_r}\|\mathsf R\|$ and $S_{N,*}:=\sup_{D_r}\|\widetilde{\mathsf S}_N\|$ with $\Pi_{N,*}:=\sup_{D_r}\|\Pi_N\|$.
A same-point numerator error and its cross-root extension are
\begin{align}
 e_{\mathrm{num},N}
 &:=\epsilon_{\mathsf R,N}\Pi_{N,*}S_{N,*}
   +R_*d_{\Pi,N}S_{N,*}+R_*\Pi_*\epsilon_{\mathsf S,N},\\
 D_{\mathrm{num},N}
 &:=e_{\mathrm{num},N}+L_{\mathrm{num},N}d_{z,N},\\
 D_{\mathrm{slope},N}
 &:=e_{\lambda',N}+L_{\lambda',N}d_{z,N},
 \label{eq:app-cross-root-errors}
\end{align}
with constants
$L_{\mathrm{num},N}\geq\sup_{D_r}\|\mathcal N_N'\|$ and
$L_{\lambda',N}\geq\sup_{D_r}|\lambda_N''|$.
Assume $D_{\mathrm{slope},N}<\alpha$. Then
\begin{equation}
 |\lambda_N'(z_N)|\geq\alpha-D_{\mathrm{slope},N}>0.
 \label{eq:app-uniform-transversality}
\end{equation}
If the complementary resolvents and direct terms are analytic on $D_r$
the residues at the respective roots satisfy
\begin{equation}
 \begin{gathered}
 \mathcal Z=-\frac{\mathcal N(0)}{\lambda'(0)},
 \qquad
 \mathcal Z_N=-\frac{\mathcal N_N(z_N)}{\lambda_N'(z_N)},\\
 \|\mathcal Z_N-\mathcal Z\|
 \leq\frac{D_{\mathrm{num},N}}{\alpha-D_{\mathrm{slope},N}}
 +\frac{\|\mathcal N(0)\|D_{\mathrm{slope},N}}
 {\alpha(\alpha-D_{\mathrm{slope},N})}.
 \end{gathered}
 \label{eq:app-certified-cross-root-residue}
\end{equation}
For a scalar response $|\mathcal N(0)|>D_{\mathrm{num},N}$ also gives
$\mathcal N_N(z_N)\neq0$.
\end{proposition}

\begin{proof}
For the solution estimate $T(Y)=Y-A^{-1}H_\tau(Y)$ is a contraction
mapping the ball into itself since
$\|T(Y)-\overline Y_*\|\leq b+q\|Y-\overline Y_*\|\leq r_Y$.
Its fixed point satisfies $\delta_Y\leq b+q\delta_Y$ and uniform
analytic iterations give the parameter dependence. Exact matching
is substitution into $H_\tau=0$.

The sole zero of $1-\lambda$ in $D_r$ is simple and lies at the origin. On its boundary
$|1-\lambda(z)|\geq m_rr>e_{\lambda,N}$ so Rouch\'e's theorem
gives exactly one zero of $1-\lambda_N$. At that zero
$m_r|z_N|\leq|1-\lambda(z_N)|\leq e_{\lambda,N}$.
The identities $\lambda=\operatorname{tr}(\cB\Pi)$ and
$\lambda'=\operatorname{tr}(\cB'\Pi)$ together with their truncated
counterparts give Eq.~\eqref{eq:app-eigenvalue-and-slope-errors}. Here one
uses $|\operatorname{tr}T|\leq\operatorname{rank}(T)\|T\|$ for the finite-rank
differences. Telescoping the three factors in $\mathcal N_N-\mathcal N$ and
integrating along the segment from zero to $z_N$ then gives
\begin{equation}
 \begin{gathered}
 \|\mathcal N_N(z_N)-\mathcal N(0)\|\leq D_{\mathrm{num},N},
 \\
 |\lambda_N'(z_N)-\lambda'(0)|\leq D_{\mathrm{slope},N}.
 \end{gathered}
 \label{eq:app-across-root-comparison}
\end{equation}
The lower slope bound is the reverse triangle inequality. Subtracting the
two quotient residues gives Eq.~\eqref{eq:app-certified-cross-root-residue}.
\end{proof}

At a singular stationarity point $I-A^{-1}D_YH_\tau$ has eigenvalue
one so the contraction alternative is unavailable. Exact branch matching
or a separate existence argument supplies the branch there. Once it is
known a state-variable Lipschitz bound on $D_Y\mathfrak F_\tau$ and
Eq.~\eqref{eq:app-certified-branch-displacement} give the displacement
term in Eq.~\eqref{eq:paper1-kernel-defect-bound}.

Cauchy's estimates provide derivative bounds from uniform convergence
on a strictly larger analytic disk. Under a real-analytic symmetry the
unique root in the symmetric disk is real. Otherwise $z_N$ may be
complex. The proposition compares the finite and exact residues.
Their physical spectral interpretation is supplied by
Theorem~\ref{thm:tt-fredholm-spectral-bridge}.

\begin{proposition}[Telescoping operator-error bound]
\label{prop:app-telescoping-error-budget}
Let $\cB_N^{[j]}(z)\in\mathcal L(\cX)$ for $j=0,\ldots,4$ with
$\cB_N^{[0]}=\cB$ and $\cB_N^{[4]}=\widetilde{\cB}_N$. The successive operators implement vertex truncation and momentum discretisation followed by response-complement approximation and closure. Define
\begin{equation}
 \delta_{j,N}
 :=\sup_{z\in K}
 \|\cB_N^{[j]}(z)-\cB_N^{[j-1]}(z)\|_{\mathcal L(\cX)},
 \qquad j=1,\ldots,4.
 \label{eq:app-link-error-definitions}
\end{equation}
Then
\begin{equation}
 \epsilon_{\cB,N}
 \leq \delta_{1,N}+\delta_{2,N}+\delta_{3,N}+\delta_{4,N}.
 \label{eq:app-total-kernel-error}
\end{equation}
For a last link implemented on the finite response space $\cX_N$ let
$\mathcal C_{\tau,N}$ be boundedly invertible on that space. Define the
finite response operators and their common-space lifts by
\begin{equation}
 \begin{aligned}
 B_{3,N}&:=\mathcal C_{\tau,N}\cB_N^{\rm base}
                    \mathcal C_{\tau,N}^{-1},
 &B_{4,N}&:=\cB_{\tau,N}(Y_{\tau,N,*}),\\
 \cB_N^{[3]}&:=\iota_N B_{3,N}P_N,
 &\cB_N^{[4]}&:=\iota_N B_{4,N}P_N.
 \end{aligned}
 \label{eq:app-closure-link-lifts}
\end{equation}
At the transported reference point set
$DE_{\tau,N}=\mathcal C_{\tau,N}\cB_N^{\rm base}
-\cB_{\tau,N}(\overline Y_{N,*})\mathcal C_{\tau,N}$.
Let $Y_{\tau,N,*}$ be a solution branch with displacement bound
$\delta_{Y,N}\geq\|Y_{\tau,N,*}-\overline Y_{N,*}\|_{\cX_N}$.
Let $L_{\tau,N}$ bound the state-variable Lipschitz constant of
$\cB_{\tau,N}$ in $\mathcal L(\cX_N)$ with respect to this finite-space
state norm. Then
\begin{equation}
 \delta_{4,N}
 \leq\|\iota_N\|\,\|P_N\|\sup_{z\in K}\bigl(
 \|DE_{\tau,N}(z)\|_{\mathcal L(\cX_N)}
 \|\mathcal C_{\tau,N}^{-1}(z)\|_{\mathcal L(\cX_N)}
 +L_{\tau,N}(z)\delta_{Y,N}(z)\bigr).
 \label{eq:app-closure-link-bound}
\end{equation}
The embedding and projection norms account for the passage from the
finite-space error to $\mathcal L(\cX)$.
Consequently let $M_\gamma\leq\overline M_\gamma$ and
$\delta_{j,N}\leq\overline\delta_{j,N}$ be upper bounds
and put $\overline\epsilon_N=\sum_j\overline\delta_{j,N}$. If
$\overline M_\gamma\overline\epsilon_N<1$ then the computable estimate is
\begin{equation}
 \|\Pi_N(z)-\Pi(z)\|
 \leq\frac{L_\gamma}{2\pi}
 \frac{\overline M_\gamma^2\overline\epsilon_N}
      {1-\overline M_\gamma\overline\epsilon_N}.
 \label{eq:app-computable-riesz-bound}
\end{equation}
\end{proposition}

\begin{proof}
The lifted telescoping identity gives Eq.~\eqref{eq:app-total-kernel-error}
by the triangle inequality. Differentiation of the closure defect followed by composition with $\mathcal C_{\tau,N}^{-1}$ and inclusion of the actual-branch displacement gives
Eq.~\eqref{eq:app-closure-link-bound}. The lift norms account for the change
of space. Exact branch matching removes the displacement term and the final
Riesz estimate is Eq.~\eqref{eq:app-riesz-stability-bound}.
\end{proof}

Applying the same decomposition to $\partial_z\cB$ and the source maps $\mathsf S$ and $\mathsf R$ with their solution-branch dependence supplies the
uniform estimates needed in
Proposition~\ref{prop:app-root-residue-certificate}.
The resulting bounds depend on the intermediate operators and comparison norms together with the contour resolvent constant of the chosen approximation.
Appendix~\ref{app:benchmarks} illustrates the root and residue shifts
in an exactly solvable finite response.

\paragraph{Errors supplied by a finite evolution.}
When the response is obtained from an evolution its approximation is
compared with the full sourced equation on a common interval. For
semilinear spectral generators the stronger spectral bound controls
both the omitted initial modes and the discarded nonlinear source in a
weaker forcing norm. The Duhamel stability estimate then propagates
these errors to the state
\cite{Guo2026GeneratedDynamics}.

Let $\mathsf L\geq0$ be bounded and self-adjoint and put $f(\sigma)=e^{-\sigma\mathsf L}f_0$. Let $w$ be an absolutely
continuous finite approximation on $[0,T_{\rm obs}]$. If
$\|\mathsf L-\widehat{\mathsf L}_N\|\leq\delta_{L,N}$ put
$r_N=\partial_\sigma w+\widehat{\mathsf L}_Nw$. Contractivity and
variation of constants give
\[
 \|f(\sigma)-w(\sigma)\|
 \leq\|f_0-w(0)\|+\int_0^\sigma\|r_N(s)\|\,\dd s
       +\delta_{L,N}\int_0^\sigma\|w(s)\|\,\dd s
 =:E_N(\sigma).
\]
For a bounded observation $C_{\rm obs}$ with
$\|(C_{\rm obs}-\widehat C_{\rm obs})w\|
\leq\eta_C\|w\|$ it follows that
\[
 \|C_{\rm obs}f-\widehat C_{\rm obs}w\|
 \leq\|C_{\rm obs}\|E_N+\eta_C\|w\|.
\]
The operator discrepancy acts on the same computed path as the finite
residual and the observation has its own approximation error. Finite matrix and path enclosures bound these terms separately. Source derivatives are controlled through the differentiated
equations including the variation of the generator and state domain together with the observation. Uniform operator estimates for the resulting response maps
supply Eqs.~\eqref{eq:app-operator-norm-convergence} and
\eqref{eq:app-residue-comparison-errors}. Estimates on a larger analytic
parameter domain also control their $z$ derivatives by Cauchy's formula.
The critical-root comparison above then transfers these bounds to the
residue at the actual finite root.

\subsection{Spectral accumulation at zero}

If a continuum begins at $s=0$ the continued response generally has a
nonmeromorphic threshold singularity or boundary behaviour. Such behaviour can prevent a common analytic neighbourhood for $\cB(z)$ and $\mathcal C_{\rm reg}(z)$ and the complementary resolvent. Atom extraction at that boundary uses spectral integrability rather than the isolated Laurent decomposition of Eq.~\eqref{eq:app-resolvent-decomposition}.

An atom can coexist with the threshold. For example a measure may have
the form
\begin{equation}
 \dd\rho_{\rm resp}(s)=Z_0^{(\rm resp)}\delta_0(\dd s)
 +f(s)\theta(s)\dd s,
 \label{eq:app-atom-at-threshold}
\end{equation}
with both terms nonzero. Proposition~\ref{prop:tt-atom-extraction} can still
identify $Z_0^{(\rm resp)}$ when the continuum satisfies its integrability hypotheses
although the isolated Fredholm theorem does not apply. A kernel analysis at the threshold uses weighted resolvent estimates or a limiting-absorption principle.

\section{Spectral examples, compact sources, and finite approximations}
\label{app:benchmarks}

The examples distinguish a gapped continuum from an atomic contribution. They also exhibit a hidden-channel shift under Schur reduction and collective concentration of finite-volume levels. The compact calculations supply the
spectral sum and matter vertices used in Sections~\ref{sec:native-tt-spectral-branch}
and~\ref{sec:native-compact-tt-visibility}.

\subsection{Free scalar stress tensor}

For a free real scalar of mass $m$ the stress tensor is bilinear and
its first non-vacuum matrix element creates two particles. TT contraction
gives $2\varepsilon^{\mu\nu}p_{1\mu}p_{2\nu}$ independent of scalar
improvements. The two-body phase-space integral yields
\begin{equation}
 \rho_2^{\rm scalar}(s)
 =\frac{s^2}{960\pi^2}
  \left(1-\frac{4m^2}{s}\right)^{5/2}
  \theta(s-4m^2),
 \label{eq:benchmark-scalar-density}
\end{equation}
with one-particle normalisation
$\langle p|p'\rangle=2E_p(2\pi)^3\delta^3(\boldsymbol p-\boldsymbol p')$
and the orbit measure of Section~\ref{sec:tt-wightman-measure}. The
angular average for a unit TT polarisation is $2/15$ and the
identical-particle factor is $1/2!$. Thus $Z_0=0$ both for $m>0$ and
for $m=0$ where the continuum has density proportional to $s^2$ down
to zero.

Massless free fermion and gauge-field stress tensors likewise have
positive TT densities proportional to $s^2$ with normalisations fixed
by their two-point functions~\cite{OsbornPetkou1994}. By contrast for
$Z_{\rm ref},A,M^2>0$ the reference measure
$\dd\rho_{\rm ref}=Z_{\rm ref}\delta_0+A s^2e^{-s/M^2}\dd s$
defines $C_{\rm ref}(Q^2):=\int_0^\infty
\dd\rho_{\rm ref}(s)/(s+Q^2)$ and has atomic weight
$\lim_{Q^2\downarrow0}Q^2C_{\rm ref}(Q^2)=Z_{\rm ref}$
unchanged by finite Gaussian averaging even though its continuum reaches
zero.

\subsection{Finite response and volume limits}

Let $P=\diag(1,0)$ retain the first channel of the nonnormal family
with $|\delta|<1/4$
\begin{equation}
 \cB_\delta(z)=
 \begin{pmatrix}
  5/4-z&1/2\\
  -1/4+\delta&1/2
 \end{pmatrix},
 \qquad
 \|\cB_\delta-\cB_0\|_2=|\delta|.
 \label{eq:benchmark-two-channel-kernel}
\end{equation}
Direct projection gives the false pole $z=1/4$ whereas Schur reduction
gives $\cB_{\rm eff}=1-z+\delta$. For
$\mathsf S=(s_1,s_2)^{\mathsf T}$ and $\mathsf R=(r_1,r_2)$ complete
elimination yields
\begin{equation}
 C(z)=C_{\rm dir}+2r_2s_2
 +\frac{[r_1+(-1/2+2\delta)r_2](s_1+s_2)}{z-\delta}.
 \label{eq:benchmark-root-shift}
\end{equation}
The critical root moves by $\delta$ and the residue contains both
overlaps. At $\delta=0$ the normalised right and left eigenvectors $|r\rangle=(2,-1)^{\mathsf T}$ and $\langle\ell|=(1,1)$ have
$\lambda'(0)=-2$ reproducing the numerator in
Eq.~\eqref{eq:benchmark-root-shift}. For real maps its sign is the
sign of the overlap product. Either vanishing factor removes the pole.
The choice $\mathsf S=(1,2)^{\mathsf T}$ with $\mathsf R=(3,4)$ and $C_{\rm dir}=5$ gives $21+3/z$ including the complementary direct term.

\paragraph{Cumulative spectral weight.}

For a finite-volume regularisation with discrete spectrum write
\begin{equation}
 \dd\rho_{2,\ell_{\rm box}}(s)
 =\sum_n Z_n(\ell_{\rm box})\delta_{s_n(\ell_{\rm box})}(\dd s).
 \label{eq:benchmark-finite-volume-measure}
\end{equation}
A sequence of low-lying levels can mimic $Z_0/Q^2$ over a restricted
momentum range. The limiting atom is determined by cumulative spectral
weight which also includes concentration of individually vanishing levels.

\begin{proposition}[Cumulative finite-volume atom criterion]
\label{prop:finite-volume-cumulative-atom}
Let $\rho_{2,\ell_{\rm box}}$ be positive Radon measures on $[0,\infty)$ converging locally vaguely to a locally finite measure $\rho_2$. Thus their integrals converge against every continuous compactly supported test function.
For each $\epsilon>0$ with $\rho_2(\{\epsilon\})=0$
\begin{equation}
 \lim_{\ell_{\rm box}\to\infty}\rho_{2,\ell_{\rm box}}([0,\epsilon])
 =\rho_2([0,\epsilon]).
 \label{eq:finite-volume-cumulative-convergence}
\end{equation}
Consequently along any decreasing sequence of such continuity radii
\begin{equation}
 Z_0=\rho_2(\{0\})
 =\lim_{\epsilon\downarrow0}\lim_{\ell_{\rm box}\to\infty}
   \rho_{2,\ell_{\rm box}}([0,\epsilon])
 =\lim_{\epsilon\downarrow0}\liminf_{\ell_{\rm box}\to\infty}
   \rho_{2,\ell_{\rm box}}([0,\epsilon]).
 \label{eq:finite-volume-cumulative-atom}
\end{equation}
Local finiteness suffices for these statements even with infinite total UV weight. Multiplication by a fixed continuous positive function $w(s)$
with $w(0)=1$ preserves the criterion and the atomic weight.
\end{proposition}

\begin{proof}
The compact interval $[0,\epsilon]$ has boundary $\{\epsilon\}$ in the relative topology of $[0,\infty)$. The Portmanteau argument for local vague convergence therefore gives Eq.~\eqref{eq:finite-volume-cumulative-convergence}
at continuity radii. Continuity from above then yields the atom formula. Such radii can be chosen arbitrarily close to zero because a locally finite
measure has at most countably many atoms. Multiplication by a continuous
$w$ with $w(0)=1$ preserves local vague convergence and the atomic weight.
\end{proof}

Under these convergence assumptions a line with
$s_0(\ell_{\rm box})\to0$ and
$Z_0(\ell_{\rm box})\to Z_{\rm line}>0$ is sufficient for
$\rho_2(\{0\})\geq Z_{\rm line}$ although an atom may also be formed collectively.
The elementary example
\begin{equation}
 \dd\rho_m^{(D)}(s)
 =\frac1m\sum_{n=1}^{m}\delta_{n/m^2}(\dd s)
 \longrightarrow\delta_0(\dd s),
 \qquad \max_nZ_n^{(D)}(m)=m^{-1}\longrightarrow0
 \label{eq:benchmark-concentrating-atom}
\end{equation}
has unit limiting weight because its entire support lies in
$[0,1/m]$.
For $R_m(Q^2):=Q^2\int\dd\rho_m^{(D)}(s)/(s+Q^2)$ the two ordered
limits are
\begin{equation}
 \lim_{m\to\infty}\lim_{Q^2\downarrow0}R_m(Q^2)=0,
 \qquad
 \lim_{Q^2\downarrow0}\lim_{m\to\infty}R_m(Q^2)=1.
 \label{eq:benchmark-collective-ordered-limits}
\end{equation}
Thus the infinite-volume limit precedes the infrared limit when many
finite-volume levels collectively form an atom.
By contrast a nonatomic continuum has cumulative weight
$\rho_{2,c}([0,\epsilon])\downarrow0$ even though its largest individual
weights also vanish. The cumulative limit distinguishes these cases.

\subsection{The finite compact spectral branch}
\label{sec:app-compact-spectral-branch}

The response bandwidth $R=5$ in
Eq.~\eqref{eq:native-tt-coefficient} has nonempty support only when
$\widehat\lambda<5\pi$. For the three-generation assignment the scalar eigenvalues are $l(l+2)$ with even $l$. The one-form eigenvalues are $(l+1)^2$ with odd $l$. The spinor eigenvalues are $(n+3/2)^2$ with even right-handed and odd left-handed $n$.
Their respective multiplicities are $(l+1)^2$ and $2l(l+2)$ and $(n+1)(n+2)$.
Thus the finite sum reduces exactly to the eight terms in
Table~\ref{tab:compact-spectral-rows}.

\begin{table}[htbp]
\centering
\begin{tabular}{@{}lrrl@{}}
\toprule
Carrier and degree & $\widehat\lambda$ & $d_jw_j$ & $\widehat m_j^2$\\
\midrule
scalar $l=0$ & $0$ & $4$ & $3/4$\\
scalar $l=2$ & $8$ & $36$ & $3/4$\\
colour vectors $l=1$ & $4$ & $192$ & $3/2+\vartheta^2/6$\\
weak vectors $l=1$ & $4$ & $72$ & $1+\vartheta^2/6$\\
abelian vector $l=1$ & $4$ & $24$ & $\vartheta^2/6$\\
right-handed spinors $n=0$ & $9/4$ & $168$ & $3\vartheta^2/8$\\
right-handed spinors $n=2$ & $49/4$ & $1008$ & $3\vartheta^2/8$\\
left-handed spinors $n=1$ & $25/4$ & $576$ & $3\vartheta^2/8$\\
\bottomrule
\end{tabular}
\caption{Nonzero contributions to the compact spectral response with bandwidth $R=5$.
Eigenvalues and mass shifts are in the dimensionless endpoint units.
$d_jw_j$ combines modal multiplicity with matter content and channel weight.}
\label{tab:compact-spectral-rows}
\end{table}

The weights include four scalar components. The vector content consists of eight colour vectors and three weak vectors together with one abelian vector. The Weyl content consists of twenty-one right-handed and twenty-four left-handed components. The channel weights are those
of Section~\ref{sec:native-tt-spectral-branch}. Higher modes have empty frequency intervals at this bandwidth so the
eight rows exhaust the retained sum. Tensor-carrier and ghost terms
have zero weight in the prescription of
Section~\ref{sec:native-tt-spectral-branch}.

For these entries all denominators are positive near $\vartheta=1/50$
including the scalar constant mode because of its $3/4$ shift.
Consequently the endpoint formula and its source derivatives are smooth
on an open interval around that value.

\paragraph{Bandwidth comparison and numerical error.}
Let $A_R$ denote the sum over all modal indices accepted by the fixed
response bandwidth $R$. At $R=5$ the table exhausts this sum so
$A_{60,5}=A_5$ exactly. An interval quadrature enclosure
$|\widehat A_5-A_5|\leq\epsilon_A$ gives the endpoint error
\[
 |\widehat c_5-c_5|\leq\epsilon_A/\kappa_{\rm TT},\qquad
 c_5=A_5/\kappa_{\rm TT}.
\]
This is the numerical error at the declared response level.

The full Gaussian response provides a useful comparison of source
domains. For the sum without modal or frequency restrictions suppose the modal
weights have polynomial growth and $\widehat m_j^2\geq0$. The bound
$0\leq K_{\rm TT}\leq1$ proves absolute convergence. For a retained
set $\mathcal J_{n_{\rm sp}}$ write $A_\infty$ for the complete sum over the
same untruncated spectral towers and $A_{n_{\rm sp},R}$ for the finite sum.
Then
\begin{align}
0\leq A_\infty-A_{n_{\rm sp},R}\leq\frac1{64\pi^2}
\bigg[&
 \sum_{j\in\mathcal J_{n_{\rm sp}}}d_jw_j e^{-\widehat\lambda_j/\pi}
 \operatorname{erfc}\!\left(\frac{B_j}{\sqrt\pi}\right)\notag\\
&+\sum_{j\notin\mathcal J_{n_{\rm sp}}}d_jw_j e^{-\widehat\lambda_j/\pi}\bigg].
\label{eq:app-compact-gaussian-tail}
\end{align}
Indeed the omitted frequency integral is bounded by
$e^{-\widehat\lambda_j/\pi}\int_{B_j}^\infty e^{-u^2/\pi}\dd u$
and the integral is $\pi\operatorname{erfc}(B_j/\sqrt\pi)/2$.
For omitted modes use $B_j=0$. Positivity permits termwise summation.
By the radius equation of Section~\ref{sec:native-tt-spectral-branch}
the corresponding difference of endpoint coordinates is exactly
$(A_\infty-A_{n_{\rm sp},R})/\kappa_{\rm TT}$.
When the modal enumeration exhausts the domain $P_R$ this measures the
change in the endpoint on admitting the Gaussian spectrum outside that
domain. The finite-response-level construction in
ref.~\cite{Guo2026CompactSpectrumII} keeps $R$ fixed while controlling
the evaluation of its integrals.

\paragraph{Critical multiplicity and the complete slope.}
\label{sec:app-native-finite-rank-example}

The finite-dimensional critical space in
Theorem~\ref{thm:native-finite-rank-residue} can have several
source-visible directions. Consider
\begin{equation}
 L(z)=\begin{pmatrix}2z&0&z\\0&3z&2z\\z&2z&1+z\end{pmatrix},
 \qquad
 S=\begin{pmatrix}1&1\\0&1\\1&0\end{pmatrix},\qquad R=S^*.
 \label{eq:app-native-finite-rank-example}
\end{equation}
Here $P=\operatorname{diag}(1,1,0)$ with $G=\operatorname{diag}(2,3)$ on its range and $\det L=z^2(6-5z)$. Eliminating the last component gives
\begin{equation}
 L_{\rm eff}(z)=z\begin{pmatrix}2&0\\0&3\end{pmatrix}
 -\frac{z^2}{1+z}\begin{pmatrix}1&2\\2&4\end{pmatrix},\qquad
 \operatorname*{Res}_{z=0}S^*L(z)^{-1}S
 =\begin{pmatrix}1/2&1/2\\1/2&5/6\end{pmatrix}.
 \label{eq:app-native-finite-rank-answer}
\end{equation}
The residue has determinant $1/6$ so it has rank two despite a common
pole. Replacing the source columns by $(1,0,1)^{\mathsf T}$ and
$(2,0,-1)^{\mathsf T}$ gives the rank-one residue
$\bigl(\begin{smallmatrix}1/2&1\\1&2\end{smallmatrix}\bigr)$.
The alignment of the projected sources determines the rank.

\paragraph{Internal scale dependence.}
\label{sec:app-native-scale-identity}

Along $L(k_{\rm win})=\zeta/k_{\rm win}$ keep the dimensionless product
$\zeta:=k_{\rm win}L$ fixed. The selected compact branch has
$\zeta=\sqrt\pi$.  The TT response uses $p_{\rm int}^2=8/L^2$ and $m_{\rm curv}^2=6/L^2$ with exponential regulator $R_{k_{\rm win}}=p_{\rm int}^2/(e^{p_{\rm int}^2/k_{\rm win}^2}-1)$. Put
\begin{equation}
 u=\frac8{\zeta^2},\qquad v=\frac6{\zeta^2},\qquad
 d(\zeta)=u+v+\frac{u}{e^u-1}.
 \label{eq:app-native-scale-coefficient}
\end{equation}
Direct substitution gives
\begin{equation}
 G_{\rm TT}(k_{\rm win})=\frac{1}{k_{\rm win}^2d(\zeta)},\qquad
 p_{\rm int}^2G_{\rm TT}(k_{\rm win})=\frac{u}{d(\zeta)}.
 \label{eq:app-native-scale-identity}
\end{equation}
Along this family the radius and curvature change together with the internal momentum. Comparison with a response at varying external momentum
uses the corresponding state map together with source and observable transport between scales.

The mass-shift deformation including its fixed window and rescaled
vertex is evaluated in Eq.~\eqref{eq:native-chi-spectral-ladder}.

\subsection{Compact matter sources and tensor vertices}

\paragraph{Local scalar energy under geometric compression.}
\label{sec:app-local-scalar-energy}

Use $\Delta_g=\operatorname{div}_g\nabla$ in this subsection. In dimension
three the standard conformal transformation formulas give
\begin{align}
 \Delta_{g_\omega}f
 &=e^{-2\omega}(\Delta_gf+\nabla\omega\cdot\nabla f),\notag\\
 \mathcal R_{g_\omega}
 &=e^{-2\omega}(\mathcal R_g-4\Delta_g\omega
                                  -2|\nabla\omega|^2).
 \label{eq:app-scalar-conformal-change}
\end{align}
Applying the product rule to $\mathcal A_g(e^{\omega/2}f)$ yields
\[
 \mathcal A_{g_\omega}=e^{-5\omega/2}\mathcal A_ge^{\omega/2}.
\]
For compression $\omega=-\eta$ and the unitary volume transport is
$U_\eta=e^{-3\eta/2}$. Hence
\begin{equation}
 U_\eta\mathcal A_{g_\eta}U_\eta^{-1}=e^\eta\mathcal A_ge^\eta,
 \label{eq:native-local-conformal-congruence}
\end{equation}
including all curvature and derivative terms of the local source.

At the round reference $\mathcal A_g\geq3/(4L^2)$ on its form domain
$H^1(M)$. Multiplication by a smooth $e^\eta$ is an isomorphism of
this domain so the transported forms remain closed and satisfy
\[
 \langle u,e^\eta\mathcal A_ge^\eta u\rangle
 \geq\frac3{4L^2}e^{2\min\eta}\|u\|^2.
\]
The positive square root is therefore defined for every smooth real
$\eta$. Differentiation at zero gives
$\dot{\mathcal A}=\eta\mathcal A+\mathcal A\eta$.
For smooth reference eigenvectors differentiating the square root
may be written as the convergent resolvent pairing
\begin{equation}
 \langle u_i,\dot H u_j\rangle
 =\frac1\pi\int_0^\infty
 \frac{t^{1/2}\langle u_i,\dot{\mathcal A}u_j\rangle}
 {(M_i^2+t)(M_j^2+t)}\,\dd t
 =\frac{\langle u_i,\dot{\mathcal A}u_j\rangle}{M_i+M_j}.
 \label{eq:app-scalar-square-root-derivative}
\end{equation}
The common elliptic form domain and its positive lower bound justify
the derivative. Equivalently the matrix pairings solve
$H\dot H+\dot H H=\dot{\mathcal A}$. Substituting
$\dot{\mathcal A}_{ij}=(M_i^2+M_j^2)\eta_{ij}$ gives
\begin{equation}
 (DH_g[\eta])_{ij}
 =\frac{M_i^2+M_j^2}{M_i+M_j}\langle u_i,\eta u_j\rangle.
 \label{eq:native-local-square-root-source}
\end{equation}
In particular it remains
regular for equal eigenvalues. Finite occupation support makes the energy trace a finite sum of these pairings.

For distinct energies the off-diagonal coefficient is
\begin{equation}
 \frac{M_i^2+M_j^2}{M_i+M_j}
 =\frac{M_i+M_j}{2}
       +\frac{(M_i-M_j)^2}{2(M_i+M_j)}.
 \label{eq:app-coherent-energy-source}
\end{equation}
Thus the local positive density in
Eq.~\eqref{eq:native-local-scalar-charge} uses stationarity of the
occupation matrix. Coherent transitions are still determined by the
same source but need not have that diagonal density representation.
For the constant mode the density is
$\sqrt3/(2LV_3)>0$. It has a positive local compression charge and zero TT vertex.

For a scalar section of a Hermitian vector bundle with unitary
connection $B$ write $D_B=\dd+B$ and
\begin{equation}
 \mathcal A_{g,B}=D_B^*D_B+\mathcal R_g/8.
 \label{eq:app-covariant-scalar-energy}
\end{equation}
The connection is held fixed as a one-form under the metric source.
Since $D_B(fu)=\dd f\,u+fD_Bu$ the calculation in
Eq.~\eqref{eq:app-scalar-conformal-change} applies with covariant
derivatives and gives
$U_\eta\mathcal A_{g_\eta,B}U_\eta^{-1}
=e^\eta\mathcal A_{g,B}e^\eta$. At the round reference
$D_B^*D_B\geq0$ gives the same strict lower bound. All linear and
quadratic connection terms are retained in this identity. The local
energy density is consequently $\sum_i n_iM_i\|u_i(x)\|^2$ for any
such bundle with this scalar operator. The eigenvalues $M_i$ already
include the connection dependence.

\paragraph{Compact TT tensors and scalar matrix elements.}
\label{sec:app-compact-tt-visibility}

Let $E_a$ be the left-invariant orthonormal frame with
$[E_a,E_b]=(2/L)\epsilon_{abc}E_c$ and coframe $\theta^a$.
The round connection is $\nabla_{E_a}E_b=L^{-1}\epsilon_{abc}E_c$.
For constant $q=q^{\mathsf T}$ with $\operatorname{tr}q=0$ set
$h_q^+=V_3^{-1/2}q_{ab}\theta^a\theta^b$. Its trace is zero and
\[
 \nabla^a(h_q^+)_{ab}
 =-\frac1{L\sqrt{V_3}}\sum_{a,c}\epsilon_{abc}q_{ac}=0.
\]
Write $(J_a)_{bc}=\epsilon_{bac}$ for the three infinitesimal rotation
matrices. On these constant tensors the rough Laplacian is
$-L^{-2}\sum_a[J_a,[J_a,q]]$. Direct multiplication gives
\begin{equation}
 -\sum_a[J_a,[J_a,q]]=6q-2\operatorname{tr}(q)I_3=6q,
 \qquad \Delta_{TT}h_q^+=\frac6{L^2}h_q^+.
 \label{eq:app-homogeneous-tt-eigenvalue}
\end{equation}
The tensors descend through the central quotient. The right-invariant frame gives $h_q^-$. Their
inner products within each family are $\operatorname{tr}(qr)$.
Between families the Haar average is proportional to
$\operatorname{tr}(q)\operatorname{tr}(r)$ and vanishes. Choosing a
Frobenius-orthonormal basis of trace-free symmetric matrices therefore
gives the ten-dimensional internal TT space used in
Section~\ref{sec:native-compact-tt-visibility}.

For $R(U)=\operatorname{Ad}(U)$ normalized Haar integration gives
\begin{equation}
 \int_M R_{ib}R_{je}\,\dd v_g
 =\frac{V_3}{3}\delta_{ij}\delta_{be},\qquad
 E_aR=\frac2L R J_a.
 \label{eq:app-adjoint-modes}
\end{equation}
The Haar identity follows by invariance and the trace. The derivative
follows from right multiplication by $\exp(2tJ_a/L)$. Thus
$u_{ib}=(3/V_3)^{1/2}R_{ib}$ is the nine-dimensional scalar shell with
eigenvalue $8/L^2$. Substitution in
Eq.~\eqref{eq:native-geometric-vertex} uses
\[
 \sum_{a,d,c}q_{ad}(J_a)_{cb}(J_d)_{ce}
 =\operatorname{tr}(q)\delta_{be}-q_{be}=-q_{be}
\]
and proves $H[h_q^+]=4I_3\otimes q/(L^2\sqrt{V_3})$.
Left multiplication proves the second vertex with the factors
interchanged. On a round Einstein metric the scalar-curvature
variation is
$D\mathcal R[h]=-\langle\operatorname{Ric},h\rangle
+\nabla^a\nabla^b h_{ab}-\Delta\operatorname{tr}h$
which vanishes on both TT families. Hence the same vertices apply to
$\mathcal A_g$ including its curvature shift.
The square-root derivative on this equal-energy shell divides the
vertex by $2M_2$ proving Eq.~\eqref{eq:native-compact-tt-vertices}.

For a real occupation matrix $\mathsf N_{\rm occ}$ the energy is
$M_2\operatorname{Tr}\mathsf N_{\rm occ}$. Its source derivative is the trace
of $\mathsf N_{\rm occ}$ with the energy vertex so
partial traces give Eq.~\eqref{eq:native-compact-tt-source}. Arbitrary
small perturbations of $I_9/9$ by $q\otimes I_3$ and $I_3\otimes r$
remain positive and vary the two trace-free partial traces
independently. The map from occupations to compact tensor sources
therefore has rank ten.

For real unit vectors $a,b$ a pure product occupation has source
$\alpha_2((aa^{\mathsf T})_0,(bb^{\mathsf T})_0)$.
Two such occupations obey
\begin{equation}
 \langle j_{\rm hom}(a,b),j_{\rm hom}(a',b')\rangle
 =\alpha_2^2\bigl[(a\cdot a')^2+(b\cdot b')^2-2/3\bigr].
 \label{eq:app-compact-mixed-source-sign}
\end{equation}
The contraction can have either sign for states of the same positive
energy while an isotropic occupation has zero tensor source. In
Section~\ref{sec:native-static-exchange} these vertices enter the
complete source together with the compression density and constraint
terms. The reduced response determines the spatial dependence.

\paragraph{The tensor source of the $(2,1)$ one-form band.}
\label{sec:app-native-long-root-source}

Choose the right coframe conventions $\dd\theta_R^a=L^{-1}\epsilon_{abc}\theta_R^b\wedge\theta_R^c$ and $E_aR_{bi}=-(2/L)\epsilon_{abc}R_{ci}$.
For $a=(qRv)_b\theta_R^b$ symmetry and tracelessness of $q$ give
\begin{equation}
 \delta a=0,\qquad *\dd a=\frac4L a,\qquad
 \Delta_1 a=\frac{16}{L^2}a .
 \label{eq:app-long-root-curl}
\end{equation}
Indeed the derivative part of the curl is $2qRv/L$ and the coframe
part contributes the same amount. Left translation acts on $q$ by
orthogonal conjugation and right translation acts on $v$. The span is
therefore $(2,1)$. Haar orthogonality gives
$\int Rv(Rw)^T\dd v_g=(V_3/3)(v\cdot w)I_3$ so the factor
$(3/V_3)^{1/2}$ in the text makes this identification isometric.

To differentiate the Hodge operator retain a fixed differential
one-form while varying the internal comparison metric by $h$.
At the round point put $B=*\dd a=4a/L$ and write
$N_g(a)=\|a\|_g^2$ and
$Q_g(a)=\|\dd a\|_g^2+\|\delta_ga\|_g^2$.
For trace-free $h$ inverse-metric differentiation gives
\begin{equation}
 \dot N_g(a)=-\int h^{ab}a_a a_b\,\dd v_g,\qquad
 \dot Q_g(a)=\int h^{ab}B_aB_b\,\dd v_g .
 \label{eq:app-long-root-form-source}
\end{equation}
The codifferential term has zero first derivative because $\delta a=0$.
In dimension three the two-form norm gives the plus sign in the
second identity. The source compressed to the degenerate eigenband
is the polarization of $\dot Q_g-(16/L^2)\dot N_g$.
Consequently for $\|q\|_{\rm HS}=|v|=1$
\begin{equation}
 \langle a_{q,v},H_{\rm lr}[h_t^-]a_{q,v}\rangle
 =\frac{32}{L^2\sqrt{V_3}}\operatorname{tr}(tq^2).
 \label{eq:app-long-root-pairing}
\end{equation}
Polarization and the same Haar identity prove
Eq.~\eqref{eq:native-long-root-tt-source} on the full fifteen-dimensional
band. This is the derivative after unitary comparison of the varying
one-form spaces. Within a degenerate band the transport commutator
vanishes. The source is right-translation invariant and hence does not
mix this band with the opposite $(1,2)$ band at the same eigenvalue.
Its eigenvalues give the first-order splittings in the $(2,1)$ band.

For the unshifted energy operator $\sqrt{\Delta_1}$ division by
$2\sqrt{16/L^2}$ gives the tensor source
$j_{\rm lr}(q,v)=4(q^2)_0/(L\sqrt{V_3})$ on these unit decomposable
states. A real trace-free $3\times3$ matrix satisfies
$\operatorname{tr}q^4=(\operatorname{tr}q^2)^2/2$ whence
\begin{equation}
 \|(q^2)_0\|_{\rm HS}^2=\frac16,\qquad
 \|j_{\rm lr}(q,v)\|_{\rm HS}=\frac4{\sqrt6L\sqrt{V_3}}>0.
 \label{eq:app-long-root-visible-source}
\end{equation}
The trace $\operatorname{Tr}_{\mathcal Q}\mathcal A_t$ is an invariant linear functional of a trace-free matrix and therefore vanishes.
Equal occupation of the complete band therefore has zero TT source.
This gives the bilinear coupling of the one-form spectrum to the
internal tensor source. A prescribed mass shift contributes its own
derivative. The one-form eigenvalue $16/L^2$ and the rough tensor
eigenvalue $6/L^2$ belong to different operators. The scale-dependent
response is given by Eq.~\eqref{eq:app-native-scale-identity}.

The Coarse-Graining Genesis Framework~\cite{Guo2026CoarseGrainingFramework} supplies the finite matrix and Haar-integral evaluations together with spectral quadrature and source variations.

\subsection{Continuum and truncation limits}
\label{sec:lattice-test}
\label{sec:frg-kernel-test}

For a reflection-positive lattice channel let
$C_2(Q^2;a_{\rm lat},\ell_{\rm box})$ be the Ward-normalised TT
coefficient and $P_2$ its local subtraction. The continuum and volume
limits are taken at fixed nonzero momentum in the window
\begin{equation}
 a_{\rm lat}|Q|\ll1\ll\ell_{\rm box}|Q|,
 \label{eq:lattice-scaling-window}
\end{equation}
before sending $Q^2\downarrow0$ in
\begin{equation}
 R_2(Q^2;a_{\rm lat},\ell_{\rm box})
 :=Q^2[C_2(Q^2;a_{\rm lat},\ell_{\rm box})
       -P_2(Q^2;a_{\rm lat},\ell_{\rm box})].
 \label{eq:lattice-residue-estimator}
\end{equation}
Proposition~\ref{prop:tt-atom-extraction} identifies this ordered limit
with the atom while Proposition~\ref{prop:finite-volume-cumulative-atom}
gives its cumulative-weight form.

For functional approximations the common-space operator and source
bounds of Appendix~\ref{app:fredholm} control residues at their displaced
roots. Endpoint convergence and
Theorem~\ref{thm:tt-fredholm-spectral-bridge} then identify the physical
atomic weight. The strict lower bound in
Proposition~\ref{prop:compact-spectrum-interface} gives its positivity.
Scalar two- and four-point spectral functions have been computed using
source flows~\cite{Horak2023spectral}. For the complete covariant stress
tensor of Theorem~\ref{thm:null-shell-exclusion} the corresponding
atomic weight is zero.


\section*{Acknowledgments}
\label{sec:acknowledgments}

All numerical and symbolic calculations were carried out using standard
mathematical software. No original code is released. The author is grateful
to the reviewers for their constructive comments on this work.

AI-assisted tools were used for language editing and for drafting and
checking the numerical verification and formalization code. They also
assisted selected numerical consistency checks. The author is responsible
for the scientific content and conclusions.


\section*{Declarations}
\label{sec:declarations}

Conflict of interest. The author declares no competing interests.

Data and code availability. The synthetic spectral sequences and finite
response models are defined in Appendix~\ref{app:benchmarks}.
The finite-dimensional source responses and constrained systems together with the critical residues use the Coarse-Graining Genesis
Framework~\cite{Guo2026CoarseGrainingFramework}.
No external numerical dataset is used.

Funding. No funding was received for this work.

\bibliographystyle{JHEP}
\bibliography{references}

@Article{Bordag2001Casimir,
  author = {M. Bordag and U. Mohideen and V. M. Mostepanenko},
  title = {New Developments in the Casimir Effect},
  journal = {Phys. Rept.},
  volume = {353},
  pages = {1--205},
  year = {2001},
  doi = {10.1016/S0370-1573(01)00015-1},
  eprint = {quant-ph/0106045},
  archivePrefix = {arXiv},
}

@misc{Guo2026Arithmetic,
  author = {Jinku Guo},
  title = {A Spectral Construction of Natural-Number Arithmetic},
  year = {2026},
  howpublished = {SSRN preprint 7421398},
  doi = {10.2139/ssrn.7421398},
  url = {https://ssrn.com/abstract=7421398},
}

@Book{Sato2013,
  author = {K.-I. Sato},
  title = {L\'evy Processes and Infinitely Divisible Distributions},
  publisher = {Cambridge University Press},
  address = {Cambridge},
  series = {Cambridge Studies in Advanced Mathematics},
  volume = {68},
  year = {2013},
  note = {Revised edition},
  isbn = {9781107656499},
}

@Article{Feshbach1958,
  author = {H. Feshbach},
  title = {Unified Theory of Nuclear Reactions},
  journal = {Ann. Phys.},
  volume = {5},
  pages = {357--390},
  year = {1958},
  doi = {10.1016/0003-4916(58)90007-1},
}

@misc{Guo2026GeneratedDynamics,
  author = {Jinku Guo},
  title = {Dynamics of Generated Spectra from Minimal Change: Certified Function Classes and Finite Windows},
  year = {2026},
  howpublished = {SSRN preprint 7649818},
  doi = {10.2139/ssrn.7649818},
  url = {https://ssrn.com/abstract=7649818},
}

@Article{BrownKuchar1995,
  author = {J. D. Brown and K. V. Kucha\v{r}},
  title = {Dust as a standard of space and time in canonical quantum gravity},
  journal = {Phys. Rev. D},
  volume = {51},
  pages = {5600--5629},
  year = {1995},
  doi = {10.1103/PhysRevD.51.5600},
  eprint = {gr-qc/9409001},
  archivePrefix = {arXiv},
}

@misc{Guo2026GeneralChannels,
  author        = {J. Guo},
  title         = {Coarse-Graining and Long-Range Response in Generated Systems},
  year          = {2026},
  eprint        = {2608.01989},
  archivePrefix = {arXiv},
  doi           = {10.48550/arXiv.2608.01989},
  howpublished = {arXiv preprint},
  url = {https://arxiv.org/abs/2608.01989},
}

@misc{Guo2026CompactSpectrumII,
  author = {J. Guo},
  title  = {The Spectrum of a Compact Internal Space {II}: Effective Couplings and Mass Scales},
  howpublished = {SSRN preprint 7351999},
  year   = {2026},
  doi    = {10.2139/ssrn.7351999},
  url = {https://ssrn.com/abstract=7351999},
}

@misc{Guo2026CompactSpectrumI,
  author = {J. Guo},
  title  = {The Spectrum of a Compact Internal Space {I}: Gauge Structure and a Scalar Mass Gap},
  howpublished = {SSRN preprint 7352020},
  year   = {2026},
  doi    = {10.2139/ssrn.7352020},
  url = {https://ssrn.com/abstract=7352020},
}

@Article{Kallen1952,
  author  = {G. K\"all\'en},
  title   = {On the Definition of the Renormalization Constants in Quantum Electrodynamics},
  journal = {Helv. Phys. Acta},
  year    = {1952},
  volume  = {25},
  pages   = {417--434},
  doi     = {10.5169/seals-112316},
}

@Book{Streater1964,
  author    = {R. F. Streater and A. S. Wightman},
  title     = {{PCT}, Spin and Statistics, and All That},
  publisher = {Benjamin},
  year      = {1964},
  address   = {New York},
}

@Book{Weinberg1995,
  author    = {S. Weinberg},
  title     = {The Quantum Theory of Fields, Vol.~1},
  publisher = {Cambridge University Press},
  year      = {1995},
  isbn      = {0-521-55001-7},
  address   = {Cambridge},
  doi       = {10.1017/CBO9781139644167},
}

@Article{Weinberg2020Massless,
  author        = {S. Weinberg},
  title         = {Massless Particles in Higher Dimensions},
  journal       = {Phys. Rev. D},
  volume        = {102},
  pages         = {095022},
  year          = {2020},
  doi           = {10.1103/PhysRevD.102.095022},
  eprint        = {2010.05823},
  archivePrefix = {arXiv},
  primaryClass  = {hep-th},
}

@Article{Fehre2023prl,
  author        = {J. Fehre and D. F. Litim and J. M. Pawlowski and M. Reichert},
  title         = {Lorentzian Quantum Gravity and the Graviton Spectral Function},
  journal       = {Phys. Rev. Lett.},
  year          = {2023},
  volume        = {130},
  pages         = {081501},
  doi           = {10.1103/PhysRevLett.130.081501},
  eprint        = {2111.13232},
  archivePrefix = {arXiv},
  primaryClass  = {hep-th},
}

@Article{Sakharov1967,
  author  = {A. D. Sakharov},
  title   = {Vacuum Quantum Fluctuations in Curved Space and the Theory of Gravitation},
  journal = {Sov. Phys. Dokl.},
  year    = {1968},
  volume  = {12},
  pages   = {1040--1041},
  note    = {[Dokl. Akad. Nauk SSSR 177, 70--71 (1967)]},
  url     = {https://www.ufn.ru/ufn91/ufn91_5/ufn915i.pdf},
}

@Article{Weinberg1965,
  author  = {S. Weinberg},
  title   = {Photons and Gravitons in {$S$}-Matrix Theory: Derivation of Charge Conservation and Equality of Gravitational and Inertial Mass},
  journal = {Phys. Rev.},
  year    = {1964},
  doi     = {10.1103/PhysRev.135.B1049},
  volume  = {135},
  pages   = {B1049--B1056},
}

@Article{WeinbergWitten1980,
  author  = {S. Weinberg and E. Witten},
  title   = {Limits on Massless Particles},
  journal = {Phys. Lett. B},
  year    = {1980},
  doi     = {10.1016/0370-2693(80)90212-9},
  volume  = {96},
  pages   = {59--62},
}

@Article{OsbornPetkou1994,
  author  = {H. Osborn and A. Petkou},
  title   = {Implications of Conformal Invariance in Field Theories for General Dimensions},
  journal = {Ann. Phys.},
  year    = {1994},
  doi     = {10.1006/aphy.1994.1045},
  volume  = {231},
  pages   = {311--362},
}

@Article{Visser2002,
  author  = {M. Visser},
  title   = {{Sakharov}'s Induced Gravity: A Modern Perspective},
  journal = {Mod. Phys. Lett. A},
  year    = {2002},
  doi          = {10.1142/S0217732302006886},
  volume  = {17},
  pages   = {977--991},
}

@Article{Amati1981,
  author  = {D. Amati and G. Veneziano},
  title   = {Metric from Matter},
  journal = {Phys. Lett. B},
  year    = {1981},
  doi          = {10.1016/0370-2693(81)90779-6},
  volume  = {105},
  pages   = {358--362},
}

@Article{Amati1982,
  author  = {D. Amati and G. Veneziano},
  title   = {A Unified Gauge and Gravity Theory with Only Matter Fields},
  journal = {Nucl. Phys. B},
  year    = {1982},
  doi          = {10.1016/0550-3213(82)90201-2},
  volume  = {204},
  pages   = {451--476},
}

@Article{Maitiniyazi2026emergence,
  author        = {Y. Maitiniyazi and M. Yamada},
  title         = {Emergence of Dynamical Tensor Fields in Composite Models of Gravity},
  journal       = {Phys. Rev. D},
  year          = {2026},
  volume        = {113},
  pages         = {084057},
  doi           = {10.1103/1y2m-xxt7},
  eprint        = {2512.08435},
  archivePrefix = {arXiv},
  primaryClass  = {hep-th},
}

@Article{Horak2023spectral,
  author        = {J. Horak and F. Ihssen and J. M. Pawlowski and J. Wessely and N. Wink},
  title         = {Scalar spectral functions from the spectral functional renormalization group},
  journal       = {Phys. Rev. D},
  volume        = {110},
  pages         = {056009},
  year          = {2024},
  doi           = {10.1103/PhysRevD.110.056009},
  eprint        = {2303.16719},
  archivePrefix = {arXiv},
  primaryClass  = {hep-th},
}

@Article{Jenkins2009,
  author  = {A. Jenkins},
  title   = {Constraints on Emergent Gravity},
  journal = {Int. J. Mod. Phys. D},
  year    = {2009},
  doi     = {10.1142/S0218271809015941},
  volume  = {18},
  pages   = {2249--2255},
}

@Article{Salpeter1951,
  author  = {E. E. Salpeter and H. A. Bethe},
  title   = {A Relativistic Equation for Bound-State Problems},
  journal = {Phys. Rev.},
  year    = {1951},
  doi     = {10.1103/PhysRev.84.1232},
  volume  = {84},
  pages   = {1232--1242},
}

@Article{Lehmann1954,
  author  = {H. Lehmann},
  title   = {{\"U}ber {Eigenschaften} von {Ausbreitungsfunktionen} und {Renormierungskonstanten} quantisierter {Felder}},
  journal = {Nuovo Cim.},
  volume  = {11},
  year    = {1954},
  pages   = {342--357},
  doi     = {10.1007/BF02783624},
}

@Article{OsterwalderSchrader1973,
  author  = {K. Osterwalder and R. Schrader},
  title   = {Axioms for {Euclidean} {Green}'s Functions},
  journal = {Commun. Math. Phys.},
  volume  = {31},
  year    = {1973},
  pages   = {83--112},
  doi     = {10.1007/BF01645738},
}

@Article{Wetterich1993,
  author  = {C. Wetterich},
  title   = {Exact Evolution Equation for the Effective Potential},
  journal = {Phys. Lett. B},
  volume  = {301},
  year    = {1993},
  pages   = {90--94},
  doi     = {10.1016/0370-2693(93)90726-X},
}

@Article{FloerchingerWetterich2009,
  author        = {S. Floerchinger and C. Wetterich},
  title         = {Exact Flow Equation for Composite Operators},
  journal       = {Phys. Lett. B},
  volume        = {680},
  year          = {2009},
  pages         = {371--376},
  doi           = {10.1016/j.physletb.2009.09.014},
  eprint        = {0905.0915},
  archivePrefix = {arXiv},
  primaryClass  = {hep-th},
}

@Article{CornwallJackiwTomboulis1974,
  author  = {J. M. Cornwall and R. Jackiw and E. Tomboulis},
  title   = {Effective Action for Composite Operators},
  journal = {Phys. Rev. D},
  volume  = {10},
  year    = {1974},
  pages   = {2428--2445},
  doi     = {10.1103/PhysRevD.10.2428},
}

@Article{CollinsScalise1994,
  author        = {J. C. Collins and R. J. Scalise},
  title         = {Renormalization of Composite Operators in {Yang--Mills} Theories Using a General Covariant Gauge},
  journal       = {Phys. Rev. D},
  volume        = {50},
  year          = {1994},
  pages         = {4117--4136},
  doi           = {10.1103/PhysRevD.50.4117},
  eprint        = {hep-ph/9403231},
  archivePrefix = {arXiv},
}

@Article{JoglekarLee1976,
  author  = {S. D. Joglekar and B. W. Lee},
  title   = {General Theory of Renormalization of Gauge Invariant Operators},
  journal = {Ann. Phys.},
  volume  = {97},
  year    = {1976},
  pages   = {160--215},
  doi     = {10.1016/0003-4916(76)90225-6},
}

@Article{Barnich2000,
  author        = {G. Barnich and F. Brandt and M. Henneaux},
  title         = {Local {BRST} Cohomology in Gauge Theories},
  journal       = {Phys. Rept.},
  volume        = {338},
  year          = {2000},
  pages         = {439--569},
  doi           = {10.1016/S0370-1573(00)00049-1},
  eprint        = {hep-th/0002245},
  archivePrefix = {arXiv},
}

@Book{Kato1995,
  author    = {T. Kato},
  title     = {Perturbation Theory for Linear Operators},
  publisher = {Springer},
  address   = {Berlin},
  year      = {1995},
  series    = {Classics in Mathematics},
  edition   = {2},
  doi       = {10.1007/978-3-642-66282-9},
}

@Book{Folland2016Harmonic,
  author    = {Gerald B. Folland},
  title     = {A Course in Abstract Harmonic Analysis},
  edition   = {2},
  publisher = {CRC Press},
  address   = {Boca Raton, FL},
  year      = {2016},
  doi       = {10.1201/b19172},
}

@Book{Hormander1990LPDO1,
  author    = {Lars H\"ormander},
  title     = {The Analysis of Linear Partial Differential Operators I:
               Distribution Theory and Fourier Analysis},
  edition   = {2},
  publisher = {Springer},
  address   = {Berlin},
  year      = {1990},
  doi       = {10.1007/978-3-642-61497-2},
}

@Book{Chatelin1983,
  author    = {F. Chatelin},
  title     = {Spectral Approximation of Linear Operators},
  publisher = {Academic Press},
  address   = {New York},
  year      = {1983},
}

@Book{TrefethenEmbree2005,
  author    = {L. N. Trefethen and M. Embree},
  title     = {Spectra and Pseudospectra: The Behavior of Nonnormal Matrices and Operators},
  publisher = {Princeton University Press},
  address   = {Princeton},
  year      = {2005},
  doi       = {10.1515/9780691213101},
}

@Article{Deser1970,
  author  = {S. Deser},
  title   = {Self-Interaction and Gauge Invariance},
  journal = {Gen. Rel. Grav.},
  volume  = {1},
  year    = {1970},
  pages   = {9--18},
  doi     = {10.1007/BF00759198},
}

@Article{ReggeTeitelboim1974,
  author  = {T. Regge and C. Teitelboim},
  title   = {Role of Surface Integrals in the Hamiltonian Formulation of General Relativity},
  journal = {Annals Phys.},
  volume  = {88},
  year    = {1974},
  pages   = {286--318},
  doi     = {10.1016/0003-4916(74)90404-7},
}

@Article{Boulanger2000,
  author        = {N. Boulanger and T. Damour and L. Gualtieri and M. Henneaux},
  title         = {Inconsistency of Interacting, Multi-Graviton Theories},
  journal       = {Nucl. Phys. B},
  volume        = {597},
  year          = {2001},
  pages         = {127--171},
  doi           = {10.1016/S0550-3213(00)00718-5},
  eprint        = {hep-th/0007220},
  archivePrefix = {arXiv},
}

@Article{Porrati2008,
  author        = {M. Porrati},
  title         = {Universal Limits on Massless High-Spin Particles},
  journal       = {Phys. Rev. D},
  volume        = {78},
  year          = {2008},
  pages         = {065016},
  doi           = {10.1103/PhysRevD.78.065016},
  eprint        = {0804.4672},
  archivePrefix = {arXiv},
  primaryClass  = {hep-th},
}

@Article{Giulini2018Laue,
  author        = {Domenico Giulini},
  title         = {{Laue}'s Theorem Revisited: Energy-Momentum Tensors,
                   Symmetries, and the Habitat of Globally Conserved Quantities},
  journal       = {Int. J. Geom. Meth. Mod. Phys.},
  volume        = {15},
  year          = {2018},
  pages         = {1850182},
  doi           = {10.1142/S0219887818501827},
  eprint        = {1808.09320},
  archivePrefix = {arXiv},
  primaryClass  = {math-ph},
}

@Article{Porrati2012,
  author        = {M. Porrati},
  title         = {Old and New No-Go Theorems on Interacting Massless Particles in Flat Space},
  year          = {2012},
  eprint        = {1209.4876},
  archivePrefix = {arXiv},
  primaryClass  = {hep-th},
}

@Misc{Guo2026CoarseGrainingFramework,
  author = {Jinku Guo},
  title = {{Coarse-Graining Genesis Framework}},
  year = {2026},
  howpublished = {Computer software, Zenodo},
  doi = {10.5281/zenodo.22067005},
  url = {https://doi.org/10.5281/zenodo.22067005},
}

\end{document}